\documentclass[a4paper,onecolumn,11pt,accepted=2023-10-06]{quantumarticle}
\pdfoutput=1
\usepackage[utf8]{inputenc}
\usepackage[english]{babel}
\usepackage[T1]{fontenc}

\usepackage{cite}
\usepackage{amsmath,amsthm,amsfonts,amssymb,xspace,graphicx,float,mathtools}
\usepackage{braket,xcolor,textcomp,pifont}
\usepackage{enumitem} 
\usepackage{hyperref}
\usepackage{tcolorbox}
\usepackage{multirow}
\usepackage{booktabs}
\usepackage{tikz}
\usetikzlibrary{quantikz}

\newtheorem{theorem}{Theorem}[section]
\newtheorem{lemma}[theorem]{Lemma}

\newtheorem{corollary}[theorem]{Corollary}

\theoremstyle{definition}
\newtheorem{definition}[theorem]{Definition}
\newtheorem{remark}[theorem]{Remark}

\newtheorem{example}[theorem]{Example}

\newcommand{\Z}{\mathbbm Z}
\newcommand{\Co}{\mathbbm C}
\newcommand{\N}{\mathbbm N}
\newcommand{\bbone}{\mathbbm 1}

\newcommand{\poly}{\operatorname{poly}}
\newcommand{\polylog}{\operatorname{polylog}}
\newcommand{\E}{\mathop{\bf E\/}}
\newcommand{\BQP}{\mathsf{BQP}}
\newcommand{\PSPACE}{\mathsf{PSPACE}}
\newcommand{\NC}{\mathsf{NC}}
\newcommand{\QNC}{\mathsf{QNC}}

\DeclarePairedDelimiter\parens{\lparen}{\rparen}
\DeclarePairedDelimiter\abs{\lvert}{\rvert}
\DeclarePairedDelimiter\norm{\lVert}{\rVert}
\DeclarePairedDelimiter\floor{\lfloor}{\rfloor}
\DeclarePairedDelimiter\ceil{\lceil}{\rceil}
\DeclarePairedDelimiter\braces{\lbrace}{\rbrace}
\DeclarePairedDelimiter\bracks{\lbrack}{\rbrack}
\DeclarePairedDelimiter\bbracks{\lbrack\!\lbrack}{\rbrack\!\rbrack}
\DeclarePairedDelimiter\angles{\langle}{\rangle}

\newcommand{\calB}{\mathcal{B}}

\newcommand{\calH}{\mathcal{H}}

\newcommand{\calJ}{\mathcal{J}}

\newcommand{\calN}{\mathcal{N}}
\newcommand{\calO}{\mathcal{O}}
\newcommand{\calP}{\mathcal{P}}

\newcommand{\calT}{\mathcal{T}}

\newcommand{\calX}{\mathcal{X}}
\newcommand{\calY}{\mathcal{Y}}

\newcommand{\ba}{\boldsymbol{a}}

\newcommand{\bj}{\boldsymbol{j}}
\newcommand{\bk}{\boldsymbol{k}}

\newcommand{\bt}{\boldsymbol{t}}

\newcommand{\bw}{\boldsymbol{w}}
\newcommand{\bx}{{\boldsymbol{x}}}

\newcommand{\bz}{\boldsymbol{z}}

\begin{document}

\title{Parallel Quantum Algorithm for Hamiltonian Simulation}

\author{Zhicheng Zhang}
\affiliation{Centre for Quantum Software and Information, University of Technology Sydney, Sydney, Australia}
\affiliation{University of Chinese Academy of Sciences, Beijing, China}
\thanks{part of the work was done when the author was at the University of Chinese Academy of Sciences, Beijing, China.}
\orcid{0000-0002-7436-0426}
\email{iszczhang@gmail.com}

\author{Qisheng Wang}
\affiliation{Graduate School of Mathematics, Nagoya University, Nagoya, Japan}
\affiliation{Department of Computer Science and Technology, Tsinghua University, Beijing, China}
\thanks{part of the work was done when the author was at the Department of Computer Science and Technology, Tsinghua University, Beijing, China.}
\orcid{0000-0001-5107-8279}
\email{qishengwang1994@gmail.com}

\author{Mingsheng Ying}
\affiliation{State Key Laboratory of Computer Science, Institute of Software, Chinese Academy of Sciences, Beijing, China}
\affiliation{Department of Computer Science and Technology, Tsinghua University, Beijing, China}
\orcid{0000-0003-4847-702X}
\email{yingms@ios.ac.cn}
\email{yingmsh@tsinghua.edu.cn}

\maketitle

\begin{abstract}
    We study how parallelism can speed up quantum simulation.
    A parallel quantum algorithm is proposed for simulating the dynamics of a large class of Hamiltonians with good sparse structures, 
    called uniform-structured Hamiltonians, 
    including various Hamiltonians of practical interest like local Hamiltonians and Pauli sums.
    Given the oracle access to the target sparse Hamiltonian,
    in both query and gate complexity,
    the running time of our parallel quantum simulation algorithm 
    measured by the quantum circuit depth has a
    doubly (poly-)logarithmic dependence $\operatorname{polylog}\log(1/\epsilon)$ on the simulation precision $\epsilon$. This presents 
    an \textit{exponential improvement}  over the dependence $\operatorname{polylog}(1/\epsilon)$ of previous optimal sparse Hamiltonian simulation algorithm without parallelism.
    To obtain this result, we introduce a novel notion of parallel quantum walk, based on Childs' quantum walk.
    The target evolution unitary is approximated by a truncated Taylor series, which is obtained by combining these 
    quantum walks in a parallel way.
    A lower bound $\Omega(\log \log (1/\epsilon))$ is established, showing that the $\epsilon$-dependence of the gate depth achieved in this work cannot be significantly improved.

    Our algorithm is applied to simulating three physical models: the Heisenberg model, the Sachdev-Ye-Kitaev model and a quantum chemistry model in second quantization.
    By explicitly calculating the gate complexity for implementing the oracles,
    we show that on all these models, the total gate depth of our algorithm has a $\operatorname{polylog}\log(1/\epsilon)$ dependence in the parallel setting.
\end{abstract}

\newpage
\tableofcontents
\newpage

\section{Introduction}
\label{sec:introduction}

Simulating the quantum Hamiltonian dynamics is a fundamental problem in computational physics.
Despite its ubiquity and importance, the problem is believed to be  intractable for classical computers.  
Quantum computers were originally proposed to efficiently solve this problem~\cite{Feynman82}. The first algorithm for solving this problem was given by Lloyd for  local Hamiltonians~\cite{Lloyd96}, and has been 
followed by many remarkable results over the past twenty years. 
Moreover, these results have found diverse applications in other quantum algorithms (e.g.,~\cite{CKS17,AGGW17,FGG14,CGJ19}) beyond quantum simulation.

While the state-of-the-art has achieved an optimal quantum algorithmic solution to simulating a large class of Hamiltonians~\cite{LC19},
it remains open whether the quantum simulation algorithms can be parallelized in order to provide further speed-up.
In this work, we identify a class of Hamiltonians that can be more efficiently simulated in parallel, called uniform-structured Hamiltonians. 
Then we introduce the notion of parallel quantum walk within Childs' framework~\cite{Childs09,BC12,BCK15,CKS17}.  
Based on it, we propose a parallel quantum simulation algorithm for uniform-structured Hamiltonians.

\subsubsection*{Hamiltonian simulation}

Simulating the time evolution of a quantum system governed by a time-independent Hamiltonian $H$ for a time $t$ is essentially approximating the unitary $e^{-iHt}$ to some precision $\epsilon$, according to the Schr{\"{o}}dinger equation.
In this paper, we focus on the digital quantum simulation (rather than the analog quantum simulation~\cite{LPSS18}),
that is, simulating the Hamiltonian with a fault-tolerant universal quantum computer,
given some oracle access to the Hamiltonian $H$.
The complexity of a quantum simulation algorithm consists of query complexity and gate complexity,
which depends on several factors:
the simulation time $t$, the precision $\epsilon$,
and other parameters of the target Hamiltonian (e.g., size, matrix norm, and sparsity of $H$).

In the literature,  there are basically three approaches to simulating a Hamiltonian:
\begin{itemize}\item 
The \textit{product formula approach} is conceptually the simplest without introducing ancilla qubits.
Early works~\cite{Lloyd96,AT03,BACS06,WBHS10,CK11,CW12} on product-formula-based algorithms often had a poor complexity dependence on the precision $\epsilon$,
which was later improved~\cite{LKW19} by techniques borrowed from other simulation approaches.
This approach has regained attention in recent years due to a number of new results (e.g.,~\cite{CS19,LKW19,Campbell19,COS19,OWC20,CHKT20,SHC21,FSKKE21,HW22,CBH22,LSTT22,ZSJZ22,RWW22})
and its potential to be implemented in the near term.

\item The \textit{LCU (linear combination of unitaries)-based approach}~\cite{CW12} spawned the groundbreaking work~\cite{BCCKS15} that improves the complexity dependence on $\epsilon$
from $\poly\parens*{1/\epsilon}$ to $\polylog\parens*{1/\epsilon}$
(which is also achieved earlier in~\cite{BCCKS14} by fractional queries).
The idea is to approximate $e^{-iHt}$ with a linear combination of unitaries
easier to implement.
This approach can be used to simulating a $d$-sparse Hamiltonian $H$,
which is an Hermitian matrix with at most $d$ nonzero entries in each row.
The Hamiltonian $H$ is accessed by two oracles:
an oracle $\calO_H$ giving the entry $H_{jk}$ according to the index pair $(j,k)$,
and an oracle $\calO_L$ giving the column index of the $c^{\text{th}}$ nonzero entry in row $j$ according to $c$ and $j$.

In particular, combined with Childs' quantum walks~\cite{Childs09,BC12,Kothari14},
this approach achieves a nearly optimal~\cite{BCK15} query complexity for simulating sparse Hamiltonians.
The LCU technique has also been applied to solving the quantum linear systems problem~\cite{CKS17}, which was originally solved in \cite{HHL09} by phase estimation.

\item The lower bound on query complexity for sparse Hamiltonian simulation was finally reached by
the \textit{quantum signal processing approach}~\cite{LYC16,LC17,LC19},
which provides a new way to transform the eigenvalue of a unitary 
by manipulating a single ancilla qubit without performing phase estimation.
The input model was also generalized beyond  sparse matrices by subsequent works on qubitization and block-encoding~\cite{LC19,CGJ19,GSLW19}.\end{itemize}

Later works make further improvements on the complexity dependence on other parameters~\cite{LC17,HHKL18,LW19,Low19,MLC+21},
and the average-case complexity~\cite{ZZS+21}.
We particularly note that almost all\footnote{
The only exception we know is \cite{HHKL18}.
The depth complexity of their algorithm for simulating time-dependent lattice Hamiltonians
is smaller than the size complexity by an $O(n)$ factor, where $n$ is the number of qubits.
}
of the above approaches for Hamiltonian simulation are \textit{sequential}.

\subsubsection*{Parallel quantum computation}

The aim of this paper is to study parallel quantum simulation of Hamiltonians. The computational model we adopt is the quantum circuit model where the running time of a quantum algorithm is measured by the depth of its circuit implementation,
with both gates and oracle queries being allowed to be performed in parallel.

A pioneering work in this direction  is Cleve and Watrous' $O\parens*{\log n}$-depth parallel quantum circuit for $n$-qubit quantum Fourier transform~\cite{CW00},
which can be used to parallelize Shor's factoring algorithm~\cite{Shor94}
with poly-time classical pre- and post-processing.
Following works on parallel Shor's algorithm include factoring on a 2D quantum architecture~\cite{PS13}
and discrete logarithm on elliptic curves~\cite{RS14}.
The limits on parallelizing Grover's 
search algorithm~\cite{Grover96} were first considered by Zalka~\cite{Zalka99}.
This inspired a line of studies~\cite{GWC00,GR04,JMW17,Burchard19,GKLPZ20} on parallel quantum query algorithms and complexity bounds.
The low-depth (parallel) quantum circuit classes are also studied (e.g., \cite{GHP00,MN02,GHMP02,TD04,FGHZ05,HS05,BGH07,TT13,CM20,CCL20,JSTWWZ20}),
amongst which one surprising result is a quantum advantage established by constant-depth quantum circuits over their classical counterparts~\cite{BGK18,WKST19,LeGall19,BGKT20}.
Recently, \cite{QKW22} proposed a quantum algorithm with a constant quantum depth for multivariate trace estimation.

The research on  parallel quantum computation is not restricted to the circuit model.
For example, in measurement-based quantum computing, it was observed that 
parallelism can provide more benefits than in the circuit model~\cite{Jozsa05,BK09,BKP11}.
Another parallel model closer to the current quantum hardware is distributed quantum computing,
which can efficiently simulate the quantum circuit model with low depth overhead~\cite{BBGHKLSS2013}. 
Parallelism is also studied at more abstract levels like quantum programming~\cite{YF09,YZL18}.

\subsection{Main Results}
\label{sub:results}

Our main result is a parallel quantum simulation algorithm for uniform-structured Hamiltonians,
which will be formally defined in Sections~\ref{sub:parallel_version} and~\ref{sub:general_uniform_structured_hamiltonian}.
These Hamiltonians include \textit{local Hamiltonians},
\textit{Pauli sums} and other Hamiltonians of interest.\footnote{A local Hamiltonian can be actually represented as a Pauli sum,
by decomposing each local term with respect to the Pauli basis.
Instead of writing a local Hamiltonian as a Pauli sum,
in this paper we adopt a more natural way to describe the local Hamiltonians,
which still fits well within our framework of uniform-structured Hamiltonians.
}
Roughly speaking, a uniform-structured Hamiltonian $H$ acting on $n$ qubits has the form $H=\sum_{w\in [m]}H_w$,\footnote{In this paper, $[m]$ denotes the set $\braces*{0,1,\ldots,m-1}$ for all $m\in \N$.}
where $m=\poly(n)$, and for each $w$, $H_w$ is a sparse Hamiltonian whose structure is specified by a parameter $s_w$.
Here, the structure of a sparse Hamiltonian is basically a compact way to describe how the non-zero entries are arranged in its matrix representation. For example, for a local Hamiltonian $H$ with $H_w$ being local terms, we can choose each $s_w$ to be an $n$-bit string with $l$ bits of $1$ that indicate the subsystem of $l$ qubits on which $H_w$ non-trivially acts. A good observation of the sparse structure, i.e., a good choice of $s_w$, is definitely important for our algorithm.
We adopt the sparse matrix input model for these Hamiltonians,
that is, a target Hamiltonian $H$ is accessed by two oracles:
an oracle $\calO_H$ giving an entry $H_{jk}$ by the mapping $\ket{j,k,0}\mapsto\ket{j,k,H_{jk}}$,
and an oracle  $\calO_P$ giving the parameter $s_w$ by the mapping $\ket{w,0}\mapsto\ket{w,s_w}$.
Here, $\calO_P$ might be different for various types of Hamiltonians.
Formal details of the input model will be described  in Section~\ref{sec:preliminaries}.

Throughout this paper we assume the target Hamiltonian $H$ is normalized such that $\norm{H}_{\max}=1$,
where $\norm{H}_{\max}:=\max_{jk}\abs*{H_{jk}}$. Then our main result can be stated as the following:
 
\begin{theorem}[Informal version of Theorem~\ref{thm:parallel simulation for uniform-structured Hamiltonian}]\label{main-informal}
	Any uniform-structured Hamiltonian $H=\sum_{w\in [m]}H_w$ acting on $n$ qubits
	with each $H_w$ being $d$-sparse,
	can be simulated for time $t$ to precision $\epsilon$ by
	a quantum circuit of depth $\poly\parens*{\log\log\parens*{1/\epsilon},\log n,m,d,t}$ 
	and size $\poly\parens*{\log\parens*{1/\epsilon},n,m,d,t}$.
\end{theorem}

Here, the running time of our algorithm, i.e., the quantum circuit depth,
has a doubly \mbox{(poly-)}logarithmic dependence on the precision $\epsilon$.
To the best of our knowledge,
this is the first Hamiltonian simulation algorithm that achieves such dependence on $\epsilon$.  

Applying this theorem (actually Theorem~\ref{thm:parallel simulation for uniform-structured Hamiltonian}) to simulating local Hamiltonians, we have: 

\begin{corollary}[Parallel simulation of local Hamiltonians]
	\label{cor:parallel simulation of local Hamiltonians}
	Let $H=\sum_{w\in [m]}H_w$ be an $l$-local Hamiltonian\footnote{
        In this paper, for an $l$-local Hamiltonian, $l$ is not necessarily constant (as is usually assumed) and can depend on $n$ (e.g., $l=\polylog (n)$).
        Also, the local Hamiltonian $H=\sum_w H_w$ is not necessarily \textit{geometrically local},
        which further requires each $H_w$ to act on adjacent $l$ qubits.
    } acting on $n$ qubits,
	where each $H_w$ acts on a subsystem of $l$ qubits
	whose positions are indicated by $l$ bits of $1$ in an $n$-bit string $s_w$.
	Suppose the oracle $\calO_P$ has access to $s_w$ such that $\calO_P\ket{w,0}=\ket{w,s_w}$. 
	Then $H$ can be simulated for time $t$ to precision $\epsilon$ by a quantum circuit of
	\begin{itemize}
		\item 
			depth $O\parens*{\tau\log \gamma}$ and size $O\parens*{\tau\gamma}$
		  w.r.t.\ queries\footnote{
        For simplicity,
		the depth/size (complexity) w.r.t.\ queries/gates refers to the depth/size of the circuit composed of the specified queries/gates, respectively. See Section~\ref{sub:basis_terminologies} for detailed definitions.
			} to $\calO_H$,
		\item
			depth $O\parens*{\tau\log \gamma}$ and size $O\parens*{m\tau\gamma}$ w.r.t.\ queries to $\calO_P$, and
		\item
			depth $O\parens*{\tau\parens*{\log^2\gamma\cdot \log^2 n +\log^3 \gamma}}$
			and size $O\parens*{\tau\gamma^2\cdot\parens*{mn^4+\gamma^3}}$ w.r.t.\ gates,
	\end{itemize}
	where $\tau:=m2^l\cdot t$, $\gamma:=\log\parens*{\tau/\epsilon}$, and we allow arbitrary one- or two-qubit gates.
\end{corollary}

The best known algorithm for this task is by applying the optimal sparse Hamiltonian simulation of~\cite{LC17,LC19}
(note that $H$ in Corollary~\ref{cor:parallel simulation of local Hamiltonians} is $\parens*{m2^l}$-sparse),  
which requires $O\parens*{\tau+\frac{\log (1/\epsilon)}{\log\log(1/\epsilon)}}$ queries
and a factor $O\parens*{n+\log(1/\epsilon)\cdot \polylog(1/\epsilon)}$ of additional gates.
By introducing parallelism, 
our algorithm exponentially improves the dependence on $\epsilon$, in the depth w.r.t.\ both queries and gates.

It is worth noting the difference between the oracle $\calO_P$ used in Corollary~\ref{cor:parallel simulation of local Hamiltonians} for local Hamiltonians
and the oracle $\calO_L$ in previous works~\cite{AT03,Childs09,BC12,BCK15,LC17} for generic sparse Hamiltonians. Oracle $\calO_L$ computes a function $L\parens*{j,c}$ denoting the column index of the $c^{\text{th}}$ non-zero entry in row $j$ of $H$.
In practice, oracle $\calO_P$ is a more natural choice than $\calO_L$,
because if one wants to exploit the local structure of the Hamiltonian to be simulated,
knowing the locality parameter $s_w$ given by $\calO_P$ is intuitively the minimal requirement.
As evidence, for example, when we apply our algorithm to the Heisenberg model in Section~\ref{sub:simulation_of_the_heisenberg_model},
it turns out that implementing the oracle $\calO_P$ is gate-efficient.
Note that in the local Hamiltonian case, a query to the oracle $\calO_L$ can be achieved by at most $m2^l$ queries to $\calO_P$.

\subsubsection*{Lower bounds}

It was shown in~\cite{BCCKS14} that sparse Hamiltonian simulation requires
$\Omega\parens*{\frac{\log\parens*{1/\epsilon}}{\log\log\parens*{1/\epsilon}}}$ queries.
We are able to further prove a lower bound on the gate depth for the Hamiltonian simulation:

\begin{theorem}[$\epsilon$-dependence depth lower bound for Hamiltonian simulation]
    \label{thm:lower bounds}
    Any quantum algorithm for sparse Hamiltonian simulation to precision $\epsilon$ has depth complexity
    $\Omega\parens*{\log\log\parens*{1/\epsilon}}$ w.r.t.\ gates.
    The same holds even for uniform-structured Hamiltonian simulation.
\end{theorem}

This theorem
implies that our parallel quantum algorithm given in Theorem~\ref{main-informal} for simulating uniform-structured Hamiltonians cannot be significantly improved in the $\epsilon$-dependence.

\subsubsection*{Applications} 

For illustration, our algorithm is applied in Section~\ref{sec:applications} to simulating three quantum dynamical models in physics and chemistry of practical interest:
\begin{itemize}
    \item 
        The Heisenberg model for studying the self-thermalization and many-body localization~\cite{NH15,LLA15,CMNRS18}; 
    \item 
        The Sachdev-Ye-Kitaev (SYK) model for studying the simplest AdS/CFT duality~\cite{SY93,Kitaev15,MS16,GELdS17}; and 
    \item 
        A second-quantized molecular model for studying the electronic structure of a molecule~\cite{YWBTA14,BBMC20,BBKWLA16}.
\end{itemize}
We explicitly calculate the gate cost for implementing the oracles mentioned above and the total gate complexity for the simulation.
Table~\ref{tab:comparison with prior works simulating three models} shows a comparison of our algorithm with previous best known algorithms on the same tasks. From it, one can see that by introducing parallelism, our algorithm achieves an exponential speed-up on the $\epsilon$-dependence for simulating all these models. 

It is however worth pointing out that the parallel speed-up in the above applications has limitations in practice,
because the $n$-dependence in the depth complexity (w.r.t.\ gates) of our algorithm is worse than previous best algorithms in these cases.
For simulating local Hamiltonians, as stated in Corollary~\ref{cor:parallel simulation of local Hamiltonians},
the $n$-dependence comes from the linear (up to a poly-logarithmic factor) dependence on $m$ (the number of local terms)
and $\norm*{H}_{\max}$ (the norm of the Hamiltonian).
For example, in the SYK model,
$m$ scales as $O\parens*{n^4}$ and $\norm*{H}_{\max}$ scales as $O\parens*{n^{2.5}}$,
so our algorithm has \mbox{$n$-dependence} $\tilde{O}\parens*{n^{6.5}}$.
In comparison, the previous best algorithm~\cite{BBN19} for simulating the SYK model has a sublinear dependence on $m$,
by using asymmetric qubitization to exploit the structure of the Hamiltonian,
which results in a better $n$-dependence $\tilde{O}\parens*{n^{3.5}}$.
Due to the above issue,
in practical scenarios where the required precision is not too high, for example, $\epsilon=2^{-O\parens*{n}}$,
our improvement on the $\epsilon$-dependence becomes less significant because the $n$-dependence is the dominant part.

\begin{table}
\centering
\bgroup
\def\arraystretch{1.3}
\begin{tabular}{c | c c c}
    \hline
    Model & Algorithms & Gate depth ($n,t$) & Gate depth ($\epsilon$)\\
    \hline
    \multirow{3}{*}{Heisenberg model} 
    &Childs et al.~\cite{CMNRS18}
    &$\tilde{O}\parens*{n^3}$ & $O\parens*{\frac{\log\parens*{1/\epsilon}}{\log\log\parens*{1/\epsilon}}}$\\
    &Haah et al.~\cite{HHKL18}
    &$\tilde{O}\parens*{n}$ & $\polylog\parens*{1/\epsilon}$\\
    \cline{2-4}
    &Our algorithm
    &$\tilde{O}\parens*{n^3}$
    &$O\parens*{\log^3\log\parens*{1/\epsilon}}$\\
    \hline
    \multirow{2}{*}{SYK model}
    & Babbush et al.~\cite{BBN19}
    & $\tilde{O}\parens*{n^{3.5} t}$ & $\polylog\parens*{1/\epsilon}$\\
    \cline{2-4}
    & Our algorithm
    &$\tilde{O}\parens*{n^{6.5}t}$ & $O\parens*{\log^3\log\parens*{1/\epsilon}}$\\
    \hline
    \multirow{4}{*}{Molecular model}
    & Babbush et al.~\cite{BBKWLA16}
    & $\tilde{O}\parens*{n^8 t}$ & $O\parens*{\frac{\log\parens*{1/\epsilon}}{\log\log\parens*{1/\epsilon}}}$\\
    & \begin{tabular}{@{}c@{}}
    Later improvements\vspace{-0.1cm}\\ (e.g.~\cite{BBSKSWLA17,BWMMNC18,KMWGACB18,BBMN19,BGMMB19})
    \end{tabular}
    & $\poly\parens*{n,t}$ & $\polylog\parens*{1/\epsilon}$\\
    \cline{2-4}
    & Our algorithm
    & $\tilde{O}\parens*{n^8t}$ & $O\parens*{\log^3\log\parens*{1/\epsilon}}$
    \\
    \hline
\end{tabular}
\egroup
\caption{A comparison of our algorithm  (Theorem~\ref{thm:parallel simulation for uniform-structured Hamiltonian}) with previous best algorithms in simulating three physical models. Here, parameter $n$ is the size of the system to be simulated,\protect\footnotemark{} $t$ is the simulation time, and $\epsilon$ is the precision of simulation. The complexity of an algorithm is measured by the depth complexity w.r.t.\ gates, where for readability the dependence on different parameters are split. The notation $\tilde{O}\parens*{\cdot}$ denotes an asymptotic upper bound suppressing poly-logarithmic factors. For Heisenberg model, we follow the convention of taking $t=n$.}
\label{tab:comparison with prior works simulating three models}
\end{table}

\subsection{High-level Overview of the Algorithm}
\label{sub:high_level_overview}

\footnotetext{More specifically, $n$ has different meanings in different models: for the Heisenberg model, it is the exact number of qubits on the spin chain; for the SYK model, it is a half of the number of Majorana fermions; and for the second-quantized molecular model, it is the number of spin orbitals.}

Our algorithm is motivated by the LCU (linear combination of unitaries)-based approach to Hamiltonian simulation~\cite{CW12,BCCKS15,BCK15}.
The basic idea of this approach is to approximate the target unitary $e^{-iHt}$
with a linear combination of unitaries easier to implement.
In particular, \cite{BCK15} uses
a Chebyshev series approximation $e^{-iHt}\approx\sum_{r}\alpha_r\calT_r\parens*{H}$,
where for each $r$, $\alpha_r\in \Co$ is some appropriate coefficient, and $\calT_r\parens*{x}$ is the $r$-degree Chebyshev polynomial.
Each $\calT_r\parens*{H}$ can be obtained by $r$ steps of Childs' quantum walk~\cite{Childs09,BC12}.
Then a linear combination of these quantum walks is performed by the LCU technique~\cite{CW12,Kothari14,BCK15}.
Essentially, this approach is sequential due to the fact that $r$ steps of quantum walk require $r$ sequential queries,
and to achieve a total precision $\epsilon$ of the simulation, $r$ should be as large as $\Theta\parens*{\log\parens*{1/\epsilon}}$,
inducing a logarithmic precision-dependence.

In this work, we introduce a parallel quantum walk which is implementable with a quantum circuit of \textit{constant} depth w.r.t.\ queries,
for a large class of Hamiltonians with good sparse structures --- uniform-structured Hamiltonians.
The parallel quantum walk is \textit{not} a direct parallelization of Childs' quantum walk;
instead it implements a monomial of $H$.
We express the unitary $e^{-iHt}$ as a Taylor series $e^{-iHt}\approx\sum_r\beta_r H^r$ (like in the previous work~\cite{BCCKS15}),
where each $r$-degree monomial $H^r$ can be obtained (with a proper scaling factor) by an $r$-parallel quantum walk.
These parallel quantum walks are then linearly combined in parallel by a technique described in Section~\ref{sec:linear_combinations_of_hamiltonian_powers},
which exploits parallelism in the LCU algorithm to combine $R$ terms with depth complexity $\polylog\parens*{R}$.
Since there are about $O\parens*{\log\parens*{1/\epsilon}}$ terms in the LCU to achieve a total precision $\epsilon$ of the simulation,
the depth complexity of our parallel algorithm w.r.t.\ queries is roughly $\polylog\log\parens*{1/\epsilon}$,
achieving a doubly (poly-)logarithmic precision-dependence.

\subsubsection*{Parallel quantum walk}
The main ingredient in our parallel simulation algorithm is the parallel quantum walk.
For an intuition of the algorithm,
let us consider a very special case of uniform-structured Hamiltonians for example, a tensor product of Pauli matrices;
that is, $H=\bigotimes_{t\in [n]}\sigma_t$ for $\sigma_t \in \braces*{\bbone, X,Y,Z}$ being Pauli matrices.
Although $H$ is a simple $1$-sparse Hamiltonian,
it suffices for an illustration of the main idea in our algorithm.
To begin with,
one can think of $H$ as a weighted adjacency matrix of a $1$-sparse graph $\mathsf{H}$,
then a step of Childs' quantum walk without parallelism is a ``superposition version'' of a classical random walk:
it performs $\ket{j}\mapsto H\ket{j}=H_{jk}\ket{k}$ for all vertices $j$ in $\mathsf{H}$,
where $k$ is the unique neighbor of vertex $j$,
and $H_{jk}\in \braces*{1,i,-1,-i}$.

For this special case,
multiple steps of Childs' quantum walk can be directly implemented in parallel,
as shown in the following.
Observe that the graph $\mathsf{H}$
consists of $2^{n-1}$ pairs of vertices,
among which for each pair $(j,k)$, multiple steps of Childs' quantum walk simply
alternate between the state $\ket{j}$ and $\ket{k}$ with an accumulated phase.
If we can predict (compute) the destination of $r$ steps of this walk
and the accumulated phase by quantum circuits with depth significantly less than $r$,
then these $r$ steps can be implemented efficiently in parallel.

Note that for $\mathsf{H}$,
there is a ``uniform'' way to determine the unique neighbor of a vertex.
Consider an $n$-bit string $s$ that characterizes the diagonality of Pauli matrices in the sequence $\angles*{\sigma_t}_{t\in [n]}$:
$s$ has its $t^{\text{th}}$ bit being $0$ if $\sigma_t\in \braces*{\bbone,Z}$ is diagonal,
and being $1$ if $\sigma_t\in \braces*{X,Y}$ is off-diagonal.
Then it is easy to see that $j,k$ are neighbors if and only if $k=j\oplus s$, where $\oplus$ is the bit-wise XOR operator.
Assume $s$ is given by an oracle $\calO_P$,
a step of Childs' quantum walk $\ket{j}\mapsto H_{jk}\ket{k}$ can be implemented by first querying $\calO_P$ to calculate $k=j\oplus s$,
followed by querying $\calO_H$ to add the phase $H_{jk}$.

Now $r$ sequential steps of the walk essentially perform the mapping 
$\ket{j}\mapsto H^r \ket{j}=H_{jk}^{\ceil{r/2}} H_{kj}^{\floor{r/2}}\ket{l}$,
where $l=j\oplus s\oplus \ldots \oplus s=j\oplus s^{\oplus r}\in \braces*{j,k}$ is the destination of these steps.
We can perform this mapping in parallel through two stages.
In the first stage, we compute the destination $l$:
first query $r$ oracles $\calO_P$ simultaneously to compute $r$ copies of $s$;
then classically compute $s^{\oplus r}$ in a binary tree of depth $O\parens*{\log r}$,
using the associativity of XOR.
In the second stage,
the accumulated phase $H_{jk}^{\ceil{r/2}} H_{kj}^{\floor{r/2}}$ is computed in a similar fashion,
by querying $\calO_H$ in parallel combined with parallel classical computation.

By a standard technique in reversible computation~\cite{Bennett73},
the above parallel classical computation can be easily converted to parallel quantum computation,
inducing depth complexity $O\parens*{1}$ w.r.t.\ queries and depth complexity $O\parens*{\log r}$ w.r.t.\ gates.

Although for the above example,
one can easily find faster ways (by noting that $s^{\oplus r}$ is either $s$ or $0$ depending on $r\bmod 2$)
to implement multiple steps of Childs' quantum walk in parallel,
we adopt the above two-stage procedure because it captures the main idea of our algorithm.
In general, denote $H$ the Hamiltonian to be simulated and $\mathsf{H}$ its corresponding graph.
We define an $r$-parallel quantum walk that can be implemented in two stages,
where the two oracles $\calO_H$ and $\calO_P$ are queried separately.
The first stage is called \textit{pre-walk}
which, roughly speaking, prepares a superposition over all paths of length $r$ generated from $r$ steps of unweighted random walk on the graph $\mathsf{H}$.
This stage intuitively predicts all possible paths in $r$ steps of quantum walks, and can be done efficiently in parallel,
with only a constant depth complexity w.r.t.\ queries to $\calO_P$,
provided that the Hamiltonian $H$ is uniform-structured.
The second stage is called \textit{re-weight},
which adjusts the weights (i.e., the quantum amplitudes) of the state prepared by pre-walk,
according to the entry values $H_{jk}$ given by the oracle $\calO_H$.
This stage does not depend on the structure of $H$,
and can be done efficiently in parallel with a constant depth complexity w.r.t.\ queries to $\calO_H$.
Finally combining with other techniques as in Childs' quantum walk,
we can implement the monomial operator $H^r$ (with a proper scaling factor, see Section~\ref{sec:parallel_quantum_walk} for details).

Note again that in the above example, 
although the parallel quantum walk performs multiple steps of Childs' quantum walk,
in general they are \textit{not} essentially equivalent. Specifically, $r$-parallel quantum walk implements $H^r$,
while $r$ steps of Childs' quantum walk yields $\calT_r(H)$.
For those readers familiar with block-encoding (see Definition~\ref{def:block-encoding}),
the parallel quantum walk actually provides a depth-efficient way to block-encode the operator $H^r$,
if $H$ is uniform-structured. It is more like an extension of a single step of Childs' quantum walk.

\subsubsection*{Parallel LCU for Hamiltonian series} To approximate the evolution unitary $e^{-iHt}$ by a truncated Taylor series,
the final step of our algorithm is to linearly combine the monomials $H^r$ obtained from the parallel quantum walks discussed above. The ordinary LCU algorithm~\cite{BCK15,CKS17} implementing a linear combination of $R$ unitaries has depth complexity $\Theta\parens*{R}$.
As pointed out in~\cite{CKS17}, if these unitaries are powers of a single unitary,
then the LCU can be done in a parallel way analogous to the phase estimation~\cite{NC10}
with depth complexity $O\parens*{\log R}$. 
We slightly generalize this result to implementing a linear combination of block-encoded (see Definition~\ref{def:block-encoding}) powers of a Hamiltonian
(called a Hamiltonian power series) in parallel.
Since the LCU requires a state corresponding to the coefficients in the linear combination,
we also present a parallel quantum algorithm for this state-preparation procedure,
based on standard results in quantum sampling~\cite{GR02}.

{\vskip 4pt}

To summarize, the  whole algorithm is visualized in Figure~\ref{fig:scheme}.

\begin{figure}
	\centering
	\begin{tikzpicture}[>=stealth]
		\tikzstyle{box}=[draw,
			rounded corners=2mm,
			minimum width=4.5em,
			minimum height=2em];
		\node[draw,very thick,minimum height=3cm, minimum width=5.2cm, align=center] (Parallel-quantum-walk) [label={[label distance=-0.8cm]above:Parallel quantum walk $Q^{(r)}$}] {};
		\node (midpoint) [below = 0.15cm of Parallel-quantum-walk.west] {};
		\node (plus) [right=2.1cm of midpoint]{\textbf{+}};
		\node[draw,align=center] (Pre-walk) [above right = 0.1cm and 0.3cm of midpoint] {1. Pre-walk (Sec.~\ref{sub:pre_walk})};
		\node[draw,align=center] (Re-weight) [below right = 0.1cm and 0.3cm of midpoint] {2. Re-weight (Sec.~\ref{sub:re_weight})};

		\node[box,align=center] (Hamiltonian) [right = 4.7cm of Parallel-quantum-walk.east] {A uniform-structured\\ Hamiltonian $H$\\ (Sec.~\ref{sub:pre_walk} \& Sec.~\ref{sub:general_uniform_structured_hamiltonian})};
		\draw [->] (Hamiltonian.north west) to [bend right=20] node[pos=0.4,below=0.1cm] {sparse structure} node[pos=0.3,above] {$O(1)$-depth w.r.t.\ queries to $\calO_P$} (Pre-walk.east);
		\draw [->] (Hamiltonian.south west) to [bend left=20] node [pos=0.5,above=0.08cm] {entry value} node[pos=0.3,below=0.08cm] {$O(1)$-depth w.r.t.\ queries to $\calO_H$} (Re-weight.east);

		\node[box,align=center] (Hr) [below = 1.2cm of Parallel-quantum-walk.south] {A parallel implementation\\ of $H^r$(with a proper factor;\\ Sec.~\ref{sub:parallel_version})};
		\draw [->] (Parallel-quantum-walk.south) to (Hr.north);

		\node[box,align=center] (unitary) [right = 5cm of Hr.east] {Target unitary $e^{-iHt}$\\ (with small $t$; Sec.~\ref{sec:parallel_hamiltonian_simulation})};
		\draw [->] (Hr) to node[above,align=center] {Parallel LCU (Sec.~\ref{sec:linear_combinations_of_hamiltonian_powers})\\$r=1,2,\ldots,R=\ceil*{\log(1/\epsilon)}$ }
			node[below] {$\polylog\parens{R}$-depth} (unitary);
	\end{tikzpicture}
	\caption{An outline of the parallel quantum algorithm for Hamiltonian simulation.}
	\label{fig:scheme}
\end{figure}

\subsection{Related Works}
\label{sub:related_works}

Our algorithm for Hamiltonian simulation shows that by employing parallelism, the complexity dependence on the precision $\epsilon$ can be significantly reduced from $\polylog\parens*{1/\epsilon}$ to $\polylog\log\parens*{1/\epsilon}$.
For reducing the dependence on other parameters like the simulation time $t$,
Atia and Aharonov~\cite{AA17} studied the fast-forwarding of Hamiltonians (which is further explored in a recent work~\cite{GSS21}) --- 
the ability to simulate a Hamiltonian by a quantum circuit with \textit{size} significantly less than the simulation time $t$ (e.g.\ $\polylog(t)$).
Note that the concept of fast-forwarding does not concern parallelism and depth complexity.
They show that the fast-forwarding of generic Hamiltonians is impossible unless $\BQP=\PSPACE$.\footnote{Note that this differs from the well-known ``no-fast-forwarding theorem'' in~\cite{BACS06} stating that no sparse Hamiltonian can be simulated for time $t$ with sub-linear query complexity in $t$.
Atia and Aharonov's result~\cite{AA17} is not restricted to the query model.
}
Nevertheless, they provided three examples of fast-forwardable Hamiltonians:
Hamiltonians constructed from the modular exponentiation unitary in Shor's algorithm,
commuting local Hamiltonians, and quadratic Hamiltonians.
Interestingly, although these examples belong to
the class of uniform-structured Hamiltonians defined in this paper (by properly setting the oracle $\calO_H$ and $\calO_P$),
we are not aware of a direct way to extend our results to reduce the dependence on $t$ by parallelization.

Jeffery, Magniez, and de Wolf~\cite{JMW17} studied the parallel query complexity of element-distinctness and $k$-sum problems. 
The upper bounds in their results are obtained by what they called ``parallelized quantum walk algorithms''.
But it should be noted that their algorithm is developed in the framework of finding marked elements via quantum walks~\cite{MNRS11},
in particular, a modified quantum walk on multiple copies of Johnson graphs corresponding to a specific function (element-distinctness or $k$-sum),
with parallel queries to that function. In contrast, 
our parallel quantum walk is defined in the framework of quantum walks for Hamiltonians~\cite{Childs09,BC12,BCK15,CKS17},
in particular, a quantum walk that implements a monomial $H^r$ (with a proper scaling factor) by $r$ parallel queries to the oracles accessing a Hamiltonian $H$.

\subsubsection*{Recent Developments}

After the work described in this paper, the depth complexity $O\parens*{\tau\parens*{\log^2\gamma\cdot \log^2 n +\log^3 \gamma}}$ of our parallel quantum algorithm for local Hamiltonian simulation (see Corollary~\ref{cor:parallel simulation of local Hamiltonians}) was further improved to $O\parens*{\tau \log^2\gamma\cdot \log^2 n}$ in \cite{ZLY22}, which is achieved by depth-optimal quantum state preparation \cite{STY+21,Ros21,ZLY22,PZ22}. 

Recently, the depth lower bounds on Hamiltonian simulation were studied in \cite{CCHLLS23}.
They showed that simulating time-independent sparse Hamiltonians in general requires $\Omega(t)$ depth w.r.t.\ to queries;
and based on the random oracle heuristic~\cite{BR93,BDFLSZ11}, simulating time-independent local Hamiltonians
and time-dependent geometrically local Hamiltonians requires $\Omega(t/n^c)$ depth w.r.t.\ gates,
where $c$ is some constant.
This gives a negative answer to the second open question in Section~\ref{sub:discussion} in the general case.

\subsection{Discussion}
\label{sub:discussion}

In this paper, we propose a parallel quantum algorithm for Hamiltonian simulation that achieves a doubly (poly-)logarithmic precision-dependence in the depth complexity.
This exponential improvement from the previous (poly-)logarithmic precision-dependence is rather a theoretical result,
because in practice if the precision required is not too high, the improvement becomes less significant.
Still, regarding the great importance of parallel computing and its success in the classical world,
we hope our work may shed light on the design of parallel quantum algorithms 
and our ideas may contribute to future Hamiltonian simulation in practice.

Some readers might be also concerned that 
in order to extract classical information to the required precision from the output quantum state,
the overhead of measurements could dominate the overall complexity.
Indeed, our analysis only focuses on the asymptotic cost of the quantum simulation itself, without taking the cost of post-processing into account. Nevertheless, one can expect that the output quantum state of our algorithm can be coherently fed to another quantum system (see~\cite{Llo00,GJ09,WLY21}). For example, the output state of a quantum algorithm can be directly used as the input data in quantum machine learning~\cite{KPEKML20}, which is much more efficient than preparing the input data classically.

One direction of future work is to consider whether the complexity of Hamiltonian simulation can be further improved. 
\begin{itemize}
    \item Given simulation time $t$ and precision $\epsilon$, the depth complexity w.r.t.\ queries of our algorithm is $O\parens*{\tau \log \gamma}$, while the query complexity of the best sequential quantum algorithm~\cite{LC17,LC19} is roughly $O\parens*{\tau + \gamma}$, where $\tau \propto t$ and $\gamma = \log\parens*{\tau/\epsilon}$. It remains open whether the multiplicative dependence on $\tau$ and $\gamma$ can be improved to additive $O\parens*{\tau + \log \gamma}$. 
	
    \item Another open question, as mentioned in Section~\ref{sub:related_works},
is whether parallelization can significantly improve the dependence on the simulation time $t$;
more precisely, what kind of Hamiltonians can be simulated for time $t$ by a quantum circuit of depth $\polylog\parens{t}$,
by allowing parallel queries?

    \item Intuitively, the power of parallelism is provided by the use of ancilla qubits.
    Considering the dependence on precision $\epsilon$,
    from the upper bound on the circuit size in Theorem~\ref{main-informal},
    we can derive a simple upper bound $O\parens*{\polylog(1/\epsilon)}$
    on the number of ancilla qubits in our algorithm.
    On the other hand, 
    to achieve a significant speed-up (in particular, a circuit depth $O\parens*{\polylog\log(1/\epsilon)}$),
    we can also derive a lower bound $\Omega\parens*{\frac{\log(1/\epsilon)}{\polylog\log(1/\epsilon)}}$ on the number of ancillae
    in simulating uniform-structured Hamiltonians.
    This is obtained by the fact that the size of a quantum circuit is no greater than the product of the circuit depth and the number of qubits, combined with the lower bound $\Omega\parens*{\frac{\log(1/\epsilon)}{\log\log(1/\epsilon)}}$ on the circuit size (see the proof of Theorem~\ref{thm:lower bounds} and~\cite{LC17,LC19}).
    Therefore, whether the number of ancillae can be further reduced remains an interesting question.
    More generally, 
    the trade-off between the circuit depth and the number of ancillae will be an important issue in parallel quantum algorithms.

	\item 
		As mentioned in Section~\ref{sub:results},
		in several practical applications the depth complexity (w.r.t.\ gates) of our algorithm
		has worse dependence on the system size compared to the previous best algorithms.
		Is it possible to reduce the $n$-dependence in the complexity
		while retaining the improvement on the $\epsilon$-dependence?

\end{itemize}
 
Another important direction is to study what other kinds of Hamiltonians with special structures
can be simulated efficiently in parallel.
Although the class of uniform-structured Hamiltonians in this work includes many physical Hamiltonians,
one might also encounter non-physical ones in other quantum algorithms
that use Hamiltonian simulation as a subroutine (e.g., quantum linear system solver~\cite{HHL09,CKS17}, and quantum differential equations solver~\cite{BCOW17}).
How about exploiting the structures of those non-physical Hamiltonians for parallelism?

Finally, a very interesting attempt is to generalize our ideas and techniques to broader applications.
\begin{itemize}
    \item 
        In this paper we are only concerned about time-independent Hamiltonian simulation.
		For time-dependent Hamiltonian simulation, one approach~\cite{BCCKS15,LW19,KSB19,BCSWW20,CKH21} is to use the Dyson series 
		to replace the Taylor series to
		approximate the unitary evolution.
		An interesting question is whether our algorithm (in particular, the parallel quantum walk) 
		can be generalized for simulating time-dependent Hamiltonians.
    \item 
    The major technical tool introduced in this work is a parallel quantum walk for Hamiltonians,
which we believe can also be extended to general matrices.
This tool might be applicable to other problems,
for example, solving quantum linear differential equations~\cite{BCOW17},
for which an operator $e^{At}$ is to be implemented instead of the unitary $e^{iHt}$ for Hamiltonian simulation,
where $A$ is a general matrix.
    \item
As quantum singular value transform (QSVT)~\cite{GSLW19} has shown its power in manipulating the block-encodings of general matrices,
and serves as a building block for many other quantum algorithms (e.g.,~\cite{AGGW17,GAW19,KP20}),
it is also interesting to ask whether parallelism can provide speed-ups for QSVT.
\end{itemize}

\subsection{Structure of the Paper}
\label{sub:structure_of_the_paper}

For convenience of the reader, some preliminaries are presented in Section~\ref{sec:preliminaries}.  
In Section~\ref{sec:parallel_quantum_walk}, 
we introduce a parallel quantum walk for Hamiltonians. More concretely, 
we revisit the framework of Childs' quantum walk in Section~\ref{sub:quantum_walk_for_hamiltonians},
give a parallelization of it and show how to implement this parallel quantum walk in Section~\ref{sub:parallel_version},
and analyze the complexity in Sections~\ref{sub:pre_walk} and~\ref{sub:re_weight}.
Specifically in Section~\ref{sub:pre_walk}, we define uniform-structured Hamiltonians,
for which the parallel quantum walk can be performed efficiently. 
In Section~\ref{sub:general_uniform_structured_hamiltonian}
we present an extension of the parallel quantum walk to the case of a sum of sparse Hamiltonians,
where we also extend the class of uniform-structured Hamiltonians to include more Hamiltonians of interest.
In Section~\ref{sec:linear_combinations_of_hamiltonian_powers}, we show how to implement a parallel LCU for a Hamiltonian power series.
Section~\ref{sec:parallel_hamiltonian_simulation} assembles the above results to simulate a Hamiltonian by combining parallel quantum walks.
In Section~\ref{sec:lower_bounds}, we prove an $\epsilon$-dependence lower bound on the gate depth for Hamiltonian simulation.
In Section~\ref{sec:applications}, we apply our algorithm to simulate three concrete physical models.

\section{Preliminaries and Notations}
\label{sec:preliminaries}

\subsection{Basic Terminologies} 
\label{sub:basis_terminologies}
\subsubsection*{Sparse matrix input model}

Let $H$ be a $d$-sparse $N\times N$ Hamiltonian acting on $n$ qubits with $\norm{H}_{\max}=1$.
$H$ is accessed by:
\begin{itemize}
	\item 
		An entry oracle $\calO_H^b$ that gives an entry of $H$ with $b$-bit precision such that
		\begin{equation*}
			\calO_H^b\ket{j,k,z}=\ket{j,k,z\oplus H_{jk}},
		\end{equation*}
		for $j,k\in [N]$ and $z,H_{jk}\in [2^b]$,
		where the complex number $H_{jk}$ is stored in a $b$-bit string that contains its real part and imaginary part
		(each with $b/2$ bits assuming $b$ is even),
		and $b$ will be determined by the precision $\epsilon$ of the algorithm.
		When we are just referring to the oracle we may omit the superscript $b$.
	\item
		A sparse structure oracle $\calO_P$ that gives parameters about the sparse structure of $H$ such that
		\begin{equation*}
			\calO_P\ket{x,y}=\ket{x,P(x,y)}
		\end{equation*}
		for $x\in \calX, y\in \calY$,
		where $\calX,\calY$ are sets of integers,
		and $P:\calX\times\calY\rightarrow \calY$
		 is a function determined by the sparse structure of $H$
		such that for all $x\in \calX$, $P\parens*{x,\cdot}$ is a bijection for $y\in \calY$.
\end{itemize}
In many previous works~\cite{AT03,Childs09,BC12,BCK15,CKS17} on sparse Hamiltonians,
instead of $\calO_P$, an oracle $\calO_L$ is given such that
$\calO_L\ket{j,t}=\ket{j,L(j,t)}$,
where $L(j,t)\in [N]$ gives the column index of the $t^{\text{th}}$ nonzero entry in row $j$ of $H$ for $t\in [d]$
and gives $0$ for $t\notin [d]$.
We note that $\calO_L$ can be expressed as a special case of $\calO_P$ by taking $\calX=\calY=[N]$ and $P=L$.
However, compared to these work,
we adopt a more general oracle $\calO_P$
because Hamiltonians investigated in this paper have different sparse structures, which can be better exploited by different concrete forms of $\calO_P$ (see examples in Sections~\ref{sub:parallel_version} and~\ref{sub:general_uniform_structured_hamiltonian}).
Moreover, the implementations of these $\calO_P$'s turn out to be very efficient, compared to much costlier implementations of the oracle $\calO_H$, as shown in Section~\ref{sec:applications} when we calculate the total gate complexity of our algorithm in simulating practical physical models.

As in previous works,
we allow using (single-qubit) controlled versions of these oracles,
and will not explicitly distinguish between the controlled and uncontrolled versions.

\subsubsection*{Complexity model}
In this paper, we will consider both query complexity and gate complexity,
where we allow arbitrary one- or two-qubit gates.
For a parallel quantum algorithm represented by a quantum circuit,
we will measure its cost by its depth and size.
When considering the gate complexity,
each oracle query is temporarily counted as one gate.
The depth of a quantum circuit w.r.t.\ gates is defined as the length of the longest path composed of gates and wires from the circuit inputs to outputs, where the length is the total number of gates on this path.
The size of a quantum circuit w.r.t.\ gates is defined as the total number of gates in the entire quantum circuit.

We adopt the definition of parallel query complexity in previous works (see for example~\cite{JMW17}). 
The query complexity is calculated for each oracle separately.
When considering the query complexity with respect to an oracle $\calO$,
all gates and queries to other oracles are ignored.
We allow queries to multiple copies of $\calO$ in parallel ---
that is, we can perform the mapping
\begin{equation*}
	\ket{\psi_1}\otimes \ldots \otimes \ket{\psi_r}\mapsto \calO\ket{\psi_1}\otimes \ldots \otimes \calO\ket{\psi_r}
\end{equation*}
in a single time-step for some $r$.
From the viewpoint of a quantum circuit, these parallel queries to $\calO$ act as ``gates'' in the same circuit layer.
The depth of a quantum circuit w.r.t.\ queries to $\calO$ is defined as the length of the longest path composed of queries to $\calO$ and wires from the circuit inputs to outputs, where the length is the total number of queries to $\calO$ on this path.
The size of a quantum circuit w.r.t.\ queries to $\calO$ is then defined as the number of total queries in the entire quantum circuit.

We use the following terminologies. The depth/size complexity of an algorithm (resp.\ problem) refers to the depth/size (w.r.t.\ queries or gates if specified) of quantum circuits implementing (resp.\ solving) it. Without specification, the depth/size (complexity) implicitly means the depth/size (complexity) w.r.t.\ gates.
The query/gate depth refers to the depth complexity w.r.t.\ queries/gates.
We will also use the adjectives ``$\alpha$-depth and $\beta$-size'' to describe a circuit of depth $\alpha$ and size $\beta$ (w.r.t.\ queries or gates if specified).

\subsubsection*{Error model}

A pure state $\ket{\psi}$ is said to be approximated by $\ket{\tilde{\psi}}$ to precision $\epsilon$,
if they are close in the $l_2$-norm such that $\norm*{\ket{\psi}-\ket{\tilde{\psi}}}\leq \epsilon$.
A unitary $U$ is said to be implemented to precision $\epsilon$,
if a unitary $\tilde{U}$ is actually implemented such that
$\norm*{U-\tilde{U}}\leq \epsilon$, where the norm is the spectral norm.
The terms ``$\epsilon$-precise'' and ``$\epsilon$-close'' will also be used interchangeably.

\subsection{Parallel Quantum Circuit}
\label{sub:parallel_quantum_circuit}

Now we review some basic techniques for constructing parallel quantum circuits. 
Although the results are known in the previous literature, 
we provide their proofs in order to  illustrate basic ideas in designing  parallel quantum circuits.

\begin{lemma}[Parallel copying~\cite{CW00,MN02}]
	\label{lmm:parallel copy}
	Let $\mathrm{COPY}_b$ be a unitary that creates $b$ copies of a bit (including the original copy);
	that is,
	\begin{equation*}
		\mathrm{COPY}_b\ket{x}\ket{0}^{\otimes b-1}=\ket{x}^{\otimes b}
	\end{equation*}
	for all $x\in \braces*{0,1}$.
	There is an $O\parens*{\log b}$-depth and $O(b)$-size quantum circuit that implements $\mathrm{COPY}_b$.
\end{lemma}

\begin{proof}
	Suppose $b$ is a power of two w.l.o.g.
	Then as shown in Figure~\ref{fig:parallel quantum circuit for COPY},
	the gate $\mathrm{COPY}_{b}$ can be inductively constructed from $\mathrm{COPY}_{b/2}$,
	with $\mathrm{COPY}_1$ being the identity operator.
	\begin{figure}[!ht]
		\centering
		\begin{quantikz}[row sep=0.3cm, column sep=0.9cm]
			\lstick{$\ket{x}$} & \ctrl{2}  & \gate[2,bundle={2}]{\mathrm{COPY}_{b/2}} & \qw\\
			\lstick{$\ket{0}^{\otimes b/2-1}$} & \qwbundle[alternate]{} &				& \qwbundle[alternate]{}\\
			\lstick{$\ket{0}$} & \targ{}       & \gate[2,bundle={2}]{\mathrm{COPY}_{b/2}} & \qw\\
			\lstick{$\ket{0}^{\otimes b/2-1}$} & \qwbundle[alternate]{} &                               & \qwbundle[alternate]{}
		\end{quantikz}
		\caption{Inductive construction of the gate $\mathrm{COPY}_{b}$}
		\label{fig:parallel quantum circuit for COPY}
	\end{figure}
	It is easy to check the depth and size of this quantum circuit by induction.
\end{proof}

Lemma~\ref{lmm:parallel copy} can be easily extended from copying a single bit to copying an $m$-bit string, with the  circuit depth and size multiplied by $m$. 
We use $\mathrm{COPY}_b^m$ to denote such a circuit. 
For the case $b=2$, we will omit the subscript and write $\mathrm{COPY}^m$.

\begin{lemma}[Parallel controlled rotations~\cite{CW00,MN02}]
	\label{lmm:parallel controlled Z-rotation}
	Let $\mathrm{C}_b$-$R_Z$ be a unitary that performs a phase shift on a single qubit controlled by $b$ qubits;
	that is,
	\begin{equation}\label{same-XY}
	\mathrm{C}_b\text{-}R_Z\ket{\gamma}\ket{\psi}=\ket{\gamma}R_Z\parens*{2\pi\gamma\cdot 2^{-b}}\ket{\psi}
	\end{equation}
	for all $\gamma\in [2^b]$ and $\ket{\psi}\in \Co^2$,
	where $R_Z(\theta)=e^{-i\theta Z/2}$. 
	Then $\mathrm{C}_b$-$R_Z$ can be implemented by an $O\parens*{\log b}$-depth and $O(b)$-size quantum circuit.
	
	The same conclusion holds for $\mathrm{C}_b$-$R_X$ and $\mathrm{C}_b$-$R_Y$ defined by replacing $R_Z(\theta)$
	in \eqref{same-XY} with
	$R_X(\theta)=e^{-i\theta X/2}$ and $R_Y(\theta)=e^{-i\theta Y/2}$, respectively.
\end{lemma}
\begin{proof}
	Let $\gamma_j$ represent the $j^{\text{th}}$ bit of $\gamma$ for $j\in [b]$.
	A parallel quantum circuit for $\mathrm{C}_b$-$R_Z$ is shown in Figure~\ref{fig:parallel quantum circuit for C_b-R_Z}.
	To see its correctness,
	suppose $\ket{\psi}=\alpha\ket{0}+\beta\ket{1}$,
	then $\mathrm{COPY}_b\ket{\psi}\ket{0}^{\otimes b-1}=\alpha\ket{0}^{\otimes b}+\beta\ket{1}^{\otimes b}$
	is an entangled state.
	Since applying $R_Z$ to any of the $b$ entangled qubits will add to the relative phase,
	the whole state becomes $\ket{\gamma}\parens*{\alpha\ket{0}^{\otimes b}+e^{i2\pi \gamma \cdot 2^{-b}}\beta\ket{1}^{\otimes b}}$
	after applying $\mathrm{C}$-$R_Z$ in parallel.
	The final state is obtained by reverse computation with $\mathrm{COPY}_b^\dagger$.
	It is easy to check the depth and size of the quantum circuit by Lemma~\ref{lmm:parallel copy}.
	\begin{figure}[!ht]
		\centering
		\begin{quantikz}[row sep=0.3cm, column sep=0.2cm]
			\lstick{$\ket{\gamma_0}$}     & \qw                                   & \ctrl{4}                    & \qw                         & \qw    & \qw      & \qw & \qw\\
			\lstick{$\ket{\gamma_1}$}     & \qw                                   & \qw                         & \ctrl{4}                    & \qw    & \qw      & \qw &\qw\\
			\lstick{$\vdots$}             &                                       &                             &                             & \ddots &          & &\\
			\lstick{$\ket{\gamma_{b-1}}$} & \qw                                   & \qw                         & \qw                         & \qw    & \ctrl{4} & \qw &\qw\\
			\lstick{$\ket{\psi}$}         & \gate[4, nwires={3}]{\mathrm{COPY}_b} & \gate{R_Z \parens*{2^{0}\pi}} & \qw                         & \qw    & \qw      & \gate[4,nwires={3}]{\mathrm{COPY}_b^\dagger}&\qw\\[-0.4cm]
			\lstick{$\ket{0}$}            &                                       & \qw                         & \gate{R_Z \parens*{2^{-1}\pi}} & \qw    & \qw      & \qw &\qw\\[-0.4cm]
			\lstick{$\vdots$}             &                                       &                             &                             & \ddots &          & &\\[-0.4cm]
			\lstick{$\ket{0}$}            &                                       & \qw                         & \qw                         & \qw &\gate{R_Z\parens*{2^{-b+1}\pi}} &\qw &\qw
		\end{quantikz}
		\caption{A parallel quantum circuit for $\mathrm{C}_b$-$R_Z$}
		\label{fig:parallel quantum circuit for C_b-R_Z}
	\end{figure}

	The cases of $X$ and $Y$-rotations can be easily proved by combining the above proof and the identities $R_X(\theta)=\textit{Had}\,R_Z(\theta)\,\textit{Had}$ and $R_Y(\theta)=R_Z(\pi/2)R_X(\theta)R_Z(\pi/2)$,
	where $\textit{Had}$ stands for the Hadamard gate.
\end{proof}

The following lemma accommodates  a classical parallel computing technique into the quantum setting.

\begin{lemma}[Sequence of associative operators]
	\label{lmm:sequence of associative operators}
	Let $\circ$ be an associative operator. 
	Given a unitary $U_\circ$ that performs 
	$\ket{x,y,z}\mapsto\ket{x,y,z\oplus \parens*{x\circ y}}$
	and its inverse $U_\circ^\dagger$, 
	then for all $m\in \N$, the mapping
	\begin{equation*}
		\ket{x_1,\ldots,x_m}\ket{z}\mapsto \ket{x_1,\ldots,x_m}\ket{z\oplus \parens*{x_1\circ\ldots \circ x_m}}
	\end{equation*}
	can be implemented by a quantum circuit of depth $O\parens*{\log m}$ and size $O(m)$ w.r.t.\ $U_\circ$ and its inverse, 
	where the additional gate complexity is often negligible compared to $U_\circ$ and thus omitted.
\end{lemma}
\begin{proof}
	Let the partial sum $S_\circ\parens*{l,r}:=x_l\circ\ldots \circ x_r$.
	Our goal is to compute $S_\circ\parens*{1,m}=x_1\circ\ldots\circ x_m$.
	We can compute $S_\circ\parens*{l,r}$ inductively by 
	\begin{equation*}
		S_\circ\parens*{l,r}=S_\circ \parens*{l,p}\circ S_\circ \parens*{p+1,r},
	\end{equation*}
	where $p=(l+r-1)/2$.
	Assuming $m$ is a power of two w.l.o.g.,
	then the computation of $S_\circ\parens*{1,m}$ forms a tree of depth $O\parens*{\log m}$ and size $O(m)$ w.r.t.\ $U_\circ$,
	where the root is $S_\circ\parens*{1,m}$ and the leaves are $S_\circ\parens*{i,i}=x_i$ for $i\in [m]$.
	Once $S_\circ\parens*{1,m}$ is computed into an ancilla space,
	apply COPY gate to copy the result into $\ket{z}$,
	then clean all the garbage partial sums by reverse computation with $U_\circ^\dagger$.
\end{proof}

\begin{corollary}[Parallel addition of a sequence~\cite{CW00}]
	\label{cor:parallel addition of a sequence}
	The addition of a sequence, i.e., the mapping
	\begin{equation*}
		\ket{x_1,\ldots,x_m}\ket{z}\mapsto \ket{x_1,\ldots,x_m}\ket{z\oplus \parens*{x_1+\ldots+x_m}}
	\end{equation*}
	for all $x_1,\ldots,x_m,z\in [2^b]$,
	can be implemented by an $O\parens*{\log m+\log b}$-depth and $O\parens*{mb}$-size quantum circuit,
	where the addition is in $\Z_{2^b}$.
\end{corollary}
\begin{proof}
	By Lemma~\ref{lmm:sequence of associative operators},
	combined with classical techniques of three-two adder, pairwise representation and carry-lookahead adder.
	A detailed proof can be found in~\cite{CW00}.
\end{proof}

\begin{corollary}[Parallel controlled $Z$ gate]
	\label{cor:parallel controlled Z gate}
	Let $\mathrm{C}_b$-$Z$ be a unitary that perform $Z$ gate on a single qubit controlled by $b$ qubits;
	that is,
	\begin{equation*}
		\mathrm{C}_b\text{-}Z\ket{x}\ket{\psi}=
		\begin{cases}
			\ket{x}Z\ket{\psi},& x=2^b-1,\\
			\ket{x}\ket{\psi},& o.w.
		\end{cases}
	\end{equation*}
	for all $x\in [2^b]$ and $\ket{\psi}\in \Co^2$,
	then $\mathrm{C}_b$-$Z$ can be implemented by an $O\parens*{\log b}$-depth and $O(b)$-size quantum circuit.
\end{corollary}
\begin{proof}
	Write $x$ as a bit string $x=x_0\ldots x_{b-1}$,
	then the $\mathrm{C}_b$-$Z$ gate can be implemented in the following ways.
	First take $\circ$ to be AND gate in Lemma~\ref{lmm:sequence of associative operators}
	to compute $x_0\wedge \ldots \wedge x_{b-1}$ in the ancilla space by an
	$O\parens*{\log b}$-depth and $O\parens*{b}$-size quantum circuit,
	then conditioned on this apply a $\mathrm{C}$-$Z$ gate to $\ket{\psi}$,
	and finally clean the garbage by reverse computation.
\end{proof}

The following lemma translates the parallel classical results of computing elementary arithmetic functions~\cite{Reif83} to the quantum case using the technique of reversible computing~\cite{Bennett73}.

\begin{lemma}[Parallel quantum circuit for elementary arithmetics~\cite{Reif83}]
	\label{lmm:parallel quantum circuit for elementary arithmetics}
	Let $f$ be one of the following elementary arithmetic functions:
	addition, subtraction, multiplication, division, cosine, sine, arctangent,
	exponentiation, logarithm,
	maximum, minimum, factorial.\footnote{In particular, for factorial, $x$ is required to be $O\parens*{b^2}$.}
        Then the unitary that performs
	\begin{equation*}
		\ket{\tilde{x}}\ket{\tilde{y}}\ket{z}\mapsto \ket{\tilde{x}}\ket{\tilde{y}}\ket{z\oplus \tilde{f}(\tilde{x},\tilde{y})}
	\end{equation*}
	for all $\tilde{x},\tilde{y},z\in [2^b]$
	can be implemented by an $O\parens*{\log^2 b}$-depth and $O \parens*{b^4}$-size quantum circuit,\footnote{The complexity in~\cite{Reif83} is more refined and complicated. Here we take a simple upper bound.}
	where $\tilde{x},\tilde{y}$ are floating number representation of $x,y$ on suitable intervals,
	and $\tilde{f}(\tilde{x},\tilde{y})$ is $2^{-b}$-close to $f(x,y)$.
	(For unary function the second operand $y$ is omitted, e.g., $f(x,y)=cos(x)$.)
	In particular, for $f$ being addition, subtraction or multiplication, the depth can be $O\parens*{\log b}$.
\end{lemma}

In the remainder of this paper,
the name ``elementary arithmetic function'' may also refer to a composition of a constant number of the arithmetic functions in Lemma~\ref{lmm:parallel quantum circuit for elementary arithmetics}.
As elementary arithmetic operations are frequently used in this paper,
in some cases we will measure the efficiency of computation with respect to these building blocks.
More specifically we have the following definition for \textit{arithmetic-depth-efficient} computation.\footnote{Note that this differs from the notion of \textit{depth-efficient} computation, which often refers to a poly-logarithmic depth computation. As an analog of the classical depth-efficient class $\NC$, \cite{MN02} introduced the quantum depth-efficient class $\QNC$.}

\begin{definition}[Arithmetic-depth-efficient computation]
	\label{def:depth-efficient computation}
    A quantum circuit on $b$ qubits is called arithmetic-depth-efficient if it is $O\parens*{\log^2 b}$-depth and $O\parens*{b^4}$-size.
    A function $f$ is \textit{arithmetic-depth-efficiently computable} if the mapping
	$\ket{x}\ket{z}\mapsto \ket{x}\ket{z\oplus f(x)}$ for all $x,z\in [2^b]$\footnote{Functions (or arithmetic operators) with multiple inputs and outputs can be easily converted to this form.}
	can be implemented by an arithmetic-depth-efficient quantum circuit.
\end{definition}

\subsection{Block-encoding}
\label{sub:block_encoding}

Block-encoding is a recently introduced fundamental tool for arithmetic operations on matrices represented as a block of a unitary.
It has been developed through a series of researches in quantum algorithms~\cite{HHL09, BCK15, CKS17, LC19, AGGW17, CGJ19, GSLW19}.

\begin{definition}[Block-encoding]
    \label{def:block-encoding}
	An $(s+a)$-qubit unitary $U$ is an $(\alpha,a,\epsilon)$-block-encoding of an $s$-qubit operator $A$ if
	\begin{equation*}
		\norm*{A-\alpha\parens*{\bra{0}^{\otimes a}\otimes \bbone}U\parens*{\ket{0}^{\otimes a}\otimes \bbone}}\leq \epsilon.
	\end{equation*}
	Intuitively, this implies the $2^s\times 2^s$ upper-left block of $U$ is $(\epsilon/\alpha)$-close to $A/\alpha$; that is,
	\begin{equation*}
	    U=\begin{pmatrix}
	        B&*\\
	        *&*
	    \end{pmatrix}
	    \text{\qquad and \qquad}
	    \norm*{B-A/\alpha}\leq \epsilon/\alpha.
	\end{equation*}
\end{definition}

In this paper, we will slightly abuse the terminology in such a way that if the condition is
\begin{equation*}
	\norm*{A-\alpha\parens*{\bbone\otimes \bra{0}^{\otimes a}}U\parens*{\bbone\otimes \ket{0}^{\otimes a}}}\leq \epsilon,
\end{equation*}
then for simplicity $U$ is also called an $(\alpha,a,\epsilon)$-block-encoding of $A$.

\section{Parallel Quantum Walk}
\label{sec:parallel_quantum_walk}

In this section, we define a parallel quantum walk within the framework of Childs' quantum walks~\cite{Childs09,BC12,BCK15,CKS17}.
In Section~\ref{sub:quantum_walk_for_hamiltonians}, we revisit Childs' quantum walk.
Then we propose a parallelization of it in Section~\ref{sub:parallel_version} and show how to implement the parallel quantum walk
in two stages: pre-walk and re-weight, whose complexities are analyzed in Section~\ref{sub:pre_walk} and Section~\ref{sub:re_weight} respectively.
In particular, in Section~\ref{sub:pre_walk} we define uniform-structured Hamiltonians
for which the parallel quantum walk can be implemented efficiently.
The Hamiltonians considered in Section~\ref{sub:parallel_version} are $d$-sparse. 
In Section~\ref{sub:general_uniform_structured_hamiltonian}, we further consider a sum of $d$-sparse Hamiltonians and extend the parallel quantum walk and the class of uniform-structured Hamiltonians.

\subsection{A Quantum Walk for Hamiltonians}
\label{sub:quantum_walk_for_hamiltonians}

Let $H$ be a $d$-sparse $N\times N$ Hamiltonian acting on $n$ qubits (i.e.\ $N=2^n$).
By analogy with the classical Markov chain,
the Hermitian $H$ can be seen as a transition matrix with ``complex probability''
on a \mbox{$d$-sparse} undirected graph whose adjacency matrix is given by replacing each nonzero entry in $H$ with $1$.
Following~\cite{CK11}, this graph is called \textit{the graph of the Hamiltonian},
which we will often denote as $\mathsf{H}$ in the serif font throughout this paper.
We write $(j,k)\in \mathsf{H}$ if an undirected edge $(j,k)$ exists in the graph $\mathsf{H}$.
Following Childs' extension~\cite{Childs09, BC12, Kothari14, BCK15, CKS17} of Szegedy's quantum walk~\cite{Szegedy04},
we define for all $j\in [N]$ the post-transition state of $\ket{j}$:
\begin{equation}
	\ket{\psi_j}:=\frac{1}{\sqrt{d}}\sum_{\parens*{j,k}\in \mathsf{H}}\parens*{\sqrt{H_{jk}^*}\ket{k}+\sqrt{1-\abs*{H_{jk}}}\ket{k+N}}\label{eq:psi_j},
\end{equation}
as a generalization of the classical random walk,
where $z^*$ stands for the complex conjugate of $z$,
and the square root $\sqrt{H_{jk}^*}$ is chosen such that $\sqrt{H_{jk}^*}\parens*{\sqrt{H_{jk}}}^* = H_{jk}^*$.
Note that $\ket{\psi_j}\in \Co^{2N}$ because $N\leq k+N\leq 2N-1$ for $k\in [N]$.
The factor $1/\sqrt{d}$ and the garbage states $\ket{k+N}$ are introduced
to keep $\ket{\psi_j}$ a normalized state.
Now we are ready to define the quantum walk,
a unitary acting on the extended space $\Co^{2N}\otimes \Co^{2N}$
for the $N\times N$ Hamiltonian $H$.

\begin{definition}[Quantum walk for Hamiltonians~\cite{CKS17}]
	\label{def:quantum walk for H}
	Given Hamiltonian $H$ as above.
	Let $\calH=\Co^{2N}\otimes\Co^{2N}$ be the state space.
	For each $j\in [N]$, define
	\begin{equation}
	    \ket{\Psi_j}:=\ket{j}\otimes\ket{\psi_j}=\frac{1}{\sqrt{d}}\sum_{\parens*{j,k}\in \mathsf{H}}\ket{j}\parens*{\sqrt{H_{jk}^*}\ket{k}+\sqrt{1-\abs*{H_{jk}}}\ket{k+N}},
	    \label{eq:Psi_j}
	\end{equation}
	where $\ket{\psi_j}$ is defined in \eqref{eq:psi_j}.
	Let $T:\calH\rightarrow\calH$ be any unitary such that
	\begin{equation}
		T\parens*{\ket{j}\otimes\ket{z}}=
		\begin{cases}
			\ket{\Psi_j},&{\rm if}\  j\in[N]\ {\rm and}\  z=0,\\
			\text{any state,}&o.w.
		\end{cases}
		\label{eq:T}
	\end{equation}
	for all $j,z \in [2N]$.
	Let $S:\calH\rightarrow\calH$ be the SWAP operator such that
	$S\ket{a,b}=\ket{b,a}$ for all $a,b\in [2N]$.
	Then a step of quantum walk for $H$ is defined as $Q:=S\parens*{2\Pi_T-\bbone}$,
	where $\Pi_T=T\parens*{\ket{0}\!\bra{0}\otimes \bbone_{n+1}}T^\dagger$ is a projector,
	with $\ket{0}\in \Co^{2N}$ and $\bbone_{n+1}$ the identity operator acting on $n+1$ qubits.
\end{definition}

The following lemma from~\cite{CKS17} shows that we can implement a polynomial of $H$ by multiple steps of quantum walk. 
More precisely, $r$ iterations of $Q$ block-encodes a $r$-degree Chebyshev polynomial of $H$.

\begin{lemma}
	\label{lmm:r step quantum walk}
	$T^\dagger Q^r T$ is a $(1, n+2, 0)$-block-encoding of $\calT_r(H/d)$,
	if $T$ in \eqref{eq:T} is performed precisely,
	where $\calT_r(x)$ is the $r$-degree Chebyshev polynomial (of the first kind).
\end{lemma}

The proof of Lemma~\ref{lmm:r step quantum walk} involves some interesting techniques,
which were later used in~\cite{LC19} to develop qubitization.
Here we only give a proof for the special case when $r=1$
(thus $T^\dagger QT= T^\dagger ST$), which provides a basis for our generalization to parallel quantum walks in Section~\ref{sub:parallel_version}.

\begin{corollary}
	\label{cor:TdaggerST}
	$T^\dagger S T$ is a $(1, n+2, 0)$-block-encoding of $H/d$,
	if $T$ in \eqref{eq:T} is performed precisely.
\end{corollary}

\begin{proof}
	To see
	\begin{equation*}
		\parens*{\bbone\otimes \bra{0}^{\otimes n+2}}T^\dagger S T\parens*{\bbone \otimes \ket{0}^{\otimes n+2}} = H/d,
	\end{equation*}
	it suffices to prove
	\begin{equation}
		\parens*{\bra{j}\otimes \bra{0}^{\otimes n+2}}T^\dagger S T \parens*{\ket{l}\otimes \ket{0}^{\otimes n+2}} = H_{jl}/d.
		\label{eq:TdaggerST}
	\end{equation}
	We first write the state $\ket{\Psi_j}$ in \eqref{eq:Psi_j} into two parts $\ket{\Psi_j}=\ket{\Phi_j}+\ket{\Phi_j^\bot}$,
	where the subnormalized states
	\begin{equation*}
		\ket{\Phi_j}:=\frac{1}{\sqrt{d}}\sum_{\parens*{j,k}\in \mathsf{H}}\ket{j}\sqrt{H_{jk}^*}\ket{k} \quad\text{and}\quad
		\ket{\Phi_j^\bot}:=\frac{1}{\sqrt{d}}\sum_{\parens*{j,k}\in \mathsf{H}}\ket{j}\sqrt{1-\abs*{H_{jk}}}\ket{k+N}
	\end{equation*}
	with $\ket{j},\ket{k}\in \Co^{2N}$.
	It is easy to verify that
	\begin{align*}
		&\bra{\Phi_j}S\ket{\Phi_l}=H_{jl}/d\\
		&\bra{\Phi_j}S\ket{\Phi_l^{\bot}}=\bra{\Phi_j^\bot}S\ket{\Phi_l}=\bra{\Phi_j^\bot}S\ket{\Phi_l^\bot}=0
	\end{align*}
	for all $j,l\in [N]$. Thus the LHS of \eqref{eq:TdaggerST} is $\bra{\Psi_j}S\ket{\Psi_l}=H_{jl}/d$.
\end{proof}

Now we briefly illustrate how to implement one step of quantum walk by a quantum circuit.
For the sake of simplicity we only consider query complexity,
and assume each step is performed precisely.
It suffices to show how to implement the unitary $T$ (and thus $T^\dagger$),
because in $Q=ST\parens*{2\ket{0}\!\bra{0}\otimes \bbone_{n+1}-\bbone_{2n+2}}T^\dagger$
only the operator $T$ requires oracle queries.
As the ordinary quantum walk does not assume any special structure of the sparse Hamiltonian $H$
(e.g., assume $H$ is local),
we take $\calO_P=\calO_L$ here.

\begin{lemma}[Query complexity of $T$~\cite{BCK15}]
	\label{lmm:complexity of T}
	The unitary $T$ can be implemented by a quantum circuit with $O(1)$ queries to $\calO_H$ and $\calO_L$.
\end{lemma}
\begin{proof}
	Let $\calH^A\otimes \calH^B$ be the state space
	with $\calH^A=\calH^B=\Co^{2N}$.
	As in the definition \eqref{eq:T} of $T$,
	we only consider its action on the initial state $\ket{j}\otimes\ket{0}$ for $j\in [N]$,
	with $\ket{j},\ket{0}\in\Co^{2N}$.
	Then $T$ can be implemented in the following way:

	\begin{enumerate}
		\item 
			Prepare a uniform superposition over computational basis states of size $d$ in $\calH^B$.
		\item
			Query the oracle $\calO_L$ to obtain in $\calH^B$ a superposition over nonzero entries in row $j$.
		\item
			Query the oracle $\calO_H$ to compute $H_{jk}$ in an ancilla space,
			conditioned on which rotates the state in $\calH^B$,
			then uncompute $H_{jk}$ by reverse computation.
	\end{enumerate}
	That is,
	\begin{align*}
		\ket{j}\ket{0}&\xrightarrow{1} \ket{j}\frac{1}{\sqrt{d}}\sum_{t\in [d]}\ket{t}\\
					  &\xrightarrow{2} \ket{j} \frac{1}{\sqrt{d}}\sum_{(j,k)\in \mathsf{H}}\ket{k}\\
					  &\xrightarrow{3} \ket{j} \frac{1}{\sqrt{d}}\sum_{(j,k)\in \mathsf{H}}\ket{k} \ket{H_{jk}}\\
					  &\xrightarrow{3} \ket{j} \frac{1}{\sqrt{d}}\sum_{(j,k)\in \mathsf{H}}\parens*{\sqrt{H_{jk}^*}\ket{k}+\sqrt{1-\abs*{H_{jk}}}\ket{k+N}}\ket{H_{jk}}\\
					  &\xrightarrow{3} \ket{j} \ket{\psi_j}.
	\end{align*}
\end{proof}

\subsection{Parallelization}
\label{sub:parallel_version}

The algorithm for the quantum walk described in Section~\ref{sub:quantum_walk_for_hamiltonians} is highly sequential,
because $r$ steps of quantum walk need $r$ iterations of $Q$,
which in total requires $\Theta(r)$ sequential queries to $\calO_H$ and $\calO_L$. 
Now we define a parallel quantum walk,
which can be implemented by a quantum circuit of only constant depth w.r.t.\ queries for uniform-structured Hamiltonians (to be defined in Section~\ref{sub:pre_walk}). 
Slightly abusing the notation, we denote $\bj\in \mathsf{H}^r$
if a path $\bj=\parens*{j_0,\ldots,j_r}$\footnote{In this paper, a vector is written in a bold font. The coordinates will be indexed from $0$, for example, $\ba=(a_0,\ldots,a_{k-1})$, where the dimension $k$ is indicated by the context.} of length $r+1$ exists in the graph $\mathsf{H}$.

\begin{definition}[$r$-parallel quantum walk for Hamiltonians]
	\label{def:parallel quantum walk}
	Given Hamiltonian $H$ as above.
	Let $\calH=\calH^A\otimes \calH^B$ be the state space, where $\calH^A=\calH^B=\parens*{\Co^{2N}}^{\otimes (r+1)}$. For each $j_0\in [N]$, 
	define
	\begin{equation}
		\ket{\Psi_{j_0}^{(r)}}:=
		\frac{1}{\sqrt{d^r}}\sum_{\bj\in \mathsf{H}^r}
		\underbrace{\vphantom{\ket{j_0}\bigotimes_{s\in [r]}\parens*{\sqrt{H_{j_sj_{s+1}}^*}\ket{j_{s+1}}+\sqrt{1-\abs*{H_{j_sj_{s+1}}}}\ket{j_{s+1}+N}}}
		\ket{\bj}}_{\in\calH^A}
		\otimes
		\underbrace{\ket{j_0} \bigotimes_{s\in [r]}\parens*{\sqrt{H_{j_sj_{s+1}}^*}\ket{j_{s+1}}+\sqrt{1-\abs*{H_{j_sj_{s+1}}}}\ket{j_{s+1}+N}}}_{\in \calH^B}
		\label{eq:Psi_j^r}
	\end{equation}
	where $\ket{\bj}=\ket{j_0}\ldots\ket{j_r}\in \calH^A$.
	Let: \begin{itemize}\item $T^{(r)}:\calH \rightarrow \calH$
	be any unitary operator such that
	\begin{equation}
		T^{(r)}\parens*{\ket{j}\ket{z_1}\ldots\ket{z_{2r+1}}}=
		\begin{cases}
			\ket{\Psi_{j}^{(r)}},
			&j\in [N],z_1=\ldots=z_{2r+1}=0,\\
			\text{any state},
			&o.w.
		\end{cases}
		\label{eq:T^r}
	\end{equation}
	for all $j,z_1,\ldots,z_{2r+1}\in [2N]$; \item  $S^{(r)}:\calH\rightarrow\calH$ be the reverse order operator such that
        \begin{equation*}
            S^{(r)}\ket{a_0,\ldots,a_{2r+1}}=\ket{a_{2r+1},\ldots, a_0}
        \end{equation*}
	for all $a_s\in [2N], s\in [2r+2]$.\end{itemize} Then 
	a step of $r$-parallel quantum walk for $H$ is defined as $Q^{(r)}:=T^{(r)\dagger}S^{(r)}T^{(r)}$.
\end{definition}

The parallel quantum walk defined above naturally generalizes the original quantum walk in Definition~\ref{def:quantum walk for H}.
The key idea is that we extend the state $\ket{\Psi_j}$ in \eqref{eq:Psi_j} which is a superposition of one step of walk $\parens*{j,k}\in \mathsf{H}$, to the state $\ket{\Psi_{j_0}^{(r)}}$ in \eqref{eq:Psi_j^r} which is a superposition of $r$ steps of walk (i.e., a path) $\bj\in \mathsf{H}^r$. 
As shown later in  Lemma~\ref{lmm:T^rdaggerS^rT^r block encode (H/d)^r}, the walk operator $Q^{(r)}$ becomes a block-encoding of the monomial $\parens*{H/d}^r$,  
a generalization of the $H/d$ obtained from the original quantum walk. 
(In this sense, the walk operator $Q^{(r)}$ is an extension of $T^\dagger Q T$ instead of $Q$.)
It should be emphasized that an $r$-parallel quantum walk is not equivalent to $r$ sequential steps of the original quantum walk,
which instead block-encodes a Chebyshev polynomial $\calT_r\parens*{H/d}$.

\begin{remark}
The term ``parallel quantum walk'' comes from the fact
that, as proved later,
the walk operator $Q^{(r)}$
can be performed by a parallel quantum circuit with a constant query depth
if the Hamiltonian $H$ is uniform-structured (see Definition \ref{def:1-uniform-structured}  in Section~\ref{sub:pre_walk}).
This result is non-trivial, because the state $\ket{\Psi_{j_0}^{(r)}}$ in Definition~\ref{def:parallel quantum walk}
contains a dependence chain (which has a sequential nature) induced by the path $\bj\in \mathsf{H}^r$,
where $j_{s+1}$ depends on $j_{s}$ for all $l\in [r]$.
This difficulty is resolved by observing that queries to oracle $\calO_H$ can be actually separated from the dependence chain (see Section~\ref{sub:re_weight}),
while queries to oracle $\calO_P$ can be parallelized if the graph $\mathsf{H}$ has a good structure (see Section~\ref{sub:pre_walk}).
\end{remark}

Now we will illustrate why $Q^{(r)}$ is a block-encoding of $\parens{H/d}^r$.
Similar to the proof of Lemma~\ref{cor:TdaggerST},
we can write $\ket{\Psi_{j_0}^{(r)}}$ into two parts
$\ket{\Psi_{j_0}^{(r)}}=\ket{\Phi_{j_0}^{(r)}}+\ket{\Phi_{j_0}^{(r)\bot}}$,
where the subnormalized state
\begin{equation}
	\ket{\Phi_{j_0}^{\parens*{r}}}:=\frac{1}{\sqrt{d^r}}\sum_{\bj\in \mathsf{H}^r}
	\ket{\bj}\otimes
	\sqrt{H_{j_0j_1}^*\ldots H_{j_{r-1}j_r}^*}\ket{\bj}
	\label{eq:Phi_j}
\end{equation}
represents the ``good'' part of $\ket{\Psi_{j_0}^{(r)}}$.
The following lemma shows some orthogonal relations between these subnormalized states in the context of $S^{(r)}.$

\begin{lemma}
	\label{lmm:Phi and Phi^bot}
	 For all $j,k\in[N]$, we have:
	\begin{enumerate}
		\item 
			 $\bra{\Phi_j^{(r)}}S^{(r)}\ket{\Phi_k^{(r)}}=\parens*{\parens*{H/d}^r}_{jk}$.
		\item
 $\bra{\Phi_j^{(r)}}S^{(r)}\ket{\Phi_k^{(r)\bot}}=
	\bra{\Phi_j^{(r)\bot}}S^{(r)}\ket{\Phi_k^{(r)\bot}}=0$.
	\end{enumerate}
\end{lemma}
\begin{proof}
	We prove the two cases separately.
	\begin{enumerate}
		\item
			Let $j_0:=j$ and $k_0:=k$. By straightforward calculation we have:
			\begin{align*}
				S^{(r)}\ket{\Phi_k^{(r)}}&=\frac{1}{\sqrt{d^r}} \sum_{\bk\in \mathsf{H}^r}
				\sqrt{H_{k_0k_1}^*\ldots H_{k_{r-1}k_r}^*}\ket{k_r\ldots k_0}\ket{k_r\ldots k_0}\\
				\bra{\Phi_j^{(r)}}&=\frac{1}{\sqrt{d^r}}\sum_{\bj\in \mathsf{H}^r}
				\sqrt{H_{j_0j_1}\ldots H_{j_{r-1}j_r}}\bra{j_0\ldots j_r}\bra{j_0\ldots j_r}.
			\end{align*}
			Then using the self-adjointness of $H$ we obtain:
			\begin{equation*}
				\bra{\Phi_j^{(r)}}S^{(r)}\ket{\Phi_k^{(r)}}
				=\frac{1}{d^r}\sum_{\bj\in \mathsf{H}^r}H_{j_0j_1}\ldots H_{j_{r-1}j_r}
				=\parens*{\frac{H^r}{d^r}}_{jk}.
			\end{equation*}
		\item
			Recall that our state space is $\calH=\calH^A\otimes \calH^B$ with $\calH^A=\calH^B=\parens*{\Co^{2N}}^{\otimes r+1}$.
			Let us focus on the space $\calH^A$.
			Since $S^{(r)}$ is the reverse order operator,
			every computational basis component of
			the state $S^{(r)}\ket{\Phi_k^{(r)\bot}}$ has at least one subsystem $s\in [r+1]$ of the form
			$\ket{k_s+N}\in \Co^{2N}$,
			while every computational basis component of $\ket{\Phi_j^{(r)}}$ or $\ket{\Phi_j^{(r)\bot}}$ 
			has all subsystems of the form $\ket{j_t}$ for $t\in [r+1]$.
			The orthogonality statement immediately follows from $\left\langle j_t\middle\vert k_s+N \right\rangle=0$ for any $j_t,k_s\in [N]$.
	\end{enumerate}
\end{proof}

\begin{lemma}
	\label{lmm:T^rdaggerS^rT^r block encode (H/d)^r}
	$Q^{(r)}=T^{(r)\dagger} S^{(r)} T^{(r)}$ is a $(1,2rn+2r+n+2,\epsilon)$-block-encoding of $(H/d)^r$,
	if $T^{(r)}$ is implemented to precision $\epsilon/2$.
\end{lemma}
\begin{proof}
	We first show
	\begin{equation}
		\parens*{\bbone \otimes \bra{0}^{\otimes 2rn+2r+n+2}}T^{(r)\dagger}S^{(r)}T^{(r)}
		\parens*{\bbone \otimes \ket{0}^{\otimes 2rn+2r+n+2}}=(H/d)^r
		\label{eq:TdaggerST=(H/d)^r}
	\end{equation}
	for precise $T^{(r)}$.
	This part is similar to the proof of Lemma~\ref{cor:TdaggerST},
	and it suffices to show
	\begin{equation}
		\bra{\Psi_j^{(r)}}S^{(r)}\ket{\Psi_l^{(r)}} = \bra{j}(H/d)^r\ket{l}
		\label{eq:TdaggerST=(H/d)^r element}
	\end{equation}
	for all $j,l\in [N]$, where $\ket{\Psi_j^{(r)}}$ is defined as above.
	Equation \eqref{eq:TdaggerST=(H/d)^r element} can be obtained by splitting $\ket{\Psi_j^{(r)}}=\ket{\Phi_j^{(r)}}+\ket{\Phi_j^{(r)\bot}}$
	and then applying Lemma~\ref{lmm:Phi and Phi^bot}.

	For approximated $T^{(r)}$ with precision $\epsilon/2$,
	by linearity of error bound propagation, 
	the LHS of \eqref{eq:TdaggerST=(H/d)^r} is approximated to precision $\epsilon$.
\end{proof}

In order to implement an $r$-parallel quantum walk $Q^{(r)}$,
we only need to focus on the $T^{(r)}$ part,
since $S^{(r)}$ can be trivially implemented by a quantum circuit of constant depth using SWAP gates.
The outline of an implementation of $T^{(r)}$ is presented in Figure~\ref{fig:implementation of T^r}.
\begin{figure}[!ht]
	\centering
	\begin{tcolorbox}[colback=white,colframe=black,boxrule=0.2mm,arc=0mm]
		\subsubsection*{State space}
		$\calH=\parens*{\bigotimes_{s\in [r+1]}\calH_s^A}\otimes\parens*{\bigotimes_{s\in [r+1]}\calH_s^B}$, where $\calH_s^A=\calH_s^B=\Co^{2N}$ for all $s$.
		\subsubsection*{Input}
		Any state $\ket{j_0}\otimes \ket{0}\otimes \ldots \otimes \ket{0}$ for $j_0\in [N]$ (due to the definition of $T^{(r)}$ in \eqref{eq:T^r}),
		with $\ket{j_0},\ket{0}\in \Co^{2N}$.
		\subsubsection*{Output}
		The state $\ket{\Psi_{j_0}^{(r)}}$ defined in \eqref{eq:Psi_j^r}.

		$T^{(r)}$ can be implemented in the following ways.
		\subsubsection*{Pre-walk}
		\begin{enumerate}
			\item 
				\label{stp:1 in implementation of T^r}
				Prepare in the subspace $\calH^A$ a pre-walk state
				\begin{equation}
					\ket{p_{j_0}^{(r)}}:=
					\frac{1}{\sqrt{d^r}}\sum_{\bj\in \mathsf{H}^r}\ket{\bj}\in \calH^A.
					\label{eq:pre-walk state}
				\end{equation}
		\end{enumerate}
		\subsubsection*{Re-weight}
		\begin{enumerate}[resume]
			\item
				\label{stp:2 in implementation of T^r}
				Copy the computational basis states in $\calH^A$ to $\calH^B$;
				that is, apply $\mathrm{COPY}^{\parens*{r+1}\cdot\parens*{n+1}}$
				to obtained the state
				\begin{equation*}
					\frac{1}{\sqrt{d^r}}\sum_{\bj\in \mathsf{H}^r}\ket{\bj}\ket{\bj}
					\in \calH^A\otimes\calH^B.
				\end{equation*}
			\item
				\label{stp:3 in implementation of T^r}
				Query $r$ copies of the oracle $\calO_H$ in parallel,
				each in the subspace $\calH_s^A\otimes\calH_{s+1}^B$ for $s\in [r]$.\footnotemark{}
				Similar to the proof of Lemma~\ref{lmm:complexity of T},
				for each query we compute $H_{j_sj_{s+1}}$ in a temporary ancilla space,
				conditioned on which rotates the state in $\calH^B_{s+1}$,
				then uncompute $H_{j_sj_{s+1}}$.
				Finally we obtain the goal state $\ket{\Psi_{j_0}^{(r)}}$.
		\end{enumerate}
	\end{tcolorbox}
	\caption{Implementation of $T^{(r)}$.}
	\label{fig:implementation of T^r}
\end{figure}

Note that the implementation 
consists of two stages: pre-walk and re-weight,
where the two oracles $\calO_H$ and $\calO_P$ are queried separately ---
$\calO_P$ is only queried in the pre-walk stage,
while $\calO_H$ is only queried in the re-weight stage.

\subsubsection{Pre-walk and Uniform-structured Hamiltonians}
\label{sub:pre_walk}

\footnotetext{Here we exploit the two copies of $\ket{\bj}$ to parallelize the queries to the oracle $\calO_H$.
The idea of using multiple copies of data is intuitive and ubiquitous in classical parallel computing.
An alternative way to achieve the parallelization is: given a single copy of $\ket{\bj}\in \calH^A$,
first query $\ceil{r/2}$ copies of $\calO_H$ in parallel, each in the subspace $\calH_{s}^A\otimes \calH_{s+1}^A$ for even $s$;
and then query $\floor{r/2}$ copies of $\calO_H$ in parallel, each in $\calH_{s}^A\otimes \calH_{s+1}^A$ for odd $s$.
This technique is used later for proving Corollary~\ref{cor:log depth path for 1-uniform-structured}.
In this alternative way, one can also modify Definition~\ref{def:parallel quantum walk} 
(and relevant statements, including Lemma~\ref{lmm:Phi and Phi^bot})
such that \eqref{eq:Phi_j} only contains one copy of $\ket{\bj}$.
}

Now we give a detailed description of  pre-walk. At the same time, we introduce a class of Hamiltonians --- uniform-structured Hamiltonians,
for which the pre-walk can be conducted in a parallel fashion.
The state $\ket{p_{j_0}^{(r)}}=\frac{1}{\sqrt{d^r}}\sum_{\bj\in \mathsf{H}^r}\ket{\bj}$ in \eqref{eq:pre-walk state} earns the name ``pre-walk state''
because it is a superposition of all paths generated by $r$ steps of \textit{unweighted random walk} on the graph $\mathsf{H}$ starting from the vertex $j_0$.
We call the process of generating $\ket{p_{j_0}^{(r)}}$ an \textit{$r$-pre-walk}.
For simplicity, we assume $\ket{p_{j_0}^{(r)}}\in \parens*{\Co^{N}}^{\otimes r+1}$,
that is, $\ket{j_s}\in \Co^{N}$ for all $\ket{j_s}$.

For the pre-walk, we only need to focus on the graph $\mathsf{H}$ and the oracle $\calO_P$ that characterizes its sparse structure,
as $\ket{p_{j_0}^{(r)}}$ does not involve any weight $H_{jk}$.
One remaining question is what kind of oracle $\calO_P$ to be used in our algorithm.
Since the complexities that we consider are  measured in terms of the query complexity with respect to $\calO_P$ and gate complexity,
 for practical reasons, $\calO_P$ should be reasonably efficiently implementable.
Conversely, if $\calO_P$ is powerful enough (thus hard to implement),
then intuitively the pre-walk can be done with only a few queries to $\calO_P$.
For instance, given an oracle $\calO_P=\calO_{\mathit{path}}$ 
that directly gives the path generated from walks according to a sequence of choices,
as shown in the following lemma,
then the query complexity is $O\parens*{1}$.
Recall that $L\parens*{j,t}$ denotes the column index of the $t^{\text{th}}$ nonzero entry in row $j$ of $H$,
i.e., the $t^{\text{th}}$ neighbor of vertex $j$ in the graph $\mathsf{H}$.

\begin{lemma}[Pre-walk with a strong path oracle $\calO_{\mathit{path}}$]
	Let $\calO_P=\calO_{\mathit{path}}$ give
	a path generated by $r$ steps of walk starting from $j_0$,
	according to the sequence of choices $\bt\in [d]^{r}$;
	that is, take $\calX=[N],\calY=[N]^r$,
	and $P\parens*{j_0,\bt}=(j_1,\ldots,j_r)$ with $j_{s+1}:=L(j_s,t_s)$ for $s\in [r]$.
	Then the $r$-pre-walk can be implemented by a quantum circuit of
	\begin{itemize}
		\item 
			size $O\parens*{1}$ w.r.t.\ queries to $\calO_P$, and
		\item
			depth $O\parens*{1}$ and size $O\parens*{r\log d}$ w.r.t.\ gates.
	\end{itemize}
\end{lemma}
\begin{proof}
	Assume $d$ is a power of two w.l.o.g.
	Starting from the initial state $\ket{j_0}\ket{0}^{\otimes r}$,
	the pre-walk can be implemented in the following way.
	First prepare a superposition $\frac{1}{\sqrt{d^r}}\sum_{\bt\in [d]^r}\ket{\bt}$ in the second register
	by an $O\parens*{1}$-depth and $O\parens*{r\log d}$-size quantum circuit using Hadamard gates,
	then query $\calO_P$ to compute $P(j_0,\bt)=(j_1,\ldots,j_r)$ and thus obtain the goal state $\ket{p_{j_0}^{(r)}}$.
\end{proof}

In general, it might be expensive to implement $\calO_{\mathit{path}}$.
For example, given the oracle $\calO_L$ that computes the function $L$,
the straightforward way to implement the strong oracle $\calO_{\mathit{path}}$ requires $r$ sequential queries to $\calO_L$.
However, for a special class of Hamiltonians,
generating a path according to a sequence of choices can be done more efficiently
by exploiting parallelism in computing compositions of the function $L$.
This forms the basic idea of  uniform-structured Hamiltonians.
Let function $L^{(r)}:[N]\times [d]^{r}\rightarrow [N]$ be inductively defined as
\begin{equation*}
	L^{(r)}\parens*{j,t_0,\ldots,t_{r-1}}:=L\parens*{L^{(r-1)}\parens*{j,t_0,\ldots,t_{r-2}},t_{r-1}}
\end{equation*}
for $j\in [N],\bt\in [d]^{r}$,
with $L^{(1)}:=L$.
Note that $L^{(r)}$ gives the destination of $r$ steps of walk according to a sequence of choices.
Uniform-structured Hamiltonians are a class of Hamiltonians for which the function $L^{(r)}$ can be computed efficiently in parallel.

\begin{definition}[Uniform-structured Hamiltonian]
	\label{def:1-uniform-structured}
	A $d$-sparse Hamiltonian $H$ with the associated oracle $\calO_P$ is uniform-structured if:
	\begin{itemize}
		\item
			For all $r\in \N$,
			the corresponding $L^{(r)}$ can be expressed as
			\begin{equation}
				L^{(r)}\parens*{j,\bt}=f\parens*{j,g(t_0)\circ \ldots \circ g(t_{r-1})}
				\label{eq:1-uniform-structured L^r}
			\end{equation}
			where the function $f,g$ and the operator $\circ$ with input/output lengths $O\parens*{n}$ satisfy that:
			\begin{itemize}
				\item 
					$f$ and the mapping $\ket{0}\mapsto \frac{1}{\sqrt{d}}\sum_{t\in [d]}\ket{g\parens*{t}}$
					are arithmetic-depth-efficiently computable (see Definition~\ref{def:depth-efficient computation}) with $O\parens*{1}$ queries to $\calO_P$;
				\item
					$\circ$ is associative and arithmetic-depth-efficiently computable.
			\end{itemize}
		\item
			There exists an ``inverse'' function $L^{(-1)}$ such that
			$L^{(-1)}\parens*{j,L(j,t)}=g(t)$
			for all $j\in [N],t\in [d]$,
			and $L^{(-1)}$ is arithmetic-depth-efficiently computable with $O\parens*{1}$ queries to $\calO_P$.
	\end{itemize}
\end{definition}

\begin{remark}
	\label{rem:depth efficient of g implies the mapping}
	One might notice in the first condition that the expression $g(t_0)\circ\ldots\circ g(t_{r-1})$ is ready to be computed in parallel
	by Lemma~\ref{lmm:sequence of associative operators},
	which is actually the key to implement a parallel pre-walk.
	We also point out that if the function $g$ and its inverse are both arithmetic-depth-efficiently computable,
	then the mapping $\ket{0}\mapsto \frac{1}{\sqrt{d}}\sum_{t\in [d]}\ket{g(t)}$ is also arithmetic-depth-efficiently computable
	by evaluating $g$ in place on the state $\frac{1}{\sqrt{d}}\sum_{t\in [d]}\ket{t}$,
	which can be prepared by applying Hadamard gates in parallel.
	
    In the second condition, the inverse function $L^{(-1)}$ actually enables an efficient garbage cleaning, as shown later in Corollary~\ref{cor:log depth path for 1-uniform-structured}.
\end{remark}

For a better understanding of the quite involved Definition~\ref{def:1-uniform-structured}, we show three examples of uniform-structured Hamiltonians.
The first example is a band Hamiltonian,
which has its nonzero entries concentrated within a band around the diagonal.

\begin{lemma}[Band Hamiltonian]
	\label{lmm:band Hamiltonian is uniform-structured}
	Assume $d\in [N]$ is odd.
	Let $H$ be a $d$-band Hamiltonian,
	i.e., $H_{jk}=0$ if $k\notin \calB_j^d$,
	where $\calB_j^d:=\braces*{j+t-\parens*{d-1}/2:t\in [d]}$
	with the addition and subtraction in $\Z_N$.\footnotemark{}
	Let $\calO_P$ be an empty oracle,
	that is, take $\calX=\calY=\emptyset$ and $P$ to be undefined.
	Then $H$ is uniform-structured.
	\footnotetext{This allows a slightly more general definition than the usual band Hamiltonian,
		because in $\Z_N$ when the band has its center close to $0$ it will wrap back from the $N-1$ side.
	For instance, taking $H_{03}=H_{30}=1$ in Example~\ref{eg:band Hamiltonian} $H$ is still $3$-band.}
\end{lemma}
\begin{example}
	\label{eg:band Hamiltonian}
	The $4\times 4$ Hamiltonian $H$ with matrix form
	\begin{equation*}
		H:=
		\begin{bmatrix}
			1&i&0&0\\
			-i&2&3&0\\
			0&3&-1&-i\\
			0&0&i&1
		\end{bmatrix}
	\end{equation*}
	is $3$-band.
\end{example}
\begin{proof}[Proof of Lemma~\ref{lmm:band Hamiltonian is uniform-structured}]
	Note that for band Hamiltonians we do not need to query $\calO_P$.
	The lemma is proved by verifying the conditions in Definition~\ref{def:1-uniform-structured}.
Note that 	$L(j,t)=j+t-\parens*{d-1}/2$,
	where the addition and subtraction are in $\Z_N$ 
	(although in this case the nonzero entries are not necessarily ordered,
	the correctness of the algorithm is unaffected).
	\begin{itemize}
		\item 
			We have
			\begin{equation*}
				L^{(r)}\parens*{j,\bt}=j+\parens*{t_0-\parens*{d-1}/2}+\ldots+\parens*{t_{r-1}-\parens*{d-1}/2};
			\end{equation*}
			that is, take $f(j,k)=j+k$,
			$g(t)=t-\parens*{d-1}/2$,
			and $\circ$ to be addition in \eqref{eq:1-uniform-structured L^r}.
			\begin{itemize}
				\item 
					By Lemma~\ref{lmm:parallel quantum circuit for elementary arithmetics},
					$f$ is arithmetic-depth-efficiently computable.
					As in Remark~\ref{rem:depth efficient of g implies the mapping},
					the mapping $\ket{0}\mapsto \frac{1}{\sqrt{d}}\sum_{t\in [d]}\ket{g(t)}$ is arithmetic-depth-efficiently computable too.
				\item
					The addition is obviously associative; and is arithmetic-depth-efficiently computable by Lemma~\ref{lmm:parallel quantum circuit for elementary arithmetics}.
			\end{itemize}
		\item
			Take the inverse function to be $L^{(-1)}(j,k)=k-j$,
			which is arithmetic-depth-efficiently computable by Lemma~\ref{lmm:parallel quantum circuit for elementary arithmetics}.
	\end{itemize}
\end{proof}

The second example is a tensor product of Pauli matrices.
Recall that any Hamiltonian can be expressed as a sum of (scaled) tensor products of Pauli matrices,
which form a basis for the Hermitian space. 
In Section~\ref{sub:general_uniform_structured_hamiltonian}, we will further show that this Pauli sum is also uniform-structured (according to the extended definition).

\begin{lemma}[Tensor product of Pauli matrices]
	\label{lmm:tensor product of Pauli matrices}
	Let $H$ be a (scaled) tensor product of Pauli matrices,
	that is, $H=\alpha\bigotimes_{k\in [n]}\sigma_k$ with $\sigma_k\in \braces*{\bbone,X,Y,Z}$ and $\alpha$ being a constant.
	Let $\calO_P$ give an \mbox{$n$-bit} string $s$ characterizing the Pauli string $\angles*{\sigma_k}_{k\in [n]}$.
	In particular, take $\calX=[1]$, $\calY=[N]$ and $P\parens*{x,y}=y\oplus s$,
	where the $k^{\text{th}}$ bit of $s$ is defined as
	\begin{equation*}
		s_k:=
		\begin{cases}
			0,&\text{$\sigma_k\in \braces*{\bbone, Z}$ (diagonal),}\\
			1,&\text{$\sigma_k\in \braces*{X,Y}$ (off-diagonal).}
		\end{cases}
	\end{equation*}
	Then $H$ is uniform-structured.
\end{lemma}
\begin{proof}
	Note that the sparsity $d=1$ for $H$, because all Pauli matrices are $1$-sparse.
	The lemma is proved by verifying the conditions in Definition~\ref{def:1-uniform-structured}.
It holds that 	$L(j,t)=j\oplus s$ where $\oplus$ is the bit-wise XOR operator.
	\begin{itemize}
		\item
			We have
			\begin{equation*}
				L^{(r)}\parens*{j,\bt}=j\oplus s\oplus \ldots\oplus s;
			\end{equation*}
			that is, take $f\parens*{j,k}=j\oplus k$, $g(t)=s$,
			and $\circ$ to be $\oplus$ in \eqref{eq:1-uniform-structured L^r}.
			\begin{itemize}
				\item
					$f$ is arithmetic-depth-efficiently computable,
					while the mapping $\ket{0}\mapsto\frac{1}{\sqrt{d}}\sum_{t\in [d]}\ket{g(t)}$
					is trivial to perform with a single query to $\calO_P$ because $d=1$.
				\item
					The XOR $\oplus$ is obviously associative and arithmetic-depth-efficiently computable.
			\end{itemize}
		\item
			Take the inverse function to be $L^{(-1)}\parens*{j,k}=s$,
			which can be computed by a single query to $\calO_P$.
	\end{itemize}
\end{proof}

The third example is a local Hamiltonian term,
which acts non-trivially on a subsystem of $l$ qubits,
whose positions are indicated by the $l$ bits of $1$ in an $n$-bit string $s$.
The sum of many local Hamiltonian terms is a local Hamiltonian,
which will be investigated in Section~\ref{sub:general_uniform_structured_hamiltonian} as a uniform-structured Hamiltonian (according to the extended definition).

\begin{lemma}[Local Hamiltonian term]
	\label{lmm:local Hamiltonian term is uniform-structured}
	Let $H$ be an $l$-local Hamiltonian term;
	that is, $H=H_s\otimes \bbone_{\bar{s}}$,
	where $H_s$ is a Hamiltonian acting on the subsystem of $l$ qubits whose positions are indicated
	by $l$ bits of $1$ in the $n$-bit string $s$,
	and $\bbone_{\bar{s}}$ is the identity operator on the subsystem of the rest $n-l$ qubits.
	Let $\calO_P$ give the locality parameter $s$.
	In particular, set $\calX=[1]$, $\calY=[N]$ and $P\parens*{x,y}=y\oplus s$.
	Then $H$ is uniform-structured.
\end{lemma}
\begin{example}
	Let $H:=A\otimes \bbone \otimes B\otimes \bbone$ be a $16\times 16$ Hamiltonian
	with $A$ and $B$ being $2\times 2$ Hamiltonians,
	and $\bbone$ being $2\times 2$ identity matrix.
	Then $H$ is a $2$-local Hamiltonian term $H_s\otimes \bbone_{\bar{s}}$,
	with $H_s=A\otimes B$ and $s=1010$.
\end{example}
\begin{proof}[Proof of Lemma~\ref{lmm:local Hamiltonian term is uniform-structured}]
	The lemma is proved by verifying the conditions in Definition~\ref{def:1-uniform-structured}.
	We use the superscript $i$ to denote the $i^{\text{th}}$ bit of a number.
	Note that the sparsity $d=2^l$, $L(j,t)=j\vartriangleleft_s \parens*{t\restriction_s}$,
	where the operator $\restriction_s:[d]\rightarrow [N]$ lifts an $l$-bit string to an $n$-bit string
	according to $s$, defined as
	\begin{equation*}
		\parens*{b\restriction_s}^i:=s^i\cdot b^{s^0+\ldots +s^i},
	\end{equation*}
	and the operator $\vartriangleleft_s:[N]\times [N]\rightarrow [N]$ overwrites an $n$-bit string by another according to $s$, defined as
	\begin{equation}
		(a\vartriangleleft_s b)^i:= a^i(1-s^i)+b^is^i
		\label{eq:vartriangleleft}
	\end{equation}
	for all $a,b\in [N],i\in [n]$.
	For instance, $101\restriction_{01011}=01001$ and $10011\vartriangleleft_{01011}01001=11001$.
	
	Technically,
	we further define two operators $\vartriangleleft:[N]\times [N]^2\rightarrow [N]$ and $\vartriangleleft^\vee:[N]^2\times [N]^2\rightarrow [N]^2$
	such that
	\begin{equation}
    	a\vartriangleleft (b,x):=a\vartriangleleft_x b,\qquad (a,x)\vartriangleleft^\vee(b,y):=(a\vartriangleleft_y b,x\vee y)
    	\label{eq:vartriangleleft2}
	\end{equation}
	for all $a,b,x,y\in [N], i\in [n]$, where $\vee$ is the bit-wise OR operator.
	\begin{itemize}
		\item 
			We have
			
			\begin{equation*}
				L^{(r)}\parens*{j,\bt}=j\vartriangleleft \parens*{t_0\restriction_s,s} \vartriangleleft^\vee \ldots \vartriangleleft^\vee
				\parens*{t_{r-1}\restriction_s, s};
			\end{equation*}
			that is, take $f\parens*{j,\parens*{k,x}}=j\vartriangleleft \parens*{k,x}$,
			$g(t)=\parens*{t\restriction_s,s}$,
			and $\circ$ to be $\vartriangleleft^\vee$ in \eqref{eq:1-uniform-structured L^r}.
			\begin{itemize}
				\item 
					To compute $f$,
					it suffices to compute the operator $\vartriangleleft$,
					which according to \eqref{eq:vartriangleleft} and \eqref{eq:vartriangleleft2} 
					is arithmetic-depth-efficiently computable.

					The mapping $\ket{0}\mapsto \frac{1}{\sqrt{d}}\sum_{t\in [d]}\ket{g(t)}$ can be performed 
					by first querying $\calO_P$ to obtain $\ket{s}$ in an ancilla,
					then conditioned on it implementing $n$ controlled Hadamard gates in parallel,
					assuming $d$ is a power of two w.l.o.g.
					This is arithmetic-depth-efficient with $O\parens*{1}$ queries.
				\item
				    
					The associativity of $\vartriangleleft^\vee$ is easy to verify
					by noting that $a\vartriangleleft_x b\vartriangleleft_y c=a\vartriangleleft_x \parens*{b\vartriangleleft_y c}$
					for all $a,b,c,x,y\in [N]$ and the associativity of $\vee$.
					Also, $\vartriangleleft^\vee$ is arithmetic-depth-efficiently computable by \eqref{eq:vartriangleleft2}.
			\end{itemize}
		\item
			Take the inverse function to be $L^{(-1)}(j,k)=\parens*{k\wedge s,s}$, where $\wedge$ is the bit-wise AND operator.
			This is arithmetic-depth-efficiently computable with $O\parens*{1}$ queries to $\calO_P$.
	\end{itemize}
\end{proof}

Now we are ready to present the parallel pre-walk subroutine for uniform-structured Hamiltonians.
As aforementioned, the parallelism relies on the structure of $L^{(r)}$,
which will be shown to be computable by a logarithmic depth quantum circuit with respect to $r$.
The goal is achieved through several lemmas and corollaries.

\begin{lemma}
	\label{lmm:log depth L^r for 1-uniform-structured}
	For a uniform-structured Hamiltonian $H$,
	the mapping
	\begin{equation*}
		\ket{j}\ket{g\parens*{\bt}}\ket{z}\mapsto \ket{j}\ket{g\parens*{\bt}}\ket{z\oplus L^{(r)}\parens*{j,\bt}}
	\end{equation*}
	for all $j,z\in [N],\bt\in [d]^{r}$ with $g\parens*{\bt}=\parens*{g\parens*{t_0},\ldots,g\parens*{t_{r-1}}}$,
	can be implemented by a quantum circuit of
	\begin{itemize}
		\item
			size $O\parens*{1}$ w.r.t.\ queries to $\calO_P$, and
		\item 
			depth $O\parens*{\log r\cdot \log^2 n}$ and size $O\parens*{r n^4}$ w.r.t.\ gates.
	\end{itemize}
\end{lemma}
\begin{proof}
	Since $\circ$ is associative and arithmetic-depth-efficiently computable,
	by Lemma~\ref{lmm:sequence of associative operators}
	we can first evaluate
	$g(t_0)\circ\ldots \circ g(t_{r-1})$
	by an $O\parens*{\log r\cdot \log^2 n}$-depth and $O\parens*{rn^4}$-size quantum circuit.
	Also $f$ is arithmetic-depth-efficiently computable with $O\parens*{1}$ queries to $\calO_P$,
	thus $L^{(r)}\parens*{j,\bt}=f\parens*{j,g(t_0)\circ\ldots\circ g(t_{r-1})}$
	can be computed by additional depth $O\parens*{\log^2 n}$ and size $O\parens*{n^4}$ w.r.t.\ gates and $O\parens*{1}$ queries.
	Finally apply $\mathrm{COPY}^n$ to copy the result into $\ket{z}$,
	followed by garbage cleaning using  reverse computation.
	The total  complexity follows from summing these complexities up.
\end{proof}

Note that for different $s\in [r+1]$, the function $L^{(s)}$ can be computed in parallel,
then we obtain the following corollary.

\begin{corollary}
	\label{cor:log depth L^r for 1-uniform-structured parallel}
	For a uniform-structured Hamiltonian $H$,
	the mapping
	\begin{equation}
		\ket{j}\ket{g\parens*{\bt}}\ket{\bz}\mapsto\ket{j}\ket{g\parens*{\bt}}\bigotimes_{s=1}^{r}\ket{z_{s-1}\oplus L^{(s)}\parens*{j,t_0,\ldots,t_{s-1}}}
		\label{eq:log depth L^r for 1-uniform-structured parallel}
	\end{equation}
	for all $j\in [N],\bt\in [d]^r, \bz\in [N]^r$,
	can be implemented by a quantum circuit of
	\begin{itemize}
		\item
			depth $O\parens*{1}$ and size $O\parens*{r}$ w.r.t.\ queries to $\calO_P$, and
		\item 
			depth $O\parens*{\log r\cdot\log^2 n}$ and size $O\parens*{r^2n^4}$ w.r.t.\ gates.
	\end{itemize}
\end{corollary}
\begin{proof}
	First make $r$ copies of $\ket{j}\ket{g(\bt)}$
	by applying parallel $\mathrm{COPY}_{r}^{O(rn)}$ in Lemma~\ref{lmm:parallel copy},
	which has depth $O\parens*{\log r}$ and size $O\parens*{r^2n}$.
	Then apply Lemma~\ref{lmm:log depth L^r for 1-uniform-structured}
	to each of the $r$ copies,
	for the $s^{\text{th}}$ copy taking the input $\parens*{j,g\parens*{t_0},\ldots,g\parens*{t_{s-1}}}$
	to compute $L^{(s)}\parens*{j,t_0,\ldots,t_{s-1}}$,
	which can be done in parallel for $s=1$ to $r$.
	Finally apply $\mathrm{COPY}^n$ in parallel to store the result,
	followed by garbage cleaning using reverse computation.
	The final complexity comes from summing up these complexities.
\end{proof}

With the help of the inverse function $L^{(-1)}$,
we can generate a path state $\ket{\bj}$ from the state $\ket{j_0,g\parens*{\bt}}$
by erasing the redundant information $\ket{g(\bt)}$ in \eqref{eq:log depth L^r for 1-uniform-structured parallel},
as shown in the following corollary.

\begin{corollary}
	\label{cor:log depth path for 1-uniform-structured}
	For a uniform-structured Hamiltonian $H$,
	the mapping
	\begin{equation*}
		\ket{j_0,g\parens*{\bt}}\mapsto\ket{\bj}
	\end{equation*}
	for all $j_0\in [N],\bt\in [d]^r$,
	where $\bj\in \mathsf{H}^r$ has $j_{s+1}:=L(j_s,t_s)$ for $s\in [r]$,
	can be implemented by a quantum circuit of
	\begin{itemize}
		\item
			depth $O\parens*{1}$ and size $O\parens*{r}$ w.r.t.\ queries to $\calO_P$, and
		\item 
			depth $O\parens*{\log r\cdot \log^2 n}$ and size $O\parens*{r^2n^4}$ w.r.t.\ gates.
	\end{itemize}
\end{corollary}
\begin{proof}
	Since $j_s=L^{(s)}\parens*{j_0,t_0,\ldots,t_{s-1}}$ for $s=1$ to $r$,
	to perform the required mapping
	we only need to clean the $\ket{g\parens*{\bt}}$ in \eqref{eq:log depth L^r for 1-uniform-structured parallel}
	after applying Corollary~\ref{cor:log depth L^r for 1-uniform-structured parallel} to $\ket{j_0,g\parens*{\bt}}\ket{0}^{\otimes r}$.
	Recall that $L^{(-1)}(j_s,j_{s+1})=g\parens*{t_s}$ for $s\in [r]$,
	so we can compute $L^{(-1)}$ by taking inputs $\ket{j_s,j_{s+1}}$ to clean $g\parens*{t_s}$ first on odd $s$
	then on even $s$, thereby for different $s$ the computation can be done in parallel.
	Since $L^{(-1)}$ is arithmetic-depth-efficiently computable with $O\parens*{1}$ queries to $\calO_P$,
	the cleaning process can be done by a quantum circuit of depth $O\parens*{\log^2 n}$ and size $O\parens*{rn^4}$ w.r.t.\ gates and $O\parens*{1}$ queries.
	The final complexity follows from summing these complexities up.
\end{proof}

A combination of the above results gives a parallel pre-walk algorithm for uniform-structured Hamiltonians.

\begin{lemma}[Pre-walk on uniform-structured Hamiltonians]
	\label{lmm:pre-walk on 1-uniform-structured}
	Let $H$ be a uniform-structured Hamiltonian,
	then an $r$-pre-walk on its graph $\mathsf{H}$,
	i.e., preparing the state $\ket{p_{j_0}^{(r)}}$,
	can be implemented by a quantum circuit of
	\begin{itemize}
		\item
			depth $O\parens*{1}$ and size $O\parens*{r}$ w.r.t.\ queries to $\calO_P$, and
		\item 
			depth $O\parens*{\log r\cdot \log^2 n}$ and size $O\parens*{r^2n^4}$ w.r.t.\ gates.
	\end{itemize}
\end{lemma}
\begin{proof}
	Let $\bigotimes_{s\in [r+1]}\calH_s$ be the state space with $\calH_s =\Co^N$.
	The process of preparing $\ket{p_{j_0}^{(r)}}$ is shown below,
	starting from the initial state $\ket{j_0}\ket{0}^{\otimes r}$
	with $\ket{j_0},\ket{0}\in \Co^N$.
	\begin{enumerate}
		\item
			\label{stp:1 in lmm:pre-walk on 1-uniform-structured}
			Perform the mapping $\ket{0}\mapsto \frac{1}{\sqrt{d}}\sum_{t\in [d]}\ket{g\parens*{t}}$ in each $\calH_s$ for $s=1$ to $r$,
			to obtain the state
			\begin{equation*}
				\frac{1}{\sqrt{d^r}}\ket{j_0}\sum_{\bt\in [d]^r}\ket{g\parens*{\bt}}.
			\end{equation*}
		\item
			Apply Corollary~\ref{cor:log depth path for 1-uniform-structured}
			we obtain the goal state $\ket{p_{j_0}^{(r)}}=\frac{1}{\sqrt{d^r}}\sum_{\bj\in \mathsf{H}^r}\ket{\bj}$.
	\end{enumerate}
	Each mapping $\ket{0}\mapsto \frac{1}{\sqrt{d}}\sum_{t\in [d]}\ket{g\parens*{t}}$
	in Step~\ref{stp:1 in lmm:pre-walk on 1-uniform-structured} is arithmetic-depth-efficient with $O\parens*{1}$ queries to $\calO_P$,
	due to the definition of uniform-structured Hamiltonians.
	Combined with Corollary~\ref{cor:log depth path for 1-uniform-structured} the final complexity is obtained.
\end{proof}

\subsubsection{Re-weight}
\label{sub:re_weight}

Intuitively, the re-weight procedure in the implementation of $T^{(r)}$ in Figure~\ref{fig:implementation of T^r} adjusts the ``weight'' of each path $\ket{\bj}$ in the pre-walk state $\ket{p_{j_0}^{(r)}}$ according to the entries in $H$, given by the oracle $\calO_H$. 
As we will see in the following, the re-weight analysis is in fact simpler than the pre-walk because there is no requirement on the sparse structure of $H$.

\begin{lemma}[Re-weight]
	\label{lmm:walk stage complexity}
	Re-weight of $\ket{p_{j_0}^{(r)}}$,
	i.e., performing the mapping $\ket{p_{j_0}^{(r)}}\mapsto \ket{\Psi_{j_0}^{(r)}}$,
	where $\ket{\Psi_{j_0}^{(r)}}$ is defined in \eqref{eq:Psi_j^r},
	can be implemented to precision $\epsilon$ by a quantum circuit of
	\begin{itemize}
		\item 
			depth $O(1)$ and size $O(r)$ w.r.t.\ queries to $\calO_H^b$ with $b=O\parens*{\log \parens*{1/\epsilon}}$, and
		\item
			depth $O\parens*{\log^2\log \parens*{1/\epsilon}}$ and size $O\parens*{r \bracks*{n+\log^4\parens*{1/\epsilon}}}$ w.r.t.\ gates,
	\end{itemize}
	for $r=\polylog \parens*{1/\epsilon}$.
\end{lemma}
\begin{proof}
	We analyze separately the query complexity and gate complexity of the re-weight stage,
	including Step~\ref{stp:2 in implementation of T^r} and Step~\ref{stp:3 in implementation of T^r}
	in Figure~\ref{fig:implementation of T^r}.
	\begin{itemize}
		\item 
			For query complexity,
			only Step~\ref{stp:3 in implementation of T^r} involves queries to $\calO_H$.
			As $\calH_s^A\otimes\calH_{s+1}^B$ are disjoint for $s\in [r]$,
			these $r$ queries are independent, and thus can be done in parallel
			with depth $O\parens*{1}$ and size $O\parens*{r}$.
			The precision $b$ of oracle $\calO_H^b$ will be determined in the gate complexity analysis below.
		\item
			For gate complexity,
			in Step~\ref{stp:2 in implementation of T^r} the $\mathrm{COPY}^{\parens*{r+1}\cdot\parens*{n+1}}$ gate can be implemented by an $O(1)$-depth and $O(rn)$-size quantum circuit.
			In Step~\ref{stp:3 in implementation of T^r},
			one needs to apply $r$ controlled rotations conditioned on some $H_{jk}$;
			that is, perform the mapping
			\begin{equation}
				\ket{0}\ket{H_{jk}}\mapsto \parens*{\sqrt{H_{jk}^*}\ket{0}+\sqrt{1-\abs*{H_{jk}}}\ket{1}}\ket{H_{jk}}.
				\label{eq:H rotation}
			\end{equation}
			To achieve a total precision $\epsilon$ of $T^{(r)}$,
			each rotation in \eqref{eq:H rotation} needs to be $O(\epsilon/r)$-precise,
			which requires $H_{jk}$ given by the oracle $\calO_H^b$ to have
			$b=O\parens*{\log\parens*{r/\epsilon}}=O\parens*{\log\parens*{1/\epsilon}}$ bits of precision.
			To perform the rotation, first compute $\sqrt{\parens*{1-\abs*{H_{jk}}}/H_{jk}^*}$ and its arctangent
			arithmetic-depth-efficiently by Lemma~\ref{lmm:parallel quantum circuit for elementary arithmetics},
			then apply $\mathrm{C}_b$-$R_Y$ in Corollary~\ref{lmm:parallel controlled Z-rotation}.

			Following~\cite{BC12},
			to satisfy the condition $\sqrt{H_{jk}^*}\parens*{\sqrt{H_{jk}}}^*=H_{jk}^*$
			one should be careful in choosing the sign of $\sqrt{H_{jk}^*}$ for $H_{jk}<0$.
			This problem is addressed by adding an $O\left(2^{-b}\right)$ disturbance on the imaginary part of $H_{jk}$
			to force it to be nonzero (thereby forcing $H_{jk}$ to be complex) for those $H_{jk}<0$,
			with the total precision unchanged up to a constant factor.\footnotemark{}
			Now for $H_{jk}=re^{i\theta}$, let $\sqrt{H_{jk}^*}:=\sqrt{r}e^{i\theta/2}$,
			which is also arithmetic-depth-efficiently computable by Lemma~\ref{lmm:parallel quantum circuit for elementary arithmetics}.
			From Corollary~\ref{lmm:parallel controlled Z-rotation},
			a $\mathrm{C}_b$-$R_Y$ gate can be implemented by an $O\parens*{\log b}$-depth and $O\parens*{b}$-size quantum circuit.
			The total complexity follows by summing up the complexities of $r$ rotations.
	\end{itemize}
	\footnotetext{In~\cite{BC12}, to avoid this sign ambiguity,
		they assign different signs to $\sqrt{H_{jk}^*}$ for $H_{jk}$ above and below the diagonal,
		and replace $H$ with $H+\norm*{H}_{\max}\bbone$ to force the diagonal elements to be non-negative.
		For Hamiltonian simulation problem, this modification works well by only introducing a global phase.
	    Here we provide an alternative solution that is not restricted to the Hamiltonian simulation task.}
\end{proof}

Finally, combining the pre-walk complexity (Lemma~\ref{lmm:pre-walk on 1-uniform-structured} in Section~\ref{sub:pre_walk}), the re-weight complexity (Lemma~\ref{lmm:walk stage complexity} in the above),
and the negligible complexity of implementing $S^{(r)}$
gives the total complexity of the parallel quantum walk 
for uniform-structured Hamiltonians.

\begin{theorem}[Parallel quantum walks for uniform-structured Hamiltonians]
    \label{thm:parallel quantum walks for uniform-structured Hamiltonian}
	For a uniform-structured Hamiltonian $H$,
	the $r$-parallel quantum walk $Q^{(r)}$ can be performed to precision $\epsilon$ by a quantum circuit of
	\begin{itemize}
		\item 
			depth $O(1)$ and size $O(r)$ w.r.t.\ queries to $\calO_H^b$ with $b=O\parens*{\log \parens*{1/\epsilon}}$,
		\item
			depth $O\parens*{1}$ and size $O\parens*{r}$ w.r.t.\ queries to $\calO_P$, and
		\item
			depth $O\parens*{\log r\cdot \log^2 n+\log^2\log \parens*{1/\epsilon}}$ and size $O\parens*{r^2n^4 +r\log^4\parens*{1/\epsilon}}$ w.r.t.\ gates,
	\end{itemize}
	for $r=\polylog\parens*{1/\epsilon}$.
\end{theorem}

\subsection{Extension: A Parallel Quantum Walk for a Sum of Hamiltonians}
\label{sub:general_uniform_structured_hamiltonian}

As shown in the previous section,
the parallel quantum walk can be efficiently implemented in parallel for the class of uniform-structured Hamiltonians,
which are however somewhat restricted in applications.
Now we extend the framework in Section~\ref{sub:parallel_version} to a parallel quantum walk for \textit{a sum of Hamiltonians},
that is, a Hamiltonian of the form $H=\sum_{w\in [m]} H_w$,
where $H_w$ are $d$-sparse Hamiltonians and $m=\poly\parens*{n}$.
In this extended framework,
we generalize the class of uniform-structured Hamiltonians
to include more Hamiltonians of practical interest, like Pauli sums and local Hamiltonians.

Recall that in Section~\ref{sub:parallel_version} the good sparse structure of a Hamiltonian $H$ is a key to
efficiently implement the parallel quantum walk.
The intuition behind the extended framework in this section is then:
if some Hamiltonians has the same type of good sparse structures,
then a sum of them is still structured well enough for exploiting parallelism.
The organization of this section is similar to  Section~\ref{sub:parallel_version}.
For readability, we only provide the essential definitions and lemmas here,
and leave more details to Appendix~\ref{sec:details_of_section_extension}.

Let us first define the extended parallel quantum walk for a sum of Hamiltonians. 
Recall that the state $\ket{\Psi_{j_0}^{(r)}}$ in \eqref{eq:Psi_j^r},
which is a superposition of paths $\bj=(j_0,\ldots,j_{r-1})$ in the graph $\mathsf{H}$,
is a key ingredient in the  parallel quantum walk in the previous section.
Now in the case of a sum Hamiltonian $H=\sum_w H_w$,
the graph $\mathsf{H}$ can be ``decomposed'' into a sum of subgraphs $\mathsf{H}_w$,
so in a path $\bj\in \mathsf{H}$ each edge $\parens*{j_s,j_{s+1}}$ belongs to at least one subgraph.
Thus, we can define an extended state $\ket{\Psi_{j_0}^{(r,m)}}$,
which is still a superposition of paths $\bj\in \mathsf{H}$, but each tensored with a corresponding string $\bw$,
such that $\parens*{j_s,j_{s+1}}\in H_{w_s}$ for all $s\in [r]$.
In this way, the extended parallel quantum walk can better exploit the sum structure of the Hamiltonian $H=\sum_{w\in [m]} H_w$,
as shown later.

\begin{definition}[$(r,m)$-parallel quantum walk]
	\label{def:(r,m)-parallel quantum walk}
	Given the Hamiltonian $H$ as above.
	Let $\calH=\calH^W\otimes \calH^A\otimes\calH^B$ be the walk space,
	where $\calH^W=\parens*{\Co^{m}}^{\otimes r}$ and $\calH^A=\calH^B=\parens*{\Co^{2N}}^{\otimes r+1}$. For each $j_0\in [N]$,
	define $\ket{\Psi_{j_0}^{(r,m)}}:=$
	\begin{equation}
		\frac{1}{\sqrt{(md)^r}}
		\sum_{\bw\in [m]^r}
		\sum_{\bj \in \mathsf{H}^{\bw}}
		\underbrace{\vphantom{\ket{j_0}\bigotimes_{s\in [r]}\parens*{\sqrt{\tilde{H}_{j_sj_{s+1}}^*}\ket{j_{s+1}}+\sqrt{1-\abs*{\tilde{H}_{j_sj_{s+1}}}}\ket{j_{s+1}+N}}}\ket{\bw}}_{\in \calH^W}\otimes
		\underbrace{\vphantom{\ket{j_0}\bigotimes_{s\in [r]}\parens*{\sqrt{\tilde{H}_{j_sj_{s+1}}^*}\ket{j_{s+1}}+\sqrt{1-\abs*{\tilde{H}_{j_sj_{s+1}}}}\ket{j_{s+1}+N}}}\ket{\bj}}_{\in \calH^A}\otimes
		\underbrace{\ket{j_0}\bigotimes_{s\in [r]}\parens*{\sqrt{\tilde{H}_{j_sj_{s+1}}^*}\ket{j_{s+1}}+\sqrt{1-\abs*{\tilde{H}_{j_sj_{s+1}}}}\ket{j_{s+1}+N}}}_{\in \calH^B}
		\label{eq:Psi_j^(r,m)}
	\end{equation}
	where $\bj\in \mathsf{H}^{\bw}$ denotes $\parens*{j_s,j_{s+1}}\in \mathsf{H}_{w_s}$ for all $s\in [r]$,
	and $\tilde{H}_{jk}:=H_{jk}/c(j,k)$
	with $c(j,k):=\sum_{w\in [m]} \bbracks*{(j,k)\in \mathsf{H}_w}$\footnote{Here $\bbracks{\cdot}$ stands for the Iverson bracket; that is, $\bbracks{p}=1$ if $p$ is true and $\bbracks{p}=0$ if $p$ is false.} the number of subgraphs containing the edge $\parens*{j,k}\in \mathsf{H}$.
	Let 
	\begin{itemize}
	    \item 
        	$T^{(r,m)}:\calH\rightarrow\calH$ be any unitary operator such that 
        	\begin{equation*}
        		T^{(r,m)}\parens*{\ket{\bw}\ket{j}\ket{\bz}}=
        		\begin{cases}
        			\ket{\Psi_{j}^{(r,m)}},
        			&j\in [N],\bw=0,\bz=0,\\
        			\text{any state},
        			&o.w.
        		\end{cases}
        	\end{equation*}
        	for all $\bw\in [m]^r,j\in [2N],\bz\in [2N]^{2r+1}$.
        \item
        	$S^{(r,m)}:\calH\rightarrow \calH$ reverses the order in the subspace $\calH^A\otimes \calH^B$;
        	that is, $S^{(r,m)}=\bbone^W\otimes S^{(r)}$,
	        where $S^{(r)}:\calH^A\otimes\calH^B\rightarrow\calH^A\otimes\calH^B$ is the reverse order operator in Definition~\ref{def:parallel quantum walk}.
	\end{itemize}
	Then a step of $(r,m)$-parallel quantum walk for $H$ is defined as $Q^{(r,m)}:=T^{(r,m)\dagger}S^{(r,m)}T^{(r,m)}$.
\end{definition}

\begin{remark}
    \label{rem:calO_tildeH} 
    Note that in \eqref{eq:Psi_j^(r,m)}, the amplitudes are determined by a new Hamiltonian $\tilde{H}$,
    which is entry-wise rescaled from $H$.
    The rescaling factor for each entry $H_{jk}$ is $c(j,k)$,
    i.e., the number of overlapping subgraphs $H_w$ on the edge $(j,k)\in \mathsf{H}$.
    For later implementation of the extended parallel quantum walk,
    we will consider a new oracle $\calO_{\tilde{H}}$ giving an entry of $\tilde{H}$
    such that $\calO_{\tilde{H}}\ket{j,k,z}=\ket{j,k,z\oplus \tilde{H}_{jk}}$. 
    Note that if $c(j,k)$ can be efficiently computed,
    $\calO_{\tilde{H}}$ can be easily constructed from $\calO_H$,
    with the total precision scaled by at most $m$ for the construction,
    as $c\parens*{j,k}=\sum_{w\in [m]}\bbracks{\parens*{j,k}\in H_w}\leq m$.
    The overhead caused by this precision scaling will be shown to be negligible later.
\end{remark}

The extended parallel quantum walk defined above also block-encodes a monomial of $H$,
as shown in the following lemma.

\begin{lemma}
	\label{lmm:(r,m)-parallel quantum walk}
	$Q^{(r,m)}=T^{(r,m)\dagger}S^{(r,m)}T^{(r,m)}$ is a $(1,r\ceil{\log m}+2rn+2r+n+2,\epsilon)$-block-encoding of $\parens*{\frac{H}{md}}^r$,
	if $T^{(r,m)}$ is implemented to precision $\epsilon/2$.
\end{lemma}
\begin{proof}
    Postponed to Appendix~\ref{sec:details_of_section_extension}.
\end{proof}

As a straightforward generalization of the implementation of $T^{(r)}$ in Figure~\ref{fig:implementation of T^r}, an implementation of $T^{(r,m)}$ is shown in Figure~\ref{fig:implementation of T^(r,m)}.
\begin{figure}[!ht]
	\centering
	\begin{tcolorbox}[colback=white,colframe=black,boxrule=0.2mm,arc=0mm]
		\subsubsection*{State space}
		$\calH=\parens*{\bigotimes_{s\in [r]}\calH_s^W}\otimes\parens*{\bigotimes_{s\in [r+1]}\calH_s^A}\otimes\parens*{\bigotimes_{s\in [r+1]}\calH_s^B}$
		with $\calH_s^W=\Co^{m}$ and $\calH_s^A=\calH_s^B=\Co^{2N}$ for all $s$.
		\subsubsection*{Input}
		Any state $\ket{0}^{\otimes mr}\ket{j_0}\ket{0}^{\otimes 2N\parens*{2r+1}}$ for $j_0\in [N]$,
		with $\ket{j_0}\in \Co^{2N},\ket{0}\in \Co$.
		\subsubsection*{Output}
		The state $\ket{\Psi_{j_0}^{(r,m)}}$ defined in \eqref{eq:Psi_j^(r,m)}.

		$T^{(r,m)}$ can be implemented in the following ways.
		\subsubsection*{Pre-walk}
		\begin{enumerate}
			\item 
				\label{stp:1 in implementation of T^{(r,m)}}
				Prepare in the subspace $\calH^W\otimes \calH^A$ a pre-walk state
				\begin{equation}
					\ket{p_{j_0}^{(r,m)}}:=
					\frac{1}{\sqrt{\parens*{md}^r}}\sum_{\bw\in [m]^r}\sum_{\bj\in \mathsf{H}^{\bw}}\ket{\bw}\ket{\bj}\in \calH^W\otimes\calH^A.
					\label{eq:extended pre-walk state}
				\end{equation}
		\end{enumerate}
		\subsubsection*{Re-weight}
		\begin{enumerate}[resume]
			\item
				\label{stp:2 in implementation of T^{(r,m)}}
				Copy the computational basis states in $\calH^A$ to $\calH^B$;
				that is, apply $\mathrm{COPY}^{\parens*{r+1}\cdot\parens*{n+1}}$
				to obtained the state
				\begin{equation*}
					\frac{1}{\sqrt{\parens*{md}^r}}\sum_{\bw\in [m]^r}\sum_{\bj\in \mathsf{H}^{\bw}}\ket{\bw}\ket{\bj}\ket{\bj}
					\in \calH.
				\end{equation*}
			\item
				\label{stp:3 in implementation of T^{(r,m)}}
				Query $r$ copies of the modified oracle $\calO_{\tilde{H}}$ in parallel,
				each in the subspace $\calH_s^A\otimes\calH_{s+1}^B$ for $s\in [r]$.
				For each query we compute $\tilde{H}_{j_sj_{s+1}}$ in a temporary ancilla space,
				conditioned on which rotates the state in $\calH^B_{s+1}$,
				then uncompute $\tilde{H}_{j_sj_{s+1}}$ with another query.
				Finally we obtain the goal state $\ket{\Psi_{j_0}^{(r,m)}}$.
		\end{enumerate}
	\end{tcolorbox}
	\caption{Implementation of $T^{(r,m)}$}
	\label{fig:implementation of T^(r,m)}
\end{figure}
We call the procedure of preparing the (extended) pre-walk state $\ket{p_{j_0}^{(r,m)}}$ in \eqref{eq:extended pre-walk state} an $\parens*{r,m}$-pre-walk.
Now we will introduce the notion of $m$-uniform-structured Hamiltonians as an extension of the uniform-structured Hamiltonians in Section~\ref{sub:parallel_version}.
We first redefine the function $L:[m]\times[N]\times [d]\rightarrow [N]$ such that 
$L\parens*{w,j,t}$ is the $t^{\text{th}}$ neighbor of vertex $j$ in subgraph $\mathsf{H}_w$,
and let $L^{(r)}:[m]^r\times[N]\times[d]^r\rightarrow [N]$ be inductively defined as
\begin{equation*}
	L^{(r)}\parens*{\bw,j,\bt}:=L\parens*{w_{r-1},L^{(r-1)}\parens*{w_0,\ldots,w_{r-2},j,t_0,\ldots,t_{r-2}},t_{r-1}}
\end{equation*}
for $\bw\in [m]^r,j\in [N],\bt\in [d]^r$,
with $L^{(1)}:=L$.

\begin{definition}[$m$-uniform-structured Hamiltonian]
	\label{def:m-uniform-structured}
	A sum Hamiltonian $H=\sum_{w\in [m]}H_w$ with the associated oracle $\calO_P$ is $m$-uniform-structured if:
	\begin{itemize}
		\item
			For all $r\in \N$,
			its corresponding $L^{(r)}$ can be expressed as
			\begin{equation}
				L^{(r)}\parens*{\bw,j,\bt}=f\parens*{j,g(w_0,t_0)\circ \ldots \circ g(w_{r-1},t_{r-1})}
				\label{eq:m-uniform-structured L^r}
			\end{equation}
			where the function $f,g$ and the operator $\circ$ with input/output lengths $O\parens*{n}$ satisfy that:
			\begin{itemize}
				\item 
					$f$ and the mapping $\ket{w}\ket{0}\mapsto \ket{w}\frac{1}{\sqrt{d}}\sum_{t\in [d]}\ket{g(w,t)}$
					are arithmetic-depth-efficiently computable with $O(1)$ queries to $\calO_P$ for all $w\in [m]$;
				\item
					$\circ$ is associative and arithmetic-depth-efficiently computable.
			\end{itemize}
		\item
			There exists an inverse function $L^{(-1)}$ such that
			$L^{(-1)}\parens*{w,j,L(w,j,t)}=g(w,t)$
			for all $w\in [m],j\in [N],t\in [d]$,
			and $L^{(-1)}$ is arithmetic-depth-efficiently computable with $O\parens*{1}$ queries to $\calO_P$.
		\item
			The function $\bbracks{\parens*{j,k}\in \mathsf{H}_w}$ can be arithmetic-depth-efficiently computed with $O(1)$ queries to $\calO_P$,
			given $j,k\in [N],w\in [m]$ as inputs.
	\end{itemize}
\end{definition}

Note that the first two conditions in Definition~\ref{def:m-uniform-structured}
are naturally generalized from Definition~\ref{def:1-uniform-structured},
while the third condition is set to guarantee the efficiency of computing $c\parens*{j,k}=\sum_{w\in [m]}\bbracks{\parens*{j,k}\in H_w}$,
which is used in the construction of the oracle $\calO_{\tilde{H}}$ from $\calO_H$, as mentioned in Remark~\ref{rem:calO_tildeH}.

Now we present two important examples of $m$-uniform-structured Hamiltonians:
Pauli sums and local Hamiltonians,
which are generalizations of tensor products of Pauli matrices and local Hamiltonian terms respectively.

\begin{lemma}[Pauli sum]
	\label{lmm:sum of tensor products of Pauli matrices}
	Let $H$ be a Pauli sum,
	that is, $H=\sum_{w\in [m]}H_w$ where $H_w$ are (scaled) tensor products of Pauli matrices defined in Lemma~\ref{lmm:tensor product of Pauli matrices}.
	Let $\calO_P$ give an $n$-bit string $s(w)$ characterizing the Pauli string of $H_w$ for each $w$.
	In particular,
	set $\calX=[m]$, $\calY=[N]$ and $P\parens*{w,y}=y\oplus s(w)$.
	Then $H$ is $m$-uniform-structured.
\end{lemma}
\begin{proof}
    Postponed to Appendix~\ref{sec:details_of_section_extension}.
\end{proof}

\begin{lemma}[Local Hamiltonian]
	\label{lmm:local Hamiltonian is m-uniform-structured}
	Let $H$ be an $(l,m)$-local Hamiltonian,
	that is, $H=\sum_{w\in [m]} H_w$
	where $H_w$ are $l$-local Hamiltonian terms defined in Lemma~\ref{lmm:local Hamiltonian is m-uniform-structured}.
	Let $\calO_P$ give an $n$-bit string $s\parens*{w}$ characterizing the locality of $H_w$,
	i.e., the positions of the $l$ qubits $H_w$ acts on, for each $w$.
	In particular, set $\calX=[m]$, $\calY=[N]$ and $P\parens*{w,y}=y\oplus s\parens*{w}$.
	Then $H$ is $m$-uniform-structured.
\end{lemma}
\begin{proof}
    Postponed to Appendix~\ref{sec:details_of_section_extension}.
\end{proof}

\begin{remark}
	Although any Hamiltonian can be represented as a Pauli sum due to the fact that tensor products of Pauli matrices form a basis for the Hermitian space,
	the number of summands $m$ can be large.
	Since the complexity of the algorithm depends on $m$, 
	only those Pauli sums with small $m$ are of practical interest. The same difficulty exists when one tries to represent any Hamiltonian as a local Hamiltonian,
	because the parameter $l$ can be large.
\end{remark}

Following the same line of analysis as in Section~\ref{sub:parallel_version},
we have the following theorem that the (extended) parallel quantum walk for uniform-structured Hamiltonians in Definition~\ref{def:m-uniform-structured}
can be efficiently implemented in parallel.

\begin{theorem}[Parallel quantum walks for $m$-uniform-structured Hamiltonians]
	\label{thm:complexity of Q^(r,m)}
	For an $m$-uniform-structured Hamiltonian,
	the $(r,m)$-parallel quantum walk $Q^{(r,m)}$ can be implemented to precision $\epsilon$ by a quantum circuit of
	\begin{itemize}
		\item
			depth $O\parens*{1}$ and size $O\parens*{r}$ w.r.t.\ queries to $\calO_H^b$ with $b=O\parens*{\log\parens*{m/\epsilon}}$,
		\item
			depth $O\parens*{1}$ and size $O\parens*{rm}$ w.r.t.\ queries to $\calO_P$, and
		\item
			depth $O\parens*{\log r\cdot \log^2 n+\log^2\log \parens*{m/\epsilon}}$ and size $O\parens*{mr^2n^4+r\log^4\parens*{m/\epsilon}}$ w.r.t.\ gates,
	\end{itemize}
	for $r=\polylog \parens*{1/\epsilon}$.
\end{theorem}

Note that when $m=1$, the above  theorem reduces to Theorem~\ref{thm:parallel quantum walks for uniform-structured Hamiltonian}. For readability, the proof of this  theorem is postponed to  Appendix~\ref{sec:details_of_section_extension}.

\section{Parallel LCU for Hamiltonian Series}
\label{sec:linear_combinations_of_hamiltonian_powers} 

In this section, we show how to implement a linear combination of block-encoded powers of a Hamiltonian
(i.e., a Hamiltonian power series) by a parallel quantum circuit. 
The method is based on the LCU (linear combination of unitaries) algorithm developed through~\cite{SOGKL02,CW12,Kothari14,BCCKS15,BCK15,CKS17,GSLW19}. 
Here we adopt the block-encoding version of LCU in~\cite{GSLW19}.
The results in this section will be used to implement a linear combination of parallel quantum walks
to approximate the evolution unitary $e^{-iHt}$ in Section~\ref{sec:parallel_hamiltonian_simulation}.

We first recall two technical lemmas from~\cite{GSLW19}.
Although they apply to more general matrices,
here for our purpose we restrict them to the Hamiltonians.

\begin{lemma}[Product of block-encoded Hamiltonian powers~\cite{GSLW19}]
	\label{lmm:product of block-encoded Hamiltonian powers}
	Let $K$ be any Hamiltonian. For $n,m\in \N$, if 
	$U$ is an $(\alpha,a,\delta)$-block-encoding of $K^n$ and 
	$V$ is a $(\beta,b,\epsilon)$-block-encoding of $K^m$, 
	then $\parens*{\bbone_a\otimes V}\parens*{\bbone_b\otimes U}$
	is an $\parens*{\alpha\beta,a+b,\alpha\epsilon+\beta\delta}$-block-encoding of $K^{n+m}$, 
	where $\bbone_s$ is the identity operator acting on the proper subsystem composed of the $s$ qubits.
\end{lemma}

\begin{definition}[State preparation unitary]
	\label{def:state preparation unitary}
	Let $\ba\in \Co^{R}$ and $\alpha:=\norm*{\ba}_1$, where $\norm*{\cdot}_1$ is the $l_1$ norm.
	A unitary $V$ is called an $\parens*{\alpha,\epsilon}$-state-preparation-unitary of $\ba$,
	if the state $V\ket{0}$ is $\epsilon$-close\footnotemark{} to $\frac{1}{\sqrt{\alpha}}\sum_{r\in [R]}\sqrt{a_r}\ket{r}$ with $\ket{0}\in \Co^{R}$,
	where $a_r$ is the $r^{\text{th}}$ entry of $\ba$.
	\footnotetext{Here, for simplicity of later proofs, the precision is measured by $l_2$-norm,
	which is different from the one in~\cite{GSLW19}. See Appendix~\ref{sec:state_preparation_for_lcu} for more details.}
\end{definition}

This definition
is a special case of \textit{state preparation pair} in~\cite{GSLW19},
as shown in Appendix~\ref{sec:state_preparation_for_lcu}.

\begin{lemma}[LCU for Hamiltonian series~\cite{GSLW19}]
	\label{lmm:linear combination of block-encoded Hamiltonian powers}
	Let $\ba\in \Co^{R}$ be an $R$-dimensional vector,
	$K$ be a Hamiltonian, and \begin{itemize}
\item	$A:=\sum_{r\in [R]} a_r K^r$ be a power series of $K$;
\item	$V$ be an $(\alpha,\delta)$-state-preparation-unitary of $\ba$; 
\item	$U_r$ be an $(\beta,b,\epsilon)$-block-encoding of $K^r$.\end{itemize}If
	$U:=\sum_{r\in [R]} \ket{r}\!\bra{r}\otimes U_r$,
	then $\parens*{V^\dagger\otimes \bbone}U\parens*{V\otimes \bbone}$ is
	an $\parens*{\alpha\beta,s+b,\alpha\beta R\delta+\alpha\beta\epsilon}$-block-encoding of $A$ with $s:=\ceil{\log R}$.
\end{lemma}

Combining Lemma~\ref{lmm:product of block-encoded Hamiltonian powers} and
Lemma~\ref{lmm:linear combination of block-encoded Hamiltonian powers},
we can obtain the following: 

\begin{corollary}
	\label{cor:linear combination of block-encoded Hamiltonian powers}
	Let $\ba\in \Co^{R}$ be an $R$-dimensional vector,
	$K$ be a Hamiltonian, 
	$s:=\ceil{\log R}$, and 
	\begin{itemize}
	    \item 
        	$A:=\sum_{r\in [R]} a_r K^r$ be a power series of $K$; 
	    \item 
	        $V$ be an $(\alpha,\delta)$-state-preparation-unitary of $\ba$;
        \item	
            $W_r$ be a $(1, b_r,\epsilon)$-block-encoding of $K^r$, and
        	\begin{equation}
        	    W:=\sum_{r\in [R]}\ket{r}\!\bra{r}\otimes \prod_{j\in [s]} \parens*{\bbone_{b-b_{2^j}}\otimes W_{2^j}}^{r_j}
        		\label{eq:W}
        	\end{equation}
	        where $b:=\sum_{j\in [s]}b_{2^j}$ and $r_j$ is the $j^{\text{th}}$ bit of $r$.
    \end{itemize}
	Then $\parens*{V^\dagger\otimes\bbone}W\parens*{V\otimes\bbone}$ is
	an $\parens*{\alpha,s+b,\alpha R \delta+\alpha s\epsilon}$-block-encoding of $A$.
\end{corollary}
\begin{proof}
	Since $W_{2^j}$ is a $(1,b_{2^j},\epsilon)$-block-encoding of $K^{2^j}$,
 by Lemma~\ref{lmm:product of block-encoded Hamiltonian powers}, 
	\begin{equation*}
	       U_r:=\prod_{j\in [s]}\parens*{\bbone_{b-b_{2^j}}\otimes W_{2^j}}^{r_j}
	\end{equation*} 
        is a $\parens*{1,b,s\epsilon}$-block-encoding of $K^r$.
	Then by Lemma~\ref{lmm:linear combination of block-encoded Hamiltonian powers}
	we reach the conclusion.
\end{proof}

Now we present two lemmas showing how to implement $V$ and $W$ in Corollary~\ref{cor:linear combination of block-encoded Hamiltonian powers}
by parallel quantum circuits.
The idea of Lemma~\ref{lmm:parallel state preparation} is similar to a data structure for matrix sampling in~\cite{KP17},
but here we consider complex amplitudes (rather than real amplitudes in~\cite{KP17}) and explicitly compute the gate complexity. This lemma is based on previous works on quantum sampling~\cite{GR02}.

\begin{lemma}[Parallel state preparation]
	\label{lmm:parallel state preparation}
	Let $\ba\in \Co^R$ be an $R$-dimensional vector
	such that for each $r\in [R]$, $a_r$ is arithmetic-depth-efficiently computable given $r$ as input,
	and let $\alpha:=\norm*{\ba}_1$.
	An $(\alpha,\epsilon)$-state-preparation-unitary $V$ of $\ba$ can be implemented by
	an $O\parens*{\log^3\log\parens*{1/\epsilon}}$-depth and $O\parens*{\log^5\parens*{1/\epsilon}}$-size quantum circuit
	for $R=O\parens*{\log \parens*{1/\epsilon}}$.
\end{lemma}
\begin{proof}
	Assume $R=2^s$ w.l.o.g.,
	because one can enlarge the dimension of $\ba$ by appending enough $0$ entries.
	Let $a_r=e^{i\theta_r}\abs*{a_r}$ and partial sum $S(j,l):=\sum_{r=j}^{l-1} \abs*{a_r}$.
	Recall that our goal is to perform the mapping $\ket{0}^{\otimes s}\mapsto \frac{1}{\sqrt{\alpha}}\sum_{r\in [R]}\sqrt{a_r}\ket{r}$.
	The state preparation $V$ consists of $s$ steps,
	and in the $k^{\text{th}}$ step we perform the mapping
	\begin{equation}
        \begin{cases}
            \ket{0}^{\otimes s}\mapsto U_{0,0}\ket{0}\ket{0}^{\otimes s-1}&k=0,\\
            \ket{x}\ket{0}^{\otimes s-k}\mapsto \ket{x}U_{k,x}\ket{0}\ket{0}^{\otimes s-k-1}&1\leq k < s,
        \end{cases}
		\label{eq:parallel state preparation}
	\end{equation}
	for all $x\in [2^{k}]$, where $U_{k,x}$ is a single qubit gate such that
        \begin{equation*}
            U_{k,x}\ket{0}=\sqrt{\gamma_x}\ket{0}+e^{i\beta_x/2}\sqrt{1-\gamma_x}\ket{1},
        \end{equation*}
        and we define
	\begin{equation*}
		\gamma_x:=\frac{S\parens*{u,w}}{S\parens*{u,v}}
		,\qquad
		\beta_x:=\theta_w-\theta_u,
	\end{equation*}
	with $u:=2^{s-k}x$, $v:=2^{s-k}x+2^{s-k}$ and $w:=(u+v)/2$.
	A proof of the correctness of this procedure can be found in~\cite{KP17} and~\cite{GR02},
	except that here we have an additional accumulative phase $\beta_x$ to handle the complex amplitudes.\footnote{An alternative way to handle the complex amplitudes is: in the construction of the state preparation unitary, prepare the state $\frac{1}{\sqrt{\alpha}}\sum_{r\in [R]}\sqrt{\abs*{a_r}}\ket{r}$;
	and in the LCU, add all the phases $e^{i\theta_r}$ for $r\in [R]$ in parallel.}

    Let us compute the gate complexity.
	To achieve a total precision $\epsilon$,
	each step of \eqref{eq:parallel state preparation} needs to be $\parens*{\epsilon/s}$-precise.
	Each step of \eqref{eq:parallel state preparation} consists of two controlled rotations,
	$\mathrm{C}_b$-$R_Y$ and $\mathrm{C}_b$-$R_Z$,
	which are conditioned on elementary arithmetic functions of $\gamma_x$ and $\beta_x$ respectively,
	where $b=O\parens*{\log\parens*{s/\epsilon}}=O\parens*{\log\parens*{1/\epsilon}}$ is the number of bits for the required precision.
	While $\beta_x$ can be easily computed by an arithmetic-depth-efficient quantum circuit,
	as it is an elementary arithmetic function of $a_u$ and $a_w$, which are arithmetic-depth-efficiently computable from $x$;
	$\gamma_x$ needs to be computed from the more complicated $S\parens*{u,v}$.

	Note that one can compute all $S\parens*{j,l}$ required at once in an ancilla space.
	This is done by \begin{itemize}\item first computing $\parens*{2^{-b}/R}$-precise $a_r$ for all $r$
	by arithmetic-depth-efficient circuits
	on the input $\ket{0,1,\ldots,R-1}$,
	due to the assumption that $a_r$ is arithmetic-depth-efficiently computable;
	\item then computing $S(j,l)$ in an inductive way analogous to the proof of Lemma~\ref{lmm:sequence of associative operators},
	except that one needs to create a constant number of copies of each $S(j,l)$ by $\mathrm{COPY}^{b+s}$ for future computation.\end{itemize}
	The input state $\ket{0,1,\ldots,R-1}$ can be prepared by an $O\parens*{1}$-depth and $O\parens*{Rs}$-size quantum circuit,
	while the inductive summation procedure requires depth $O\parens*{s\log^2\parens*{b+s}}$ and size $O\parens*{R\parens*{b+s}^4}$
	by Lemma~\ref{lmm:parallel quantum circuit for elementary arithmetics}, Lemma~\ref{lmm:parallel copy} and the proof of Lemma~\ref{lmm:sequence of associative operators}.
	Thus all $S(j,l)$ required can be computed by an $O\parens*{\log^3\log\parens*{1/\epsilon}}$-depth and $O\parens*{\log^5\parens*{1/\epsilon}}$-size quantum circuit.

	By Lemma~\ref{lmm:parallel controlled Z-rotation} 
	the controlled rotations in each step of \eqref{eq:parallel state preparation}
	can be implemented by an $O\parens*{\log b}$-depth and $O\parens*{b}$-size quantum circuit,
	once the rotation angles,
	some elementary arithmetic functions of $\gamma_x,\beta_x$ are computed 
	by an $O\parens*{\log^2 b}$-depth and $O\parens*{b^4}$-size quantum circuit.
	Summing up the $s$ steps gives the final complexity.
\end{proof}

The next lemma
generalizes Lemma~8 in~\cite{CKS17} to the block-encoding case.
\begin{lemma}[Parallel implementation of $W$]
	\label{lmm:parallel implementation of W}
	The unitary $W$ in \eqref{eq:W} can be implemented by a quantum circuit
	of depth $\sum_{j\in [s]}\mathtt{Dep}\parens*{W_{2^j}}$ and
	size $\sum_{j\in [s]}\mathtt{Siz}\parens*{W_{2^j}}$,
	where $\mathtt{Dep}$ and $\mathtt{Siz}$ refer to the depth and size cost of a subroutine.
\end{lemma}
\begin{proof}
	A parallel quantum circuit implementation of $W$ is shown in Figure~\ref{fig:parallel quantum circuit for W},
	where we omit the identity operator $\bbone_{b-b_{2^j}}$ on the ancilla space for simplicity of notations.
	Since the controlled version of $W_{2^j}$ has the same complexity as $W_{2^j}$ up to a constant factor,
	the final complexity then follows easily.
	\begin{figure}[!ht]
		\centering
		\begin{quantikz}[row sep=0.4cm, column sep=0.9cm]
			\lstick{$\ket{r_0}$}           & \ctrl{4}                     & \qw                          & \qw                              & \qw                              & \qw\\
			\lstick{$\ket{r_1}$}           & \qw                          & \ctrl{3}                     & \qw                              & \qw                              & \qw\\
			\lstick{$\vdots$}              &                              &                              & \ddots                           &                                  & \\
			\lstick{$\ket{r_{s-1}}$}       & \qw                          & \qw                          & \qw                    & \ctrl{1}                         & \qw\\
			\lstick{$\ket{\phi}$}          & \gate[2,bundle={2}]{W_{2^0}} & \gate[2,bundle={2}]{W_{2^1}} & \ \ldots\ \qw                    & \gate[2,bundle={2}]{W_{2^{s-1}}} & \qw\\
			\lstick{$\ket{0}^{\otimes b}$} &                              &                              & \ \ldots\ \qwbundle[alternate]{} &                                  & \qwbundle[alternate]{}
		\end{quantikz}
		\caption{A parallel quantum circuit for $W$.}
		\label{fig:parallel quantum circuit for W}
	\end{figure}
\end{proof}

\section{Parallel Hamiltonian Simulation}
\label{sec:parallel_hamiltonian_simulation}

Now we are ready to present a parallel quantum simulation algorithm for uniform-structured Hamiltonians
by assembling the techniques developed in the previous sections.
Following~\cite{BCK15, BCCKS15, CKS17}, we simulate the evolution unitary $e^{-iHt}$ by first 
splitting the time interval $t$ into small segments each of length $\Delta t$,
then approximating the evolution within each segment
by the truncated Taylor series
\begin{equation}
	e^{-iH\Delta t}\approx \sum_{r=0}^{R-1} \frac{\parens*{-iH\Delta t}^r}{r!}.
	\label{eq:truncated Taylor series}
\end{equation}
In our setting, $\Delta t$ should be chosen such that the monomial $\parens*{H\Delta t}^r$ can be obtained from the parallel quantum walk in Section~\ref{sec:parallel_quantum_walk}.
Using the results of Section~\ref{sec:linear_combinations_of_hamiltonian_powers} to combine these walk operators, 
we get a block-encoding of $e^{-iH\Delta t}$,
then we apply these $e^{-iH\Delta t}$ sequentially on the initial state. In this way we are able to simulate $e^{-iHt}$.
To guarantee a close to $1$ success amplitude after each application of $e^{-iH\Delta t}$,
a technique introduced in~\cite{BCCKS14,BCCKS15,BCK15}
called robust oblivious amplitude amplification will be used.

Recall that we assume $\norm*{H}_{\max}=1$.
Since the monomial $\parens*{\frac{H}{md}}^r$ can be obtained from a parallel quantum walk for $m$-uniform-structured Hamiltonians,
we choose $\Delta t:=\frac{1}{md}$.
By rescaling $H$ with a constant factor $1/2$,
the eigenvalue of $H\Delta t$ lies in the interval $[-1/2,1/2]$.
The following lemma shows that to achieve an $\epsilon$-precision in \eqref{eq:truncated Taylor series}
it suffices to choose $R=\ceil{\log \parens*{1/\epsilon}}$,\footnote{Actually $R=O\parens*{\frac{\log(1/\epsilon)}{\log\log\parens*{1/\epsilon}}}$ is also sufficient (for example, see~\cite{BCCKS15}). 
Nevertheless, for simplicity, here we choose $R=\ceil{\log(1/\epsilon)}$ instead.
According to Lemma~\ref{lmm:parallel Hamiltonian simulation delta t},
this simplified choice only incurs a constant factor
in the depth complexity and a poly-logarithmic factor in the size complexity,
which has a negligible impact on our main results.
}
by taking $z=-iH\Delta t$.

\begin{lemma}
	\label{lmm:choice of R}
	Assume $R:=\ceil{\log \parens*{1/\epsilon}}\geq 4$ w.l.o.g., we have
	\begin{equation*}
		\abs*{e^{z}-\sum_{r=0}^{R-1} \frac{z^r}{r!}}=\abs*{\sum_{r=R}^\infty \frac{z^r}{r!}}\leq \epsilon
	\end{equation*}
	for $z\in \Co$ with $\abs*{z}\leq 1/2$.
\end{lemma}
\begin{proof}
	Since $\sum_{r=R}^\infty \frac{z^r}{r!}$ is absolutely convergent,
	we have $\abs*{\sum_{r=R}^\infty \frac{z^r}{r!}}\leq \sum_{r=R}^\infty \frac{\abs*{z}^r}{r!}$.
	The RHS of the above equation can be bounded by the Taylor remainder
	$\frac{e^{\xi}\abs*{z}^R}{R!}\leq \frac{1}{2^R}\leq \epsilon$,
	where $0\leq \xi\leq \abs*{z}\leq 1/2$.
\end{proof}

Let us analyze the complexity of implementing a block-encoding of $e^{-iH\Delta t}$.

\begin{lemma}
	\label{lmm:parallel Hamiltonian simulation delta t}
	For any $m$-uniform-structured Hamiltonian $H$, there exists a unitary $U_{\Delta t}$ that
	forms an $\parens*{\alpha, R\parens*{2\ceil{\log m}+4n}, \epsilon}$-block-encoding of $e^{-iH\Delta t}$,
	and can be implemented by a quantum circuit of
	\begin{itemize}
		\item 
			depth $O\parens*{\log R}$ and size $O\parens*{R}$
			w.r.t.\ queries to $\calO_H^b$ with $b=O\parens*{R+\log m}$,
		\item
			depth $O\parens*{\log R}$ and size $O\parens*{mR}$ w.r.t.\ queries to $\calO_P$, and
		\item
			depth $O\parens*{\log^2 R\cdot \log^2 n +\log^3 R}$
			and size $O\parens*{mR^2n^4+R\parens*{R+\log m}^4}$ w.r.t.\ gates,
	\end{itemize}
	where $\Delta t:=\frac{1}{md}$, $\alpha < e$ is a constant, and $R=\ceil{\log \parens*{1/\epsilon}}\geq 4$.
\end{lemma}
\begin{proof}
	In Corollary~\ref{cor:linear combination of block-encoded Hamiltonian powers},
	take $R=\ceil{\log \parens*{1/\epsilon}}$, $a_r=\frac{\parens*{-i}^r}{r!}$, $K=\frac{H}{md}$ and $W_r=Q^{(r,m)}$.
	Then $\alpha=\sum_{r\in [R]} \abs*{a_r} < e$ is a constant,
	where each $a_r$ is arithmetic-depth-efficiently computable by Lemma~\ref{lmm:parallel quantum circuit for elementary arithmetics}.
	Assume $R=2^s$ w.l.o.g. To achieve a total $\epsilon$-precision of $U_{\Delta t}$,
	in Corollary~\ref{cor:linear combination of block-encoded Hamiltonian powers}, 
	take the precision of $V$ to be $\epsilon/\parens*{2\alpha R}$,
	and the precision of each $W_r$ (i.e., the precision of $Q^{(r,m)}$ in Theorem~\ref{thm:complexity of Q^(r,m)})
	to be $\epsilon/\parens*{2\alpha s}$.
	By Lemma~\ref{lmm:parallel state preparation},
	the state preparation unitary $V$ can be implemented by an $O\parens*{\log^3\log\parens*{1/\epsilon}}$-depth
	and $O\parens*{\log^5\parens*{1/\epsilon}}$-size quantum circuit.
	Combining Lemma~\ref{lmm:parallel implementation of W} and Theorem~\ref{thm:complexity of Q^(r,m)},
	the implementation of $W$ has
	\begin{itemize}
		\item 
			depth $O\parens*{s}$ and size $O\parens*{R}$ w.r.t.\ queries to $\calO_H^b$ with $b=O\parens*{\log \parens*{sm/\epsilon}}$,
		\item
			depth $O\parens*{s}$ and size $O\parens*{mR}$ w.r.t.\ queries to $\calO_P$, and
		\item
			depth $O\parens*{s\bracks*{\log R\cdot\log^2n+\log^2\log\parens*{sm/\epsilon}}}$
			and size $O\parens*{mR^2n^4+R\log^4\parens*{sm/\epsilon}}$ w.r.t.\ gates.
	\end{itemize}
	The final complexity follows from summing up these complexities and the assumption $m=\poly n$.
	Note that the number of ancilla qubits required for the block-encoding is not a tight upper bound.
\end{proof}

To achieve a constant success amplitude after a sequence of $e^{-iH\Delta t}$,
we need to amplify the success amplitude to be close to 1 after each application of $e^{-iH\Delta t}$
with robust oblivious amplitude amplification.
The following lemma is a special case of Lemma~6 in~\cite{BCK15}.

\begin{lemma}[Robust oblivious amplitude amplification]
	\label{lmm:robust oblivious amplitude amplification}
	Given a unitary $U$ that $\parens*{\alpha,s,\epsilon}$-block-encodes a unitary $V$,
	where $\alpha=\arcsin\parens*{\frac{\pi}{2\parens*{2l+1}}}$ for some $l\in \N$,
	one can construct a $\parens*{1,s,O\parens*{\epsilon}}$-block-encoding of $V$
	with $O\parens*{l}$ uses of $U$, $U^\dagger$ and $\bbone-2\Pi_s$,
	where $\Pi_s:=\ket{0}\!\bra{0}_s$ is a projector on the ancilla space of $s$ qubits.
\end{lemma}

\begin{corollary}
	\label{cor:robust oblivious amplitude amplification}
	Given a unitary $U$ that $\parens*{\beta,s,\epsilon}$-block-encodes a unitary $V$,
	where $\beta=O\parens*{1}$ is a known constant,
	a $\parens*{1,s+1,O\parens*{\epsilon}}$-block-encoding of $V$ can be implemented by a quantum circuit of
	\begin{itemize}
		\item 
			size $O\parens*{1}$ w.r.t.\ $U$ and $U^\dagger$, and
		\item
			depth $O\parens*{\log s+\log^2\log\parens*{1/\epsilon}}$ and size $O\parens*{s+\log^4\parens*{1/\epsilon}}$ w.r.t.\ gates.
	\end{itemize}
\end{corollary}
\begin{proof}
	By the assumption,
	$\norm*{V-\beta\parens*{\bra{0}^{\otimes s}\otimes \bbone}U\parens*{\ket{0}^{\otimes s}\otimes \bbone}}\leq \epsilon$.
	Pick the minimum $l$ such that $\alpha:=\arcsin\parens*{\frac{\pi}{2\parens*{2l+1}}}\geq \beta$.
	Let $W$ be a unitary such that
        \begin{equation*}
            W\ket{0}=\frac{1}{\sqrt{\alpha}}\parens*{\sqrt{\beta}\ket{0}+\sqrt{\alpha^2-\beta^2}\ket{1}},
        \end{equation*}
	and let $\tilde{U}:=\parens*{W^\dagger\otimes \bbone}\parens*{\bbone\otimes U}\parens*{W\otimes \bbone}$, 
        then we have
        \begin{equation*}
            \norm*{V-\alpha\parens*{\bra{0}^{\otimes s+1}\otimes \bbone}\tilde{U}\parens*{\ket{0}^{\otimes s+1}\otimes \bbone}}\leq\epsilon.
        \end{equation*}
	Note that $\alpha=O\parens*{1}$,
	by Lemma~\ref{lmm:robust oblivious amplitude amplification} we can construct
	a $\parens*{1,s+1,O\parens*{\epsilon}}$-block-encoding of $V$
	with $O\parens*{1}$ uses of $\tilde{U}$, $\tilde{U}^\dagger$ and $\bbone-2\Pi_{s+1}$.
	To achieve an $O\parens*{\epsilon}$-precision,
	the implementation of the rotation $W$ has depth $O\parens*{\log^2\log\parens*{1/\epsilon}}$ and size $O\parens*{\log^4\parens*{1/\epsilon}}$
	by Lemma~\ref{lmm:parallel controlled Z-rotation} and Lemma~\ref{lmm:parallel quantum circuit for elementary arithmetics}.
	The reflection $\bbone-2\Pi_{s+1}$ can be implemented by a $\mathrm{C}_{s+1}$-$Z$ gate
	with an ancilla qubit,
	which has depth $O\parens*{\log s}$ and size $O\parens*{s}$ by Corollary~\ref{cor:parallel controlled Z gate}.
	The final complexity follows from summing up these complexities.
\end{proof}

Now we can give a precise statement of the main result of this paper. An informal version of this theorem was presented as Theorem \ref{main-informal}. 

\begin{theorem}[Parallel simulation of uniform-structured Hamiltonians]
	\label{thm:parallel simulation for uniform-structured Hamiltonian}
	An $m$-uniform-structured Hamiltonian $H=\sum_{w\in [m]}H_w$ acting on $n$ qubits with each $H_w$ being $d$-sparse,
	can be simulated for time $t$ to precision $\epsilon$ ($\epsilon\leq 0.05$)
	by a quantum circuit of
	\begin{itemize}
		\item 
			depth $O\parens*{\tau\log \gamma}$ and size $O\parens*{\tau\gamma}$
			w.r.t.\ queries to $\calO_H^b$ with $b=O\parens*{\gamma}$,
		\item
			depth $O\parens*{\tau\log \gamma}$ and size $O\parens*{m\tau\gamma}$ w.r.t.\ queries to $\calO_P$, and
		\item
			depth $O\parens*{\tau\parens*{\log^2\gamma\cdot \log^2 n +\log^3 \gamma}}$
			and size $O\parens*{\tau\gamma^2\cdot\parens*{mn^4+\gamma^3}}$ w.r.t.\ gates,
	\end{itemize}
	where $\tau:=mdt$, $\gamma:=\log\parens*{\tau/\epsilon}$.
\end{theorem}
\begin{proof}
	Let $\Delta t:=1/\parens*{md}$.
	First consider the case when $\tau=t/\Delta t$ is an integer.
	To achieve a total precision $\epsilon$,
	applying Lemma~\ref{lmm:parallel Hamiltonian simulation delta t} followed by Corollary~\ref{cor:robust oblivious amplitude amplification}
	with precision $\epsilon/\tau$
	gives a $\parens*{1,O\parens*{\gamma\parens*{n+\log m}},O\parens*{\epsilon/\tau}}$-block-encoding of $e^{-iH\Delta t}$.
	Repeat the above procedure $\tau$ times; that is, 
	using Lemma~\ref{lmm:product of block-encoded Hamiltonian powers} to multiply these block-encoded $e^{-iH\Delta t}$
	we obtain a $\parens*{1,O\parens*{\tau\gamma\parens*{n+\log m}},O\parens*{\epsilon}}$-block-encoding of $e^{-iHt}$.
	By properly scaling the precision $\epsilon$ we can remove the constant factor in $O\parens*{\epsilon}$
	and implement $e^{-iHt}$ to precision $\epsilon$ with the same overhead up to a constant factor.
	The final complexity follows from summing up these complexities.

	For the case when $\tau$ is not an integer,
	that is, $\tilde{t}:=t-\bracks*{t/\Delta t}\neq 0$,
	we can independently simulate the last segment for time $\tilde{t}$.
	This can be done through simulating $\tilde{H}:=H\Delta t/\tilde{t}$ instead for time $\Delta t$,
	where the oracle $\calO_{\tilde{H}}$ for $\tilde{H}$ is easy to construct from $\calO_H$, 
	with at most $O\parens*{\log^2\gamma}$-depth and $O\parens*{\gamma^4}$-size of overhead for the required precision
	by Lemma~\ref{lmm:parallel quantum circuit for elementary arithmetics}.
	The final complexity is unchanged.
\end{proof}

\section{Lower Bounds}
\label{sec:lower_bounds}

In this section,  we prove Theorem~\ref{thm:lower bounds} in Section~\ref{sub:results}, which gives a lower bound on the gate depth of simulating a uniform-structured Hamiltonian and  implies that
the $\polylog\log\parens*{1/\epsilon}$ factor in the gate depth in Theorem~\ref{thm:parallel simulation for uniform-structured Hamiltonian}
cannot be significantly improved to $o\parens*{\log\log\parens*{1/\epsilon}}$.
Our proof is based on the proof of Theorem 1.2 in~\cite{BCCKS14},
which gives a lower bound that simulating any sparse Hamiltonian to precision $\epsilon$ has size
$\Omega\parens*{\frac{\log\parens*{1/\epsilon}}{\log\log\parens*{1/\epsilon}}}$ w.r.t.\ queries,
as an extension of the ``no-fast-forwarding theorem'' in~\cite{BACS06}.
Their proof basically reduces the problem of computing the parity of $N$ bits (with unbounded error, i.e., with success probability strictly greater than $1/2$),
to simulating a $2$-sparse $2N\times 2N$ Hamiltonian (to a high precision).
Our lower bound is achieved by two simple observations: the Hamiltonian used there is $6$-band, which is actually uniform-structured as shown in Lemma~\ref{lmm:band Hamiltonian is uniform-structured};
and computing the parity of $N$ bits with unbounded error requires depth $\Omega\parens*{\log N}$.

\begin{proof}[Proof of Theorem~\ref{thm:lower bounds}]
    We will show that there exists a uniform-structured Hamiltonian $H$ such that simulating $H$ to precision $\epsilon$
    requires gate depth $\Omega\parens*{\log\log\parens*{1/\epsilon}}$.
    Following~\cite{BCCKS14}, consider a $2N\times 2N$ Hamiltonian $H$ determined by an $N$-bit string $x=x_0\ldots x_{N-1}$,
    such that
    \begin{equation*}
        \bra{j,k}H\ket{j-1,k\oplus x_j}=\bra{j-1,k\oplus x_j}H\ket{j,k}=\frac{\sqrt{\parens*{N-j+1}j}}{N}
    \end{equation*}
    for all $j\in [N]$ and $k\in [2]$,
    where $\oplus$ stands for the XOR operator.
    Note that here $H$ is $6$-band, thus is uniform-structured by Lemma~\ref{lmm:band Hamiltonian is uniform-structured}.
    Also we have $\norm{H}_{\max}\leq 1$, which can be normalized to $1$ with a constant overhead.
    In~\cite{BCCKS14} it is shown that: 
    \begin{enumerate}
        \item
            $\abs*{\bra{N,\text{\small PARITY}\parens*{x}}e^{-iHt}\ket{0,0}}=\abs{\sin\parens*{t/N}^N}$,
            where $\text{\small PARITY}(x)=x_0\oplus\ldots\oplus x_{N-1}$ is the parity of $x$;
        \item
            if $N=\Theta\parens*{\frac{\log\parens*{1/\epsilon}}{\log\log\parens*{1/\epsilon}}}$,
            then there is an unbounded-error algorithm to compute $\text{\small PARITY}(x)$,
            by simulating $H$ for a constant time $t$ to precision $\epsilon$ on the input $\ket{0,0}$,
            followed by a computational basis measurement.
    \end{enumerate}
    
    Take $N=\Theta\parens*{\frac{\log\parens*{1/\epsilon}}{\log\log\parens*{1/\epsilon}}}$.
    To finish the proof,
    it suffices to show that computing $\text{\small PARITY}(x)$ with unbounded error
    requires gate depth $\Omega\parens*{\log N}=\Omega\parens*{\log\log \parens*{1/\epsilon}}$.
    This is trivial because $\text{\small PARITY}(x)$ depends on all $x_j$ for $j\in [N]$,
    while any $o\parens*{\log N}$-depth quantum circuit only involves $o\parens*{N}$ input qubits.
    More precisely, for the sake of contradiction,
    suppose there exists an $o\parens*{\log N}$-depth quantum circuit that takes an input $\ket{x}$ and outputs $\text{\small PARITY}(x)$
    by a measurement on the first output qubit with success probability $>1/2$,
    then there must be an input qubit holding $\ket{x_k}$ for some $k\in [N]$ that is not connected to the output qubit by a path of gates and wires.
    However, $x_k=0$ and $x_k=1$ yield different values of $\text{\small PARITY}(x)$. This gives a contradiction.
\end{proof}

\section{Applications}
\label{sec:applications}

It seems that the parallel quantum algorithm for Hamiltonian simulation developed above can be applied to a wide range of simulation tasks in physics.
As a concrete illustration,
three examples of physical interest are presented in this section.
For each example, we explicitly calculate the total gate complexity of the parallel quantum simulation algorithm for it  and compare with the prior art.
In particular, we calculate the gate cost for implementing the oracle $\calO_H$ and $\calO_P$ in the algorithm. 
Since our choices of $\calO_P$ turn out to be efficiently implementable by quantum gates in these examples,
compared to the commonly chosen oracle $\calO_L$,
we believe that our definition of the oracle $\calO_P$ is more reasonable in these applications.
The results in this section were already summarized in Subsection \ref{sub:results} as Table~\ref{tab:comparison with prior works simulating three models}.

\subsection{Simulation of the Heisenberg Model}
\label{sub:simulation_of_the_heisenberg_model}

Many body localization (MBL) is an intriguing phenomenon in the long-time behaviour of a closed quantum system with disorders and interactions~\cite{NH15,AP17,AL18}. In contrary to the conventional assumption in quantum statistical mechanics that a system coupling to a bath (i.e., a large environment) after a long time will achieve a thermal equilibrium which erases the initial condition, the MBL system as an isolated many-body quantum system resists such thermalization --- in a local subsystem the information of the initial state is remembered forever. While a theoretical understanding of MBL still remains challenging since its introduction by Anderson in 1958~\cite{Anderson58}, tremendous numerical works on various systems have been conducted through recent years (for example, see~\cite{NH15,AL18,AABS19} for reviews). A typical example for numeric studies of MBL is one-dimensional Heisenberg model~\cite{NH15,LLA15,CMNRS18}. Due to the difficulty of simulating many-body dynamics classically, quantum simulations can investigate properties of MBL in larger systems intractable for classical computers.

Following~\cite{CMNRS18},
we consider the problem of  simulating the one-dimensional nearest-neighbor Heisenberg model with a random magnetic field in the $z$ direction.
More concretely, we will simulate an $n$-qubit Hamiltonian
\begin{equation}
	H=\sum_{w\in [n]}\parens*{X_wX_{w+1}+Y_wY_{w+1}+Z_wZ_{w+1}+h_w Z_w}\label{eq:Heisenburg Hamiltonian},
\end{equation}
where $h_w\in [-1,1]$ is chosen uniformly at random,
the subscript $w$ indicates the qubit $w$ that the Pauli matrix acts on,
and $w=n$ is equivalent to $w=0$ by assuming the periodic boundary conditions.

First observe that $H$ is a $(2,n)$-local Hamiltonian
with local term $H_w:=X_wX_{w+1}+Y_{w}Y_{w+1}+Z_wZ_{w+1}+h_wZ_w$ for $w\in [n]$.
The locality of $H_w$ is indicated by an $n$-bit string $s(w)$ with the $u^{\text{th}}$ bit defined as
\begin{equation*}
	s(w)_u:=\bbracks{w=u\vee w+1=u}
\end{equation*}
for all $u\in [n]$, given by the oracle $\calO_P$ as in Lemma~\ref{lmm:local Hamiltonian is m-uniform-structured}.
Then by Lemma~\ref{lmm:local Hamiltonian is m-uniform-structured}, $H$ is $m$-uniform-structured with $m=n$,
thus we can apply Theorem~\ref{thm:parallel simulation for uniform-structured Hamiltonian} to simulate $H$.
For comparison, take the simulation time $t=n$ as in~\cite{CMNRS18}.
Since $\norm*{H}_{\max}\leq\sum_{w\in [n]}\norm*{H_w}_{\max}=O\parens*{n}$,
to simulate $H$ for time $t$, it is equivalent to simulate a normalized $H$ with max norm $\leq 1$ for a rescaled time $\tilde{t}:=\norm{H}_{\max}t=O\parens*{n^2}$.
Thus we can determine the parameters in Theorem~\ref{thm:parallel simulation for uniform-structured Hamiltonian}:  $\tau=O\parens*{n^3}$, $\gamma=O\parens*{\log \parens*{n/\epsilon}}$ and $b=O\parens*{\log\parens*{n/\epsilon}}$.

For the total complexity of the algorithm, 
one should also compute the gate complexities for implementing the oracle $\calO_H^b$ and $\calO_P$.

\begin{itemize}
	\item 
		To implement $\calO_P$,
		that is, to perform the mapping $\ket{w}\ket{0}\mapsto \ket{w}\ket{s(w)}$,
		one can first prepare the state $\ket{0,\ldots,n-1}$ in the second register,
		then make $n$ copies of $\ket{w}$ by applying $\mathrm{COPY}_n^{\ceil{\log n}}$,
		and finally calculate each Boolean function $s(w)_u$ for $u\in [n]$ in parallel followed by garbage cleaning.
		Together with Lemma~\ref{lmm:parallel quantum circuit for elementary arithmetics},
		this can be done by an $O\parens*{\log n}$-depth and $O\parens*{n\log n}$-size quantum circuit.
	\item
		To implement $\calO_H^b$, one needs to first generate $n$ uniform random $h_w\in [-1,1]$ to $b$-bit precision
		in the pre-processing.
		This can be done by an $O\parens*{1}$-depth and $O\parens*{nb}$-size quantum circuit using Hadamard gates followed by single-qubit measurements.
		The pre-processing only needs to be done once.
		Later when a $h_w$ is required each time we can apply a $\mathrm{COPY}^b$ gate to prepare a new copy of it.

		To perform the mapping $\ket{j,k,0}\mapsto \ket{j,k,H_{jk}}$,
		one can calculate an entry $H_{jk}$ by summing up those $\parens*{H_w}_{jk}$ with $\parens*{j,k}\in \mathsf{H}_{w}$.
		Recall that $\bbracks{\parens*{j,k}\in \mathsf{H}_w}$ is arithmetic-depth-efficiently computable for uniform-structured $H$,
		and $\parens*{H_w}_{jk}$ is easy to determine given a copy of $h_w$,
		therefore the above procedure can be implemented by an $O\parens*{\log^2 n+\log b}$-depth and $O\parens*{n^5+nb}$-size quantum circuit
		(in a way analogous to computing $c(j,k)$ in the proof of Lemma~\ref{lmm:walk stage complexity of m-uniform-structured}).
\end{itemize}

Thus, we obtain the following: 
\begin{corollary}[Parallel simulation of the Heisenberg model]
    \label{cor:parallel_simulation_of_the_heisenberg_model}
	The Heisenberg Hamiltonian in \eqref{eq:Heisenburg Hamiltonian} can be simulated for time $t = n$ to precision $\epsilon$
	by an $O\parens*{n^3\log^3 n \cdot \log^3\log\parens*{1/\epsilon}}$-depth 
	and $O\parens*{n^8\log^5\parens*{n/\epsilon}}$-size quantum circuit.
\end{corollary}

In~\cite{CMNRS18} the performance of different quantum simulation algorithms on this task are compared,
amongst which the asymptotically best one is based on quantum signal processing~\cite{LC17,LC19}, which achieves a gate complexity $O\left(n^3\log n+ \frac{\log \left(1/\epsilon\right)}{\log\log\left(1/\epsilon\right)}\cdot n\log n\right)$. Later a quantum algorithm to simulate lattice Hamiltonians~\cite{HHKL18} shows a better dependence on $n$ in the gate complexity, with depth $O\parens*{n\polylog\parens*{n/\epsilon}}$ and size $O\parens*{n^2\polylog\parens*{n/\epsilon}}$.
We note that in the gate depth all these works have a $\polylog\parens*{1/\epsilon}$ dependence, while Corollary~\ref{cor:parallel_simulation_of_the_heisenberg_model} only contains a $\polylog\log\parens*{1/\epsilon}$ factor.

\subsection{Simulation of the Sachdev-Ye-Kitaev Model}
\label{sub:simulation_of_the_sachdev_ye_kitaev_model}

The Sachdev-Ye-Kitaev (SYK) model~\cite{SY93,Kitaev15}, a simple but important exactly solvable many-body system, has drawn an increasing interest in the condense matter physics and high energy physics communities due to its many striking properties~\cite{SY93,Kitaev15,MS16,PR16,Rosenhaus19} and its potential to have an interesting holographic dual~\cite{Kitaev15,MS16}. Like the MBL problem in Section~\ref{sub:simulation_of_the_heisenberg_model},
numeric studies of a larger SYK model enabled by quantum simulation could extend our understandings about its features and dual interpretation.

Following~\cite{GELdS17,BBN19},
we consider the problem of  simulating the SYK model
evolving under a Hamiltonian
\begin{equation}
	H=\frac{1}{4\cdot 4!} \sum_{p,q,r,s=0}^{2n-1}J_{pqrs}\gamma_p\gamma_q\gamma_r\gamma_s,\label{eq:SYK model}
\end{equation}
where each $J_{pqrs}\sim \calN\parens*{0,\sigma^2}$ is chosen randomly from a normal distribution with variance $\sigma^2=\frac{3! J^2}{\parens*{2n}^3}$ ($J$ is assumed to be a constant),
and $\gamma_p$ are Majorana fermion operators such that $\braces*{\gamma_p,\gamma_q}=2\bbracks*{p=q}\bbone$.
The Majorana operator can be expressed as a tensor product of Pauli matrices by the Jordan-Wigner transformation
\begin{equation}
	\gamma_{p}\mapsto
	\begin{cases}
		Z_0\ldots Z_{p/2-1}X_{p/2},&\text{$p$ is even},\\
		Z_0\ldots Z_{\parens*{p-3}/2}Y_{\parens*{p-1}/2},&\text{$p$ is odd}
	\end{cases}
	\label{eq:Jordan-Wigner Majorana}
\end{equation}
for $p\in [2n]$ (unlike in~\cite{GELdS17}, our index starts from $0$),
where as usual the subscript of a Pauli matrix indicates the qubit it acts on.
Now $H$ can be expressed as a Pauli sum on $n$ qubits:
\begin{equation}
	H=\alpha \sum_{w\in [2n]^4} J_{w} H_w\label{eq:SYK qubit Hamiltonian}
\end{equation}
where $\alpha$ is a constant, $J_w$ is chosen randomly from a normal distribution,
and each $H_w$ is a tensor product of Pauli matrices.

It can be seen by Lemma~\ref{lmm:sum of tensor products of Pauli matrices} that the  SYK Hamiltonian in \eqref{eq:SYK qubit Hamiltonian} is $m$-uniform-structured
 with $m=(2n)^4$,
thus we can apply Theorem~\ref{thm:parallel simulation for uniform-structured Hamiltonian} to simulate it.
Note that
\begin{equation*}
	\norm*{H}_{\max}\leq \sum_{w\in [2n]^4}\alpha\norm*{J_w H_w}_{\max}\leq
	\sum_{w\in [2n]^4}\alpha\E\bracks*{\abs*{J_w}}=\alpha\parens*{2n}^4\cdot \frac{\sqrt{2}\sigma}{\sqrt{\pi}}=O\parens*{n^{2.5}}
\end{equation*}
where we use $J_w\sim \calN\parens*{0,\sigma^2}$ with $\sigma^2=O\parens*{1/n^3}$.
Thus, to simulate $H$ for time $t$,
it is equivalent to simulate its normalized Hamiltonian with max norm $\leq 1$ for a rescaled time
$\tilde{t}:=\norm{H}_{\max}t=O\parens*{n^{2.5}t}$.
We can determine the parameters in Theorem~\ref{thm:parallel simulation for uniform-structured Hamiltonian}:   $\tau=O\parens*{n^{6.5}t}$, $\gamma=O\parens*{\log\parens*{nt/\epsilon}}$ and $b=O\parens*{\log\parens*{nt/\epsilon}}$.

For the total complexity of the algorithm,
we starts from defining a function $\calP:[2n]^4\rightarrow \braces*{1,i,-1,-i}\times\braces*{\bbone,X,Y,Z}^n \simeq [4]^{n+1}$
that maps $w\in [m]$ to the Pauli string (with a global phase) of $H_w$.
For example, if $H_w=-iX\otimes Z\otimes Z$ then $\calP(w)=\parens*{-i,X,Z,Z}$.
One can also write the Jordan-Wigner transformation in \eqref{eq:Jordan-Wigner Majorana}
as a function $\calJ:[2n]\rightarrow \braces*{\bbone,X,Y,Z}^n\simeq [4]^n$,
which can be computed by an $O\parens*{\log n}$-depth and $O\parens*{n\log n}$-size quantum circuit.
To see this, that is, to perform the mapping $\ket{p}\ket{0}\mapsto \ket{p}\ket{\calJ\parens*{p}}$,
we can first prepare the state $\ket{0,\ldots, n-1}$ in the second register,
then make $n$ copies of $\ket{p}$,
and finally compute the $q^{\text{th}}$ bit of $\calJ\parens*{p}$:
\begin{equation*}
	\calJ\parens*{p}_q:=
	\begin{cases}
		3\bbracks{q<p/2}+\bbracks{q=p/2},&\text{$p$ is even},\\
		3\bbracks{q<(p-1)/2}+2\bbracks{q=(p-1)/2},&\text{$p$ is odd}
	\end{cases}
	\in [4]
\end{equation*}
for all $q\in [n]$ in parallel followed by garbage cleaning.
By Lemma~\ref{lmm:parallel quantum circuit for elementary arithmetics},
this can be done by an $O\parens*{\log n}$-depth and $O\parens*{n\log n}$-size quantum circuit.
Since for all $(p,q,r,s)\in [2n]^4$, the function $\calP\parens*{p,q,r,s}$ can be computed from
$\calJ\parens*{p},\calJ\parens*{q},\calJ\parens*{r},\calJ\parens*{s}$
by calculating Pauli algebra on each entry of these $n$-tuples
and gathering up the global phases,
$\calP\parens*{w}$ can be computed by an $O\parens*{\log n}$-depth and $O\parens*{n\log n}$-size quantum circuit.

Now we will determine the gate complexities for implementing the oracle $\calO_H^b$ and $\calO_P$: 

\begin{itemize}
	\item 
		To implement $\calO_P$, note that $s(w)$ in Lemma~\ref{lmm:sum of tensor products of Pauli matrices}
		is easily computable from $\calP(w)$ by function restriction, thus $\calO_P$ can be implemented by an
		$O\parens*{\log n}$-depth and $O\parens*{n\log n}$-size quantum circuit.
	\item
		To implement $\calO_H^b$, one needs to first generate $m$ independent $J_w\in \calN\parens*{0,\sigma^2}$ to $b$-bit precision
		in the pre-processing.
		This can be done by the Box-Muller transform~\cite{BM58},
		which derives two independent standard normally distributed variables from two independent uniformly distributed variables in $[0,1]$
		by elementary arithmetic.
		Thus it suffices to generate $m$ uniformly random variables in $[0,1]$ to $b$-bit precision
		by an $O\parens*{1}$-depth and $O\parens*{n^4b}$-size quantum circuit using Hadamard gates followed by single-qubit measurements,
		and then use Lemma~\ref{lmm:parallel quantum circuit for elementary arithmetics}
		to perform the Box-Muller transform with a scaling $\sigma$
		by an $O\parens*{\log^2 b}$-depth and $O\parens*{n^4b^4}$-size quantum circuit.
		The pre-processing only needs to be done once.
		Later when a $J_w$ is required each time we can apply a $\mathrm{COPY}^b$ gate to make a new copy of it.

		We perform the mapping $\ket{j,k,0}\mapsto \ket{j,k,H_{jk}}$
		in a way similar to the one in Section~\ref{sub:simulation_of_the_heisenberg_model};
		that is, sum up those $\parens*{H_w}_{jk}$ with $\parens*{j,k}\in \mathsf{H}_w$.
		Recall that each $\bbracks{\parens*{j,k}\in \mathsf{H}_w}$ is arithmetic-depth-efficiently computable for uniform-structured $H$,
		and $\parens*{H_w}_{jk}$ is easy to compute from the function $\calP(w)$ and a copy of $J_w$. 
		Therefore,  the above procedure can be implemented by an $O\parens*{\log^2 n +\log b}$-depth and $O\parens*{n^8+n^4b}$-size quantum circuit
		(in a way analogous to the proof of Lemma~\ref{lmm:walk stage complexity of m-uniform-structured}).
\end{itemize}

Applying Theorem~\ref{thm:parallel simulation for uniform-structured Hamiltonian},  we obtain the following: 

\begin{corollary}[Parallel simulation of the SYK model]
    \label{cor:parallel_simulation_of_the_SYK_model}
	The SYK model defined in \eqref{eq:SYK model} can be simulated for time $t$ to precision $\epsilon$
	by an $O\parens*{n^{6.5}\log^3 n\cdot t\log^3\log\parens*{t/\epsilon}}$-depth
	and $O\parens*{n^{14.5}t\log^5\parens*{nt/\epsilon}}$-size quantum circuit.
\end{corollary}

The prior best algorithm based on asymmetric qubitization~\cite{BBN19} for this task has gate complexity $O\parens*{n^{3.5}t+n^{2.5}t\polylog\parens*{n/\epsilon}}$, which improves the product-formula-based algorithm~\cite{GELdS17} with gate complexity $O\parens*{n^{10}t^2/\epsilon}$. Later work~\cite{XSSS20} proposes a sparse SYK model also of physical interest and gives a simulation algorithm based on qubitization with a $\polylog\parens*{1/\epsilon}$ dependence in the gate complexity. Compared to these works, by introducing parallelism  Corollary~\ref{cor:parallel_simulation_of_the_SYK_model} only has a $\polylog\log\parens*{1/\epsilon}$ dependence in the gate depth.

\subsection{Simulation of Quantum Chemistry in Second Quantization}
\label{sub:simulation_of_quantum_chemistry_in_second_quantization}

One of the most attractive prospects of quantum simulation is in quantum chemistry to study the static and dynamic properties of chemicals~\cite{YWBTA14,CRODJKKMPSSVA19,BBMC20}. Prior works on quantum simulation algorithm for chemistry mainly focus on exploiting the special structure of molecular Hamiltonians (e.g.\ based on first or second quantization, using Gaussian orbital bases or plane wave bases, etc.) to obtain better complexities (see~\cite{BBMC20} for a review). Compared to these fault-tolerant quantum algorithms, diverse variational quantum algorithms to circumvent a direct simulation (e.g.\ variational quantum eigensolver~\cite{PMSYZLAO14}) also have been explored in recent years due to their potential for immediate applications in the NISQ era~\cite{Google20}. 

Following~\cite{BBKWLA16}
we consider the problem of  simulating a molecular electronic structure Hamiltonian in the second-quantized form.
In this form, we will simulate a Hamiltonian
\begin{equation}
	H=\sum_{p,q\in [n]}h_{pq}a_p^\dagger a_q + \frac{1}{2} \sum_{p,q,r,s\in [n]}h_{pqrs}a_p^\dagger a_q^\dagger a_r a_s
	\label{eq:chemistry fermion}
\end{equation}
where $n$ represents the number of spin orbitals, $h_{pq}, h_{pqrs}$ are one-electron and two-electron integrals,
and $a_p^\dagger, a_p$ are fermionic creation and annihilation operators satisfying the relations
\begin{equation*}
	\braces*{a_p^\dagger, a_q}=\bbracks*{p=q}\bbone,\qquad \braces*{a_p^\dagger, a_q^\dagger}=\braces*{a_p,a_q}=0.
\end{equation*}
As in the ``database'' algorithm~\cite{BBKWLA16},
we assume that these $b$-bit precise $h_{pq}$ and $h_{pqrs}$ are precomputed and stored in a database; 
for example, an $O\parens*{(nb)^4}$-size quantum-read/classical-write RAM (QCRAM),
to which one quantum access can be done by an $O\parens*{\log \parens*{nb}}$-depth and $O\parens*{(nb)^4}$-size quantum circuit.

One can apply Jordan-Wigner transformation to the creation and annihilation operators $a_p^\dagger,a_p$
to obtain Hamiltonians acting on qubits:
\begin{align}
	a_p^\dagger &\mapsto \frac{1}{2}Z_0\ldots Z_{p-2}\parens*{X_{p-1}-iY_{p-1}}\label{eq:creation to qubit}\\
	a_p&\mapsto \frac{1}{2} Z_0\ldots Z_{p-2}\parens*{X_{p-1}+iY_{p-1}}
	\label{eq:annihilation to qubit}
\end{align}
for $p\in [n]$,
where as usual the subscript of a Pauli matrix indicates the qubit it acts on.
Note that each resulting qubit Hamiltonian in \eqref{eq:creation to qubit} and \eqref{eq:annihilation to qubit}
is a sum of two tensor products of Pauli matrices,
hence by splitting them and applying the transformation to \eqref{eq:chemistry fermion}
we obtain $H$ as a Pauli sum on $n$ qubits:
\begin{equation}
	H=\sum_{w\in [m]}h_w H_w\label{eq:chemistry qubit},
\end{equation}
where $m=O\parens*{n^4}$,
each $h_w$ is some $h_{pq}$ or $h_{pqrs}$ in \eqref{eq:chemistry fermion} (up to a constant factor),
and each $H_w$ is a tensor product of Pauli matrices.
Here we omit the explicit mapping from 
the indices $p,q,r,s$ in \eqref{eq:chemistry fermion}
to the index $w$ in \eqref{eq:chemistry qubit},
but mention that the mapping can be efficiently performed
by an $O\parens*{1}$-depth quantum circuit.

Similar to the SYK model, we see by Lemma~\ref{lmm:sum of tensor products of Pauli matrices} that  
the molecular Hamiltonian $H$ in \eqref{eq:chemistry qubit} is $m$-uniform-structured with $m=O\parens*{n^4}$.
According to~\cite{BBKWLA16}, we have $\sum_{w\in [m]}\abs*{h_w}=O\parens*{n^4}$. 
Thus the max norm of $H$ is bounded by 
\begin{equation*}
    \norm{H}_{\max}\leq \sum_{w\in [m]}\norm{h_w H_w}_{\max}\leq \sum_{w\in [m]}\abs*{h_w}=O\parens*{n^4}.
\end{equation*}
To simulate $H$ for time $t$ it is equivalent to simulate its normalized Hamiltonian for a rescaled time $\tilde{t}:=\norm{H}_{\max}t=O\parens*{n^4t}$.
The parameters in Theorem~\ref{thm:parallel simulation for uniform-structured Hamiltonian} can be determined as follows:  $\tau=O\parens*{n^8 t}$, $\gamma=O\parens*{\log\parens*{nt/\epsilon}}$ and $b=O\parens*{\log \parens*{nt/\epsilon}}$
.

Now let us compute the gate complexities for implementing the oracle $\calO_H^b$ and $\calO_P$:

\begin{itemize}
	\item
		The construction of $\calO_P$ is similar to the one in Section~\ref{sub:simulation_of_the_sachdev_ye_kitaev_model}.
		Here we will omit the details and claim that it can be implemented by an $O\parens*{\log n}$-depth and $O\parens*{n\log n}$-size quantum circuit.
	\item
		To implement $\calO_H^b$, that is, to perform the mapping $\ket{j,k,0}\mapsto \ket{j,k,H_{jk}}$,
		we sum up those $\parens*{H_{w}}_{jk}$ with $\parens*{j,k}\in \mathsf{H}_w$ to obtain $H_{jk}$.
		Recall that $\bbracks{\parens*{j,k}\in \mathsf{H}_w}$ is arithmetic-depth-efficiently computable for uniform-structured $H$.
		Also $\parens*{H_{w}}_{jk}$ is easy to compute given: a copy of $h_w$ which can be read out from the database
		by an $O\parens*{\log \parens*{nb}}$-depth and $O\parens*{(nb)^4}$-size quantum circuit;
		together with a string characterizing the Pauli string of $H_w$ (like the function $\calP(w)$ for the SYK model)
		computable by an $O\parens*{\log n}$-depth and $O\parens*{n\log n}$-size quantum circuit.
		Therefore the above procedure can be implemented by
		an $O\parens*{\log^2 n+\log b}$-depth and $O\parens*{n^8 b^4}$-size quantum circuit
		(in a way analogous to the proof of Lemma~\ref{lmm:walk stage complexity of m-uniform-structured}).
\end{itemize}

As an application of Theorem~\ref{thm:parallel simulation for uniform-structured Hamiltonian}, we have the following: 

\begin{corollary}[Parallel simulation of molecular Hamiltonians]
    \label{cor:parallel_simulation_of_the_molecular_hamiltonian}
	The molecular Hamiltonian in \eqref{eq:chemistry fermion} can be simulated for time $t$ to precision $\epsilon$
	by an $O\parens*{n^8\log^3 n \cdot t \log^3\log\parens*{t/\epsilon}}$-depth
	and $O\parens*{n^{16}t\log^5\parens*{nt/\epsilon}}$-size quantum circuit.
\end{corollary}

Prior work~\cite{BBKWLA16} gives a quantum simulation algorithm for molecular Hamiltonians with gate complexity $O\parens*{\frac{n^8t\log\parens*{nt/\epsilon}}{\log\log\parens*{nt/\epsilon}}}$. Later algorithmic improvements (e.g.~\cite{BBSKSWLA17,BWMMNC18,KMWGACB18,BBMN19,BGMMB19}) focus on parameters other than the precision $\epsilon$, and all these works have a poly-logarithmic dependence on $\epsilon$ as the Hamiltonian simulation subroutines used there have such dependence. By allowing parallelism, Corollary~\ref{cor:parallel_simulation_of_the_molecular_hamiltonian} exponentially improves the dependence to $\polylog\log\parens*{1/\epsilon}$ with respect to the depth.

\section*{Acknowledgements}

We thank Zhengfeng Ji for pointing out that local Hamiltonians can be represented as Pauli sums.
We acknowledge an anonymous reviewer for raising the open question (in Section~\ref{sub:discussion}) about whether our method can be generalized for simulating time-dependent Hamiltonians via the Dyson series.
Zhicheng Zhang also thanks Rolando D.\ Somma for mentioning that the results in \cite{AA17} do not concern the depth complexity.

This work was supported in part by the National Natural Science Foundation of China under
Grant 61832015. 
Zhicheng Zhang was supported by the Sydney Quantum Academy, NSW, Australia.
Qisheng Wang was also supported in part by the MEXT Quantum Leap Flagship Program (MEXT Q-LEAP) grants No. JPMXS0120319794.

\bibliographystyle{quantum}
\bibliography{main}

\onecolumn
\appendix

\section{Details Omitted in Section~\ref{sub:general_uniform_structured_hamiltonian}}
\label{sec:details_of_section_extension}

In this appendix we provide details of Section~\ref{sub:general_uniform_structured_hamiltonian}.

\begin{proof}[Proof of Lemma~\ref{lmm:(r,m)-parallel quantum walk}]
	Similar to the proof of Lemma~\ref{lmm:T^rdaggerS^rT^r block encode (H/d)^r},
	we can write $\ket{\Psi_{j_0}^{(r,m)}}=\ket{\Phi_{j_0}^{(r,m)}}+\ket{\Phi_{j_0}^{(r,m)\bot}}$,
	where the subnormalized state
	\begin{equation*}
		\ket{\Phi_{j_0}^{(r,m)}}:=\frac{1}{\sqrt{(md)^r}}\sum_{\bw\in [m]^r}\sum_{\bj\in \mathsf{H}^{\bw}}\ket{\bw}\ket{\bj}
		\otimes \sqrt{\tilde{H}_{j_0j_1}^*\ldots \tilde{H}_{j_{r-1}j_r}^*}\ket{\bj}
	\end{equation*}
	These subnormalized states also satisfy some orthogonal relations analogous to Lemma~\ref{lmm:Phi and Phi^bot}:
 	for all $j,k\in [N]$,
	\begin{enumerate}
		\item 
	 $\bra{\Phi_{j}^{(r,m)}}S^{(r,m)}\ket{\Phi_{k}^{(r,m)}}=\parens*{\parens*{\frac{H}{md}}^r}_{jk}$.
		\item
	 $\bra{\Phi_j^{(r,m)}}S^{(r,m)}\ket{\Phi_k^{(r,m)\bot}}=\bra{\Phi_j^{(r,m)\bot}}S^{(r,m)}\ket{\Phi_k^{(r,m)\bot}}=0$.
	\end{enumerate}
	We omit the proof details here, but mention that this can done by the same techniques as in the proof of Lemma~\ref{lmm:Phi and Phi^bot},
	combined with the fact that $\sum_{w\in [m]} \tilde{H}_{j_{s}j_{s+1}}=H_{j_sj_{s+1}}$.
	Using the orthogonal relations the conclusion is easy to obtain as in the proof of Lemma~\ref{lmm:T^rdaggerS^rT^r block encode (H/d)^r}.
\end{proof}

\begin{proof}[Proof of Lemma~\ref{lmm:sum of tensor products of Pauli matrices}]
	The proof is similar to the proof of Lemma~\ref{lmm:tensor product of Pauli matrices}.
	We verify the conditions in Definition~\ref{def:m-uniform-structured}.
	Observe that $L(w,j,t)=j\oplus x(w)$
	where $\oplus$ is the bit-wise XOR operator.
	\begin{itemize}
		\item 
			We have
			\begin{equation*}
				L^{(r)}\parens*{\bw,j,\bt}=j\oplus s(w_0)\oplus\ldots \oplus s(w_{r-1});
			\end{equation*}
			that is, take $f(j,k)=j\oplus k$, $g(w,t)=s(w)$ and $\circ$ to be $\oplus$ in \eqref{eq:m-uniform-structured L^r}.
			\begin{itemize}
				\item 
					$f$ is arithmetic-depth-efficiently computable,
					while computing $g(w,t)$ requires $O\parens*{1}$ queries to $\calO_P$.
				\item
					The XOR $\oplus$ is obviously associative and arithmetic-depth-efficiently computable.
			\end{itemize}
		\item
			Take the inverse function to be $L^{(-1)}\parens*{w,j,k}=s\parens*{w}$,
			which can be computed by a single query to $\calO_P$.
		\item
			Observe that $(j,k)\in \mathsf{H}_w$ if and only if $j\oplus k = x(w)$;
			that is, $\bbracks{\parens*{j,k}\in \mathsf{H}_w}=\bbracks{j\oplus k=x(w)}$,
			which is obviously arithmetic-depth-efficiently computable with $O\parens*{1}$ queries to $\calO_P$.
	\end{itemize}
\end{proof}

\begin{proof}[Proof of Lemma~\ref{lmm:local Hamiltonian is m-uniform-structured}]
	In this proof we use superscript $i$ to denote the $i^{\text{th}}$ bit of a number.
	We verify the conditions in Definition~\ref{def:m-uniform-structured}.
	Observe that $L(w,j,t)=j \vartriangleleft_{s(w)} \parens*{t\restriction_{s(w)}}$,
	with $\vartriangleleft_s$ and $\restriction_s$ defined
	in the proof of Lemma~\ref{lmm:local Hamiltonian term is uniform-structured}.
	\begin{itemize}
		\item 
			We have
			\begin{equation*}
				L^{(r)}\parens*{\bw,j,\bt}=
				j\vartriangleleft\parens*{t_0\restriction_{s(w_0)},s(w_0)}\vartriangleleft^\vee \ldots\vartriangleleft^\vee \parens*{t_{r-1}\restriction_{s(w_{r-1})},s(w_{r-1})};
			\end{equation*}
			that is, take
			$f\parens*{j,(k,x)}=j\vartriangleleft(k,x)$,
			$g\parens*{w,t}=\parens*{t\restriction_{s(w)},s(w)}$,
			and $\circ$ to be $\vartriangleleft^\vee$,
			where $\vartriangleleft$ and $\vartriangleleft^\vee$ are defined in the proof of Lemma~\ref{lmm:local Hamiltonian term is uniform-structured}.
			\begin{itemize}
				\item
					As shown in the proof of Lemma~\ref{lmm:local Hamiltonian term is uniform-structured},
					$f$ is arithmetic-depth-efficiently computable.

					The mapping $\ket{w}\ket{0}\mapsto \ket{w}\frac{1}{\sqrt{d}}\sum_{t\in [d]}\ket{g(w,t)}$ can be performed
					by first querying $\calO_P$ to obtain $\ket{s\parens*{w}}$ in an ancilla,
					then conditioned on it implementing $n$ controlled Hadamard gates in parallel,
					assuming $d$ is a power of two w.l.o.g.
					This is arithmetic-depth-efficient with $O\parens*{1}$ queries.
				\item 
					$\vartriangleleft^\vee$ is associative and arithmetic-depth-efficiently computable,
					as shown in the proof of Lemma~\ref{lmm:local Hamiltonian term is uniform-structured}.
			\end{itemize}
		\item
			Take the inverse function to be $L^{(-1)}\parens*{w,j,k}=\parens*{k\wedge s(w),s(w)}$ with $\wedge$ the bit-wise AND operator.
			This is arithmetic-depth-efficiently computable with $O\parens*{1}$ queries to $\calO_P$.
		\item
			Observe that $(j,k)\in \mathsf{H}_w$ if and only if $j^i=k^i$ for all $i$ with $s(w)^i=0$;
			that is,
			\begin{equation*}
				\bbracks{\parens*{j,k}\in \mathsf{H}_w}=\bigwedge_{i\in [n]} \parens*{j^ik^i+(1-j^i)(1-k^i)}\cdot\parens*{1-s(w)^i}
			\end{equation*}
			which is arithmetic-depth-efficiently computable by an $O\parens*{\log n}$-depth and $O(n)$-size quantum circuit
			by taking $\circ$ to be AND gate in Lemma~\ref{lmm:sequence of associative operators},
			with additional $O(1)$ queries to $\calO_P$ to compute $s(w)$.
	\end{itemize}
\end{proof}

\subsection{Proof of Theorem~\ref{thm:complexity of Q^(r,m)}}

Recall that the Hamiltonians considered in Section~\ref{sub:general_uniform_structured_hamiltonian}
has the form $H=\sum_{w\in [m]} H_w$.
To prove Theorem~\ref{thm:complexity of Q^(r,m)},
following the same line of analysis as in Section~\ref{sub:parallel_version},
we will show that:
\begin{enumerate}
    \item
        for $m$-uniform-structured Hamiltonians,
        the pre-walk, i.e., preparing the pre-walk state $\ket{p_{j_0}^{(r,m)}}$ in \eqref{eq:extended pre-walk state},
        can be implemented by a parallel quantum circuit;
    \item
        the re-weight can be efficiently implemented in parallel;
    \item
        the (extended) parallel quantum walk for uniform-structured Hamiltonians can be efficiently implemented in parallel.
\end{enumerate}

Firstly we have the following lemma as a generalization of Corollary~\ref{cor:log depth path for 1-uniform-structured}.

\begin{lemma}
	\label{lmm:log depth path for m-uniform-structured}
	For an $m$-uniform-structured Hamiltonian $H$, the mapping
	\begin{equation*}
		\ket{\bw,j_0,g\parens*{\bw,\bt}}\mapsto \ket{\bw}\ket{\bj}
	\end{equation*}
	for all $j_0\in [N],\bw\in [m]^r, \bt\in [d]^{r}$,
	where $g\parens*{\bw,\bt}:=\parens*{g\parens*{w_0,t_0},\ldots,g\parens*{w_{r-1},t_{r-1}}}$
	and $\bj\in \mathsf{H}^r$ satisfies $j_{s+1}=L\parens*{w_s,j_s,t_s}$ for $s\in [r]$,
	can be implemented by a quantum circuit of
	\begin{itemize}
		\item
			depth $O\parens*{1}$ and size $O\parens*{r}$ w.r.t.\ queries to $\calO_P$, and
		\item
			depth $O\parens*{\log r\cdot \log^2 n}$ and size $O\parens*{r^2n^4}$ w.r.t.\ gates.
	\end{itemize}
\end{lemma}
\begin{proof}
	The lemma can be proved along the same line as Lemma~\ref{lmm:log depth L^r for 1-uniform-structured},
	Corollary~\ref{cor:log depth L^r for 1-uniform-structured parallel} and Corollary~\ref{cor:log depth path for 1-uniform-structured}.
	Recall that $m=\poly n$.
	We omit the proof details here, but mention that only the first two conditions of Definition~\ref{def:m-uniform-structured} are needed.
\end{proof}

Then the following lemma shows
that an $(r,m)$-pre-walk can be implemented by a parallel quantum circuit,
as a generalization of Lemma~\ref{lmm:pre-walk on 1-uniform-structured}.

\begin{lemma}[Pre-walk on $m$-uniform-structured Hamiltonians]
	\label{lmm:pre-walk on m-uniform-structured}
	An $(r,m)$-pre-walk on an $m$-uniform-structured Hamiltonian $H$,
	i.e., preparing the state $\ket{p_{j_0}^{(r,m)}}$,
	can be implemented by a quantum circuit of
	\begin{itemize}
		\item
			depth $O\parens*{1}$ and size $O\parens*{r}$ w.r.t.\ queries to $\calO_P$, and
		\item 
			depth $O\parens*{\log r\cdot \log^2 n}$ and size $O\parens*{r^2n^4}$ w.r.t.\ gates.
	\end{itemize}
\end{lemma}
\begin{proof}
	Let $\calH^W\otimes\calH^J=\parens*{\bigotimes_{s\in [r]}\calH_s^W}\otimes \parens*{\bigotimes_{s\in [r+1]}\calH_s^J}$ be the state space
	with $\calH_s^W=\Co^{m}$ and $\calH_s^J=\Co^{N}$.
	The process of preparing $\ket{p_{j_0}^{(r,m)}}$ is presented below,
	starting from the initial state $\ket{0}^{\otimes mr}\ket{j_0}\ket{0}^{\otimes Nr}$
	with $\ket{j_0}\in \Co^{N},\ket{0}\in \Co$.
	\begin{enumerate}
		\item
			\label{stp:1 in lmm:pre-walk on m-uniform-structured}
			Prepare a superposition over computational basis states of size $m^r$ in the subspace $\calH^W$:
			\begin{equation*}
				\frac{1}{\sqrt{m^r}}\sum_{\bw\in [m]^r}\ket{\bw}\ket{j_0}\ket{0}^{\otimes Nr}.
			\end{equation*}
		\item
			\label{stp:2 in lmm:pre-walk on m-uniform-structured}
			For each $s\in [r]$,
			perform in parallel the mapping $\ket{w_s}\ket{0}\mapsto \ket{w_s}\frac{1}{\sqrt{d}}\sum_{t_s\in [d]}\ket{g\parens*{w_s,t_s}}$ in the subspace $\calH_s^W\otimes\calH_{s+1}^J$,
			to obtain the state
			\begin{equation*}
				\frac{1}{\sqrt{(md)^r}}\sum_{\bw\in [m]^r}\sum_{\bt\in [d]^r}\ket{\bw}\ket{j_0}\ket{g\parens*{\bw,\bt}}
			\end{equation*}
		\item
			Apply Lemma~\ref{lmm:log depth path for m-uniform-structured}
			to obtain the goal state $\ket{p_{j_0}^{(r,m)}}=\frac{1}{\sqrt{\parens*{md}^r}}\sum_{\bw\in [m]^r}\sum_{\bj\in \mathsf{H}^{\bw}}\ket{\bw}\ket{\bj}$.
	\end{enumerate}
	Step~\ref{stp:1 in lmm:pre-walk on m-uniform-structured} can be done by an $O(1)$-depth and $O\parens*{r\log m}$-size quantum circuit using Hadamard gates in parallel,
	assuming $m$ is a power of two w.l.o.g.
	Each mapping 
        \begin{equation*}
            \ket{w}\ket{0}\mapsto \ket{w}\frac{1}{\sqrt{d^r}}\sum_{t\in [d]}\ket{g(w,t)}
        \end{equation*}
	in Step~\ref{stp:2 in lmm:pre-walk on m-uniform-structured} is arithmetic-depth-efficient with $O\parens*{1}$-queries to $\calO_P$ due to Definition~\ref{def:m-uniform-structured}.
	Recall that $m=\poly n$,
	the final complexity follows from summing up these complexities combined with Lemma~\ref{lmm:log depth path for m-uniform-structured}.
\end{proof}

Now we move to the re-weight analysis.

\begin{lemma}
	\label{lmm:walk stage complexity of m-uniform-structured}
	Re-weight of $\ket{p_{j_0}^{(r,m)}}$,
	i.e., performing the mapping $\ket{p_{j_0}^{(r,m)}}\mapsto \ket{\Psi_{j_0}^{(r,m)}}$,
	where $\ket{\Psi_{j_0}^{(r,m)}}$ is defined in \eqref{eq:Psi_j^(r,m)},
	can be implemented to precision $\epsilon$ by a quantum circuit of 
	\begin{itemize}
		\item 
			depth $O(1)$ and size $O(r)$ w.r.t.\ queries to $\calO_H^b$ with $b=O\parens*{\log \parens*{m/\epsilon}}$,
		\item
			depth $O\parens*{1}$ and size $O\parens*{rm}$ w.r.t.\ queries to $\calO_P$, and
		\item
			depth $O\parens*{\log^2 n + \log^2\log \parens*{m/\epsilon}}$ and size $O\parens*{r \bracks*{mn^4+\log^4\parens*{m/\epsilon}}}$ w.r.t.\ gates,
	\end{itemize}
	for $r=\polylog \parens*{1/\epsilon}$.
\end{lemma}
\begin{proof}
	The analysis is exactly the same as in Lemma~\ref{lmm:walk stage complexity},
	with additional complexities of:
	\begin{enumerate}
		\item
			\label{stp:1 in lmm:walk stage complexity of m-uniform-structured}
			Computing $c(j,k)=\sum_{w\in [m]}\bbracks{\parens*{j,k}\in\mathsf{H}_w} $ in $r$ subspaces,
			each of which can be implemented by a quantum circuit of
			\begin{itemize}
				\item
					depth $O\parens*{1}$ and size $O\parens*{m}$ w.r.t.\ queries to $\calO_P$, and
				\item
					depth $O\parens*{\log^2 n}$ and size $O\parens*{mn^4}$ w.r.t.\ gates.
			\end{itemize}
			To see this, note that for each subspace,
			one can first creates $m$ copies of $(j,k)$ by $\mathrm{COPY}^{2n}_{m}$,
			then compute each $\bbracks{\parens*{j,k}\in \mathsf{H}_w}$ arithmetic-depth-efficiently with $O\parens*{1}$ queries to $\calO_P$ for $w\in [m]$ in parallel,
			due to the third condition in Definition~\ref{def:m-uniform-structured}.
			Finally apply Corollary~\ref{cor:parallel addition of a sequence} to compute $c(j,k)$ in each subspace,
			followed by garbage cleaning by reverse computation.
			The final complexity follows from the assumption $m=\poly n$.
		\item
			Scaled total precision $\tilde{\epsilon}=\epsilon/m$ for the aforementioned construction of the oracle $\calO_{\tilde{H}}$.
		\item
			Arithmetic-depth-efficient circuits for the division $\tilde{H}_{jk}=H_{jk}/c(j,k)$ by Lemma~\ref{lmm:parallel quantum circuit for elementary arithmetics}.
	\end{enumerate}
\end{proof}

Finally, combining Lemma~\ref{lmm:pre-walk on m-uniform-structured}, Lemma~\ref{lmm:walk stage complexity of m-uniform-structured}
and the negligible complexity of implementing $S^{(r,m)}$
gives the complexity of the $(r,m)$-parallel quantum walk $Q^{(r,m)}$
for $m$-uniform-structured Hamiltonians in Theorem~\ref{thm:complexity of Q^(r,m)}.

\section{Definition of State Preparation}
\label{sec:state_preparation_for_lcu}

In this appendix we show that the state preparation unitary in Definition~\ref{def:state preparation unitary}
is a special case of Definition 51 in~\cite{GSLW19}:

\begin{definition}[State preparation pair~\cite{GSLW19}]
    \label{def:state preparation pair}
    Let $\ba\in \Co^R$ and $\norm{\ba}_1\leq \alpha$, the pair of unitaries $(V_C,V_D)$ is called an $(\alpha,b,\epsilon)$-state-preparation-pair
    of $\ba$ if $V_C\ket{0}=\sum_{r\in [2^b]} c_r\ket{r}$ and $V_D\ket{0}=\sum_{r\in [2^b]} d_r\ket{r}$ with $\ket{0}\in \Co^{2^b}$, such that
    $\sum_{r\in [R]}\abs*{\alpha\parens*{c_r^*d_r}-a_r}\leq \epsilon$ and for all $r\in R,\ldots,2^b-1$ we have $c_r^*d_r=0$.
\end{definition}

This is shown more precisely in the following lemma.

\begin{lemma}
    Let $\ba\in \Co^R$ and $\alpha:=\norm{\ba}_1$,
    where we assume $R=2^s \geq 4$ w.l.o.g.
    Let $V$ be an $(\alpha,\epsilon)$-state-preparation-unitary of $\ba$,
    then $(V^\dagger, V)$ is an $(\alpha, s, \alpha R \epsilon)$-state-preparation-pair of $\ba$.
\end{lemma}

\begin{proof}
    Let $V\ket{0}=\sum_{r\in [R]} \sqrt{v_r}\ket{r}$.
    By Definition~\ref{def:state preparation unitary}, we have
    \begin{equation}
        \sqrt{\sum_{r\in [R]} \abs*{\sqrt{v_r}-\sqrt{\frac{a_r}{\alpha}}}^2}\leq \epsilon.
        \label{eq:appendix B 1}
    \end{equation}
    Note that
    \begin{equation}
        \sum_{r\in [R]}\abs*{v_r-\frac{a_r}{\alpha}}\leq\sum_{r\in [R]} \abs*{\sqrt{v_r}-\sqrt{\frac{a_r}{\alpha}}}\cdot\abs*{\sqrt{v_r}+\sqrt{\frac{a_r}{\alpha}}}
        \leq 2\sum_{r\in [R]}\abs*{\sqrt{v_r}-\sqrt{\frac{a_r}{\alpha}}}.
        \label{eq:appendix B 2}
    \end{equation}
    By the well-known inequality between $l_1$-norm and $l_2$-norm:
    $\norm{\bx}_1\leq \sqrt{R}\cdot\norm{\bx}_2$ for $\bx\in \Co^R$,
    the RHS of \eqref{eq:appendix B 2} is upper bounded by
    \begin{equation*}
        2\sqrt{R}\cdot\sqrt{\sum_{r\in [R]}\abs*{\sqrt{v_r}-\sqrt{\frac{a_r}{\alpha}}}^2}\leq 2\sqrt{R}\epsilon\leq R\epsilon,
    \end{equation*}
    where the first inequality comes from \eqref{eq:appendix B 1},
    and the second inequality uses the assumption $R\geq 4$.
    Thus we have $\sum_{r\in [R]}\abs*{v_r-\frac{a_r}{\alpha}}\leq R\epsilon$.
    Now take $b=s$, $V_C=V^\dagger$ and $V_D=V$ in Definition~\ref{def:state preparation pair}, it immediately follows that
    $\sum_{r\in [R]}\abs*{\alpha \parens*{c_r^*d_r}-a_r}=\alpha \sum_{r\in [R]}\abs*{v_r-\frac{a_r}{\alpha}}\leq \alpha R \epsilon$,
    which implies that $\parens*{V^\dagger, V}$ is an $\parens{\alpha, s, \alpha R \epsilon}$-state-preparation-pair of $\ba$.
\end{proof}

We note that in the above lemma, the upper bound $\alpha R \epsilon$ on the precision of the state preparation pair is very loose,
which is for simplicity of the proofs in Section~\ref{sec:linear_combinations_of_hamiltonian_powers} with the final results unchanged.

\end{document}